\documentclass[a4paper]{article}

\usepackage{tikzit}

\definecolor{mp}{RGB}{240,240,240}

\definecolor{pinkish}{RGB}{180,140,240}

\tikzstyle{node}=[fill=mp, text centered, text depth = 1 em, text height = 1.4em, draw=black, shape=rectangle, tikzit shape=rectangle,
rounded corners=2ex,  align=center]
\tikzstyle{circle}=[fill=mp, draw=black, shape=circle]
\tikzstyle{her}=[fill=pinkish, text centered, text depth = 1 em, text height = 1.4em, draw=black, shape=rectangle, tikzit shape=rectangle,
rounded corners=2ex,  align=center]

\tikzstyle{implies}=[->, thick]
\tikzstyle{implies-mt}=[->, thick, draw=pinkish]
\tikzstyle{himplies}=[->, thick, dash pattern={on 7pt off 1pt on 1pt off 1pt}]
\tikzstyle{trivial}=[->,thick, densely dotted ]
\tikzstyle{if2}=[<->, thick]

\usepackage[english]{babel}
\usepackage{hyperref}
\usepackage{todonotes}
\usepackage{amsmath, amsfonts,xfrac, amsthm,amssymb,fancyhdr}

\usepackage{mathtools}

\usepackage{setspace}

\usepackage{geometry}
\usepackage{xypic,colortbl}

\usepackage{ifthen, sidecap, wrapfig, xspace}

\usepackage{mathrsfs, stmaryrd}

\usepackage{sectsty}

\usepackage{
 calc, enumitem, makeidx,bbold
}

\usepackage[font=small,format=plain,labelfont=sc,textfont=it]{caption}
\usepackage[T1]{fontenc}
\usepackage[utf8]{inputenc}
\usepackage[all]{xy}
\usepackage[geometry]{ifsym}

\usepackage{palatino}
\usepackage{mathpazo}
\usepackage{comment}

\newcommand{\Oof}{\mathcal{O}}
\renewcommand{\preceq}{\preccurlyeq}
\newcommand{\lowerbound}{\lfloor \frac n 3\rfloor!}

\newcommand{\ind}[1][]{%
  \mathrel{
    \mathop{
      \vcenter{
        \hbox{\oalign{\noalign{\kern-.3ex}\hfil$\vert$\hfil\cr
              \noalign{\kern-1.5ex}
              $\vert$\cr\noalign{\kern-.3ex}}}
      }
    }\displaylimits_{#1}
  }
}

\newcommand{\nind}[1][]{%
  \mathrel{
    \mathop{
      \vcenter{
        \hbox{\oalign{\noalign{\kern-.3ex}\hfil$\nmid$\hfil\cr
              \noalign{\kern-1.5ex}
              $\vert$\cr\noalign{\kern-.3ex}              
              }}
      }
    }\displaylimits_{#1}
  }
}

\newcommand{\St}[1][]{\mathrm{S}^{#1}}
\newcommand{\Types}[1][]{\mathrm{Types}^{#1}}

\newcommand{\from}{\colon}
\newcommand{\str}[1]{\mathrm{#1}}
\renewcommand{\cal}[1]{\mathcal {#1}}
\newcommand{\CC}{\mathscr C}
\newcommand{\DD}{\mathscr D}
\renewcommand{\le}{\leqslant}
\renewcommand{\ge}{\geqslant}
\renewcommand{\phi}{\varphi}

\newcommand{\mathsym}[1]{{}}

{\begin{figure}[#1]
\centering 
\includegraphics[scale=#3]{#2}
}
{\end{figure}}

\newlist{enumeratep}{enumerate}{10}
\setlist[enumeratep]{label=\quad\textit{\arabic*'.},ref=\arabic*',leftmargin=*}

\newenvironment{romanlist'}[0]
{\begin{list}{\makebox[0.5cm][l]{\textit{\roman{enumi}')}}}{\usecounter{enumi}}}
{\end{list}}

\newcommand{\savelabel}[2]{\expandafter\newtoks\csname#1\endcsname
  \global\csname#1\endcsname={#2} \label{#1} #2}
\newcommand{\loadlabel}[1]{\noindent {\bf Lemma~\ref{#1}. } \textit{\the\csname#1\endcsname}
\medskip

}
\renewcommand{\setminus}{-}

\newcommand{\loadlabelthm}[1]{\medskip\noindent {\bf Theorem~\ref{#1}. }
  \noindent  \textit{\the\csname#1\endcsname}
\medskip
}

\newcommand{\loadlabelprop}[1]{\noindent {\bf Proposition~\ref{#1}. }
  \noindent  \textit{\the\csname#1\endcsname}
\medskip

}

\newcommand{\R}{\mathbb{R}}

\newcommand{\Z}{\mathbb{Z}}
\newcommand{\N}{\mathbb{N}}

\renewcommand{\subset}{\subseteq}

\newcommand{\atleast}[1]{{\ge n}}
\newcommand{\less}[1]{{<n}}

	\newcommand{\notacol}[2]{}

\newsavebox{\quoteitbox}

{\hspace*{\fill}{\upshape(\usebox{\quoteitbox})}\end{quote}%
\end{sloppypar}\noindent\ignorespacesafterend}

\newenvironment{quoteit*} 
{\begin{sloppypar}\noindent\slshape\begin{quote}\itshape} 
	{\end{quote}\ignorespaces\end{sloppypar}\noindent\ignorespacesafterend}

\newenvironment{quotetag*}
{~\par
	\begingroup                  
	\begin{equation*}
		 \begin{minipage}[c]{115mm}
			\it\noindent{\par}
}
{
		\end{minipage}
	\end{equation*}
	\endgroup                        
\par
\textnormal
\medskip
}

\newcommand{\Ll}{{\mathcal L}}

\newcommand{\Pp}{{\mathcal P}}

\newcommand{\Rr}{{\mathcal R}}

\newcommand\set[1]{\ensuremath{\{#1\}}}
\newcommand{\setof}[2]{\set{#1\mid#2}}

\DeclareMathOperator{\tp}{tp}
\DeclareMathOperator{\atp}{atp}
\newcommand{\tup}[1]{\bar{#1}}

\newtheoremstyle{theoremstyle}
  {3pt}
  {3pt}
  {\itshape}
  {0pt}
  {\bfseries}
  {.}
  {4pt}
  {}

\DeclareMathOperator{\Av}{Av}

\theoremstyle{theoremstyle}
\newtheorem{theorem}{Theorem}[section]

\newtheorem*{theorem*}{Theorem}

\newtheorem{lemma}[theorem]{Lemma}
\newtheorem*{lemma*}{Lemma}
\newtheorem{corollary}[theorem]{Corollary}
\newtheorem{proposition}[theorem]{Proposition}
\newtheorem*{proposition*}{Proposition}
\newtheorem{claim}{Claim}[section]
\newtheorem*{claim*}{Claim}

\newtheorem{fact}{Fact}

\newtheoremstyle{remarkstyle}
  {3pt}
  {10pt}
  {}
  {0pt}
  {\itshape}
  {}
  {4pt}
  {\thmname{#1}\thmnumber{ #2}\thmnote{ (#3)}.}

\theoremstyle{remarkstyle}
\newtheorem{example}{Example}[section]

\newtheorem{remark}{Remark}[section]

\newtheoremstyle{definitionstyle}
  {3pt}
  {3pt}
  {}
  {0pt}
  {\bfseries}
  {}
  {4pt}
  {\thmname{#1}\thmnumber{ #2}\thmnote{ (#3)}.}

\theoremstyle{definitionstyle}
\newtheorem{definition}{Definition}

\numberwithin{equation}{section}

\newlength{\wideaslength}

\renewcommand{\subset}{\subseteq}

\newcommand{\seta}[1]{}

\def\lsim{\mathrel{\rlap{\lower4pt\hbox{\hskip1pt$\sim$}}
    \raise1pt\hbox{$<$}}}                

\definecolor{gray1}{rgb}{0.99,0.99,0.99}
\definecolor{gray2}{rgb}{0.97,0.97,0.97}
\definecolor{gray3}{rgb}{0.95,0.95,0.95}
\definecolor{gray4}{rgb}{0.93,0.93,0.93}
\definecolor{gray5}{rgb}{0.91,0.91,0.91}
\definecolor{gray6}{rgb}{0.89,0.89,0.89}
\definecolor{gray7}{rgb}{0.87,0.87,0.87}
\definecolor{gray8}{rgb}{0.85,0.85,0.85}
\definecolor{gray9}{rgb}{0.83,0.83,0.83}
\definecolor{gray10}{rgb}{0.81,0.81,0.81}
\definecolor{gray20}{rgb}{0.71,0.71,0.71}
\definecolor{gray40}{rgb}{0.51,0.51,0.51}

\DeclareMathOperator{\tww}{tww}

\title{Ordered graphs of bounded twin-width}
\author{Pierre Simon\quad and\quad Szymon Toruńczyk}

\renewcommand{\todo}[1]{}
\begin{document}
\maketitle
\begin{abstract}
    We consider hereditary classes of graphs equipped with a total order. We provide multiple equivalent characterisations of those classes which have bounded twin-width. In particular, we prove that those are exactly the classes which avoid certain large grid-like structures and induced substructures. From this we derive that
    the model-checking problem for first-order logic is fixed-parameter tractable over a hereditary class of ordered graphs if, and -- under common complexity-theoretic assumptions -- only if the class has bounded twin-width. We also show that bounded twin-width is equivalent to the NIP property from model theory, as well as the smallness condition from enumerative combinatorics. We prove the existence of a gap in the growth of hereditary classes of ordered graphs.
    Furthermore, we prove a grid theorem which applies to all monadically NIP classes of structures (ordered or unordered), or equivalently, classes which do not transduce the class of all finite graphs.
     \end{abstract}

\section{Introduction}\label{sec:intro}

The recently introduced notion of twin-width~\cite{tww1,tww2,tww3}
is a graph width parameter with remarkable properties.
It measures how well a given graph can be recursively decomposed into parts which have  simple interactions with each other (see Sec.~\ref{sec:prelims} for a definition).
%
The notion generalizes to arbitrary relational structures equipped with unary and binary relations.

Many well-studied classes have bounded twin-width: the class of planar graphs, and more generally, any class of graphs excluding a fixed minor; the class of cographs and more generally, any class of bounded cliquewidth; posets of bounded width; and  classes of permutations (viewed as sets with two total orders) omitting a fixed permutation as an induced substructure. 
Moreover, classes of bounded twin-width enjoy good  properties of combinatorial, algorithmic, and logical nature.
For instance, classes of bounded twin-width are closed under first-order transductions,  are small
(contain $n!\cdot 2^{\Oof(n)}$ distinct labelled graphs on $n$ vertices), and are $\chi$-bounded (the chromatic number can be bounded in terms of the clique number) \cite{tww2}.
Furthermore, it is shown that model-checking first-order logic is fixed-parameter tractable on classes of bounded twin-width, assuming the input graph $G$ is given \emph{together with a certificate} of having bounded twin-width (a certain sequence of operations). 
More precisely, given  a first-order sentence $\phi$, a graph $G$, and a certificate that $G$ has twin-width at most $d$, there is an algorithm which determines whether $G$ satisfies $\phi$ in time $f(\phi,d)\cdot |V(G)|^c$ for some computable function $f\from\N\times\N\to\N$ and fixed constant $c\in\N$.

For each of the classes $\CC$ mentioned above there is an algorithm which, given a graph $G\in \CC$, computes 
some certificate that $G$ has  twin-width bounded by a constant, in polynomial time~\cite{tww1}.
Hence, for each of those classes $\CC$, model-checking first-order logic is fixed-parameter tractable, generalizing many previous results.

The appropriate certificate of having bounded twin-width is usually obtained from a suitable ordering of the vertices.
The adjacency matrix of the graph with respect to this order should be simple in a certain sense. The existence of such an order is essential in proving results about classes of bounded twin-width, as well as obtaining efficient algorithms.
This suggests that \emph{ordered graphs} of bounded twin-width are the more fundamental object.

\subsection*{Main result}\label{sec:newmain}
We solve a number of problems which are open for graphs  of bounded twin-width, 
in the case of ordered graphs of bounded twin-width. Among other things, we show that if a class $\CC$ of ordered graphs has bounded twin-width, then for each $G\in\CC$, a certificate that $G$ has twin-width bounded by a constant can be computed in polynomial time.
Consequently, model-checking is fixed-parameter tractable on $\CC$. We also prove that the converse holds, under common complexity-theoretic assumptions.

 More importantly, we give multiple characterisations of classes 
of ordered graphs of bounded twin-width, connecting notions from various areas of mathematics and theoretical computer science, and solving several open problems on the way.

The most tangible characterisation is in terms of certain forbidden substructures, dubbed \emph{semigrids}.
Those are two-dimensional variations of  half-graphs, matchings, and complements of matchings which are depicted in Fig.~\ref{fig:halfgraphs}.
Say that 
two sets $X,Y$ of vertices of an ordered graph $G$ \emph{form an $R$-graph}, where $R\in \set{\le,\ge,=,\neq}$,
if the $i$th smallest element in $X$ is adjacent to the $j$th smallest element in $Y$ if and only if $i R j$, for all $1\le i,j\le |X|$, see Fig.~\ref{fig:halfgraphs}. 
\begin{figure}[h!]
	\centering
	\includegraphics[scale=0.45,page=3]{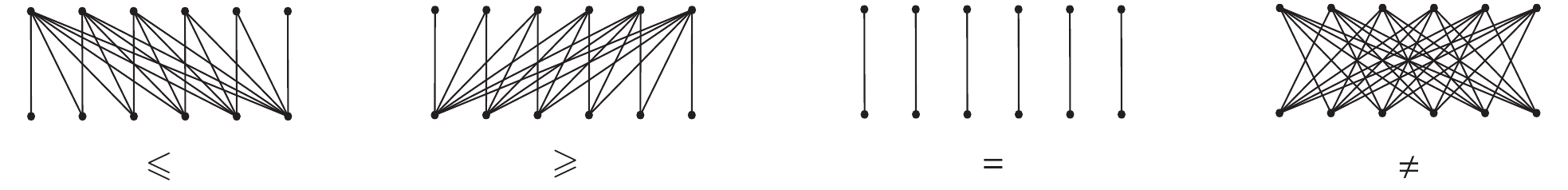}
	\caption{A $\le$-graph, a $\ge$-graph, a $=$-graph, and a $\neq$-graph formed by sets $X$ (six vertices on the top) and $Y$ (six vertices on the bottom), where $X$ and $Y$  are both ordered from left to right. Note that the edges within $X$ and within $Y$ may be arbitrary.
	}\label{fig:halfgraphs}
\end{figure}

An $m\times n$ \emph{semigrid of type $R\in\set{\le,\ge,=,\neq}$} is an ordered graph 
  $G=(V,E,\le)$ whose domain $V$ can be 
  partitioned 
into $m+1$ disjoint intervals (with respect to the order $\le$ on $V$) $I_0,I_1,\ldots,I_m$ with  $n+1$ elements each and $I_1<\ldots<I_m$, such that $I_0$ and $I_i$ form an $R$-graph, for $i=1,\ldots,m$
(cf. Fig.~\ref{fig:semigrids}).
 \begin{figure}[h!]
	\centering
	\includegraphics[scale=0.5,page=6]{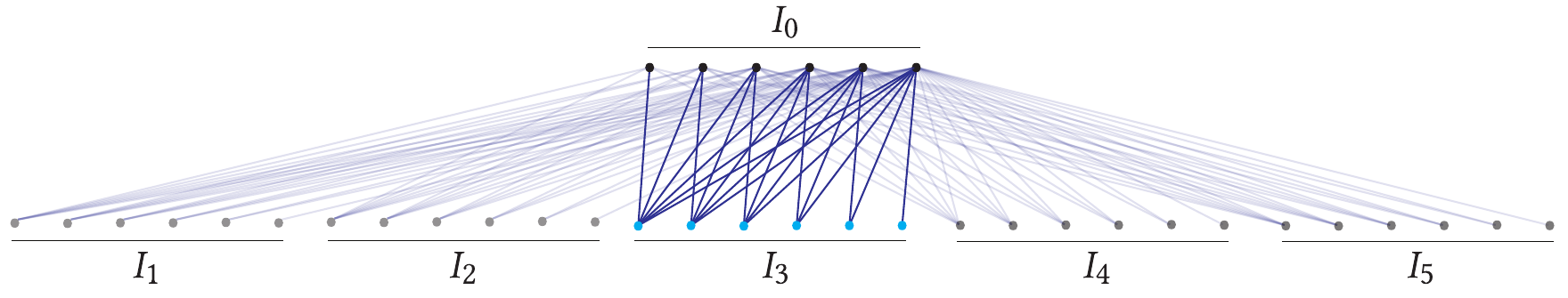}
	\caption{A $5\times 5$-semigrid 
	of type $\ge$. The vertices are ordered left to right in each interval.
	}\label{fig:semigrids}
\end{figure}
Note that  there are many $m\times n$-semigrids of a given type,
as the above specification is not complete:
 the adjacencies within $I_1\cup\cdots\cup I_m$, as well as the adjacencies within $I_0$, and also the relative order between the interval $I_0$ and  the intervals $I_1,\ldots,I_m$, are all left unspecified.
 However, using a Ramsey argument one can prove that  a large \emph{regular} $n\times n$-semigrid can be found as an induced substructure of a sufficiently large semigrid $n'\times n'$ (see Section~\ref{sec:bad} for a definition). 
 In particular, a regular $m\times n$-semigrid is uniquely specified by the dimensions $m,n\ge 1$ and one of $256$ \emph{schemes}.\todo{example}

 Another characterisation is in terms of (simple first-order) \emph{interpretations}.
 Interpretations are a means of producing new structures out of old ones, using formulas. The new structure has the same domain as the old one (or its subset defined by a formula $\delta(x)$)
while each of its relations is defined by a formula $\phi(x,y)$ interpreted in the old structure.  For example, there is an interpretation which transforms a given graph $G$ into its edge-complement (using the formula $\neg E(x,y)$), and an interpretation which transforms $G$ into its square (using the formula $\exists z.E(x,z)\land E(z,y)$).
Transductions are a similar notion, but additionally allow to arbitrarily color the old structure before applying the interpretation and then use the colors in the formulas. 
Say that $\CC$ \emph{interprets}
the class of all graphs if 
there is an interpretation $I$ such that every  (finite) graph $G$ can be obtained as the result of $I$ applied to some structure in $\CC$. Replacing interpretations with transductions, we say that $\CC$ \emph{transduces} the class of all graphs.
It is known that no class of bounded twin-width transduces all graphs.




\medskip




 We may now state our main result, concerning classes of ordered graphs which are \emph{hereditary}, that is, closed under taking induced substructures.
 Among others, it provides a dichotomy result for all such classes: either they have bounded twin-width, and are therefore very well-behaved, or otherwise, they contain arbitrarily large regular $n\times n$-semigrids,
 and are then untamable.

\begin{theorem}\label{thm:intro}
	The following conditions are equivalent for a hereditary class $\CC$ of finite, ordered graphs:
	\begin{enumerate}
		\item\label{iit:tww} $\CC$ has bounded twin-width,
		\item\label{iit:semigrids} $\CC$ does not contain arbitrarily large  regular $n\times n$-semigrids,				
		\item\label{iit:inter} $\CC$ does not interpret the class of all graphs,
		 \item\label{iit:trans} $\CC$ does not transduce the class of all graphs,
		\item\label{iit:fpt} model-checking first-order logic is fixed-parameter tractable on $\CC$ (assuming $\mathrm{FPT\neq AW[\ast]}$),
		\item\label{iit:small} $\CC$ contains at most  $2^{\Oof(n)}$   structures with $n$ elements, up to isomorphism,
		\item\label{iit:not-big}$\CC$ contains fewer than   $\lowerbound$ structures with $n$ elements, up to isomorphism,
		\item \label{iit:simple}	there are $k,t\in\N$ such that for all $G\in\CC$, 
		if the  adjacency matrix of $G$ 
		is cut into $t^2$ zones using $t-1$ vertical and $t-1$ horizontal lines, then there is  a zone with no more than $k$ non-identical rows and no more than $k$ non-identical columns.
	\end{enumerate}
\end{theorem}

The above result connects notions from logic, enumerative combinatorics, parameterized complexity, graph theory and matrix theory.
 However, at the core of our approach are 
 tools and ideas originating from \emph{model theory}.
 
 As our second main result, we provide further characterisations of bounded twin-width classes in terms of notions which originate from model theory, but which also transpire in algorithmic and structural graph theory.
We prove that  generalizations of the conditions~\eqref{iit:semigrids},~\eqref{iit:inter},~and \eqref{iit:simple} hold for arbitrary \emph{monadically NIP classes} $\CC$ of relational structures.
Those can be equivalently characterised as classes (of finite or infinite structures) which do not transduce the class of all finite graphs.
They include all graph classes of bounded twin-width (with or without an order), but also all transductions of nowhere-dense classes (see below).\todo{and much more}

The following theorem is a vast generalization of 
some of the key implications in Theorem~\ref{thm:intro}.
	
\begin{theorem}\label{thm:summary-mt0}
	For any class of structures $\CC$, consider the following statements:
	\begin{enumerate}\setcounter{enumi}8
		\item\label{qit:trans} $\CC$ does not transduce the class of all graphs,
		\item\label{qit:mNIP} $\CC$ is monadically NIP,
		\item\label{qit:grids} $\CC$ does not {define large grids} (cf. Def.~\ref{def:grids}),
		\item\label{qit:1-dim} $\CC$ is $1$-dimensional (cf. Def.~\ref{def:1-dim}),
		\item\label{qit:regular} $\CC$ is a {regular} class (cf. Def.~\ref{def:regular}).
	\end{enumerate}
    Then the implications $\eqref{qit:trans}\leftrightarrow\eqref{qit:mNIP}\rightarrow\eqref{qit:grids}\rightarrow\eqref{qit:1-dim}\rightarrow\eqref{qit:regular}$ hold.
	For classes of binary, ordered structures, the above conditions are all equivalent to:
	\begin{enumerate}[resume]		
		\item the class of all finite induced substructures of structures in $\CC$ has bounded twin-width,
		\item the class of all finite induced substructures of structures in $\CC$ is NIP.
	\end{enumerate}
\end{theorem}
The notion~\eqref{qit:grids} of defining large grids generalizes the notion~\eqref{iit:semigrids} of containing large semigrids 
to arbitrary  structures, while the notion~\eqref{qit:regular} of regularity 
generalizes  condition~\eqref{iit:simple} from Theorem~\ref{thm:intro}. In particular, those notions do not require the structures to be ordered, finite, or binary.
The notion of 1-dimensionality has a somewhat geometric flavor.
It is defined in terms of a variant of forking independence -- a central concept in stability theory, generalizing e.g. independence in vector spaces or algebraic independence.\todo{grids not 1-dim}

\medskip

Theorem~\ref{thm:summary-mt0} provides a key ingredient in our proof of Theorem~\ref{thm:intro} -- a \emph{grid theorem}.
More importantly,
we believe that it may be of independent interest, and possibly of broader applicability than 
just in the context of ordered structures. For example, by Theorem~\ref{thm:summary-mt0}, all graph classes of bounded twin-width (without an order) and all interpretations of \emph{nowhere-dense} classes~\cite{nevsetvril2011nowhere} are regular.

The implications~\eqref{qit:grids}$\rightarrow$\eqref{qit:1-dim}$\rightarrow$\eqref{qit:regular} are proved using model-theoretic tools. They yield a general grid theorem for classes that are not regular in the sense of \eqref{qit:regular}.

\medskip
Very roughly, the proof of Theorem~\ref{thm:intro} can be summarized as follows.
The goal is to show that if $\CC$ is a class of ordered graphs which has unbounded twin-width, then $\CC$ contains arbitrarily large regular semigrids. We do it in the following steps:
\begin{enumerate}
    \item  If a class $\CC$ of ordered graphs satisfies  condition~\eqref{iit:simple} then it has bounded twin-width. 
	This generalizes a result from~\cite{tww1}, and uses similar methods. 
    \item Suppose $\CC$ has unbounded twin-width; then it fails condition~\eqref{iit:simple} and in particular, is not regular. By Theorem~\ref{thm:summary-mt0}, $\CC$ defines large grids.
    \item Large grids are then made more structured using a Ramsey result in order to exhibit arbitrarily large regular semigrids in $\CC$.
\end{enumerate}
Once $\CC$ is known to have large regular semigrids,
 conditions~\eqref{iit:inter}-\eqref{iit:not-big} in Theorem~\ref{thm:intro} are easily shown to fail. 
 
Some parts of our proof are ineffective,
in particular, we use the compactness theorem for first-order logic. Quite remarkably, we are still able to derive effective bounds and algorithms. For instance, 
we get a polynomial-time algorithm approximating the twin-width of a given ordered graph $G$:
\begin{theorem}\label{thm:approx}    
    There is a computable function $f\from\N\to\N$ and an  algorithm which inputs an ordered graph $G$, and computes in  polynomial time  a certificate that $G$ has twin-width at most  $f(\textrm{twin-width}(G))$.
\end{theorem}
 
 This, combined with the result of~\cite{bonnet2020twin} yields fixed-parameter tractability of the model-checking problem for classes of ordered graphs of bounded twin-width, proving the implication \eqref{iit:tww}$\rightarrow$\eqref{iit:fpt}. We also prove the converse implication, under the common complexity-theoretic assumption FPT$\neq$AW[$\ast$].

Theorem~\ref{thm:intro} is proved in the greater generality of ordered structures over an arbitrary binary signature.

 Our results generalize several known results, including the celebrated Stanley-Wilf conjecture about permutations avoiding a fixed permutation pattern. Moreover, we solve several open questions. This is discussed below.
\subsection*{Related work}

\paragraph{Enumerative combinatorics.}
Enumerative combinatorics is involved in counting the number of labelled 
or unlabelled structures in a given class of structures.
Marcus and Tardos proved~\cite{MARCUS2004153} the Stanley-Wilf conjecture,  that every class of permutations (viewed as sets equipped with two linear orders) avoiding a fixed permutation as a substructure  has growth $2^{\Oof(n)}$, that is, contains at most $2^{\Oof(n)}$ unlabelled structures with $n$ elements (up to isomorphism). In particular, a hereditary class of permutations either has growth $2^{\Oof(n)}$ or has factorial growth $n!$.
The Marcus-Tardos theorem is fundamental in  twin-width theory,
and underlies many of the arguments used in~\cite{tww1,tww2,tww3}, and also in the current paper. It is also an immediate consequence of the equivalence~\eqref{iit:small}$\leftrightarrow$\eqref{iit:semigrids} in Theorem~\ref{thm:intro},
as every permutation is a substructure of any sufficiently large regular semigrid (defined suitably for structures with two linear orders).

Our result proves the existence of a gap in the growth of hereditary classes of ordered graphs: either the class has exponential growth $2^{\Oof(n)}$ or, or otherwise it has at least factorial growth  $\lowerbound$. This generalizes previous results and answers a problem posed by Bolagh, Bollob\'as and Morris~\cite[Sec. 8]{balogh2006hereditary}.

In~\cite{tww2} it is shown that every class of bounded twin-width is small, 
that is, contains at most $n!2^{\Oof(n)}$ \emph{labelled} structures with $n$ elements (that is, with domain $\set{1,\ldots,n$}).
It has been also conjectured that the converse  holds for every hereditary class of graphs. Our equivalence~\eqref{iit:tww}$\leftrightarrow$\eqref{iit:small} in Theorem~\ref{thm:intro} confirms this conjecture in the case of ordered graphs.
Note that a class of ordered graphs is small if and only if it contains at most $2^{\Oof(n)}$ 
unlabelled structures of size $n$ (that is, up to isomorphism)
as any ordered structure yields exactly $n!$ distinct labelled structures.

\paragraph{Transductions and interpretations.} 
The study of transductions in theoretical computer science 
originates from the study of word-like and tree-like structures,
such as graphs of bounded treewidth~\cite{10.1007/3-540-19488-6_105} or graphs of bounded cliquewidth~\cite{COURCELLE199453}.

By results of~\cite{tww1}, classes of bounded twin-width are closed under transductions; in particular, no class of bounded twin-width transduces (nor interprets) the class of all graphs. 

The equivalence \eqref{iit:tww}$\leftrightarrow$\eqref{iit:trans}, characterising hereditary classes of ordered graphs of bounded twin-width 
as precisely those  which do not transduce the class of all graphs, 
is not unlike a result~\cite{COURCELLE200791}
characterizing classes of bounded cliquewidth as precisely those which do not transduce the class of all graphs via some transduction
of counting monadic second-order logic (CMSO,
an extension of first-order logic).


  Note that in our result we require the graphs to be ordered for the implication~\eqref{iit:trans}$\rightarrow$\eqref{iit:tww} to hold:
 the class $\CC$ of graphs of maximum degree $3$ has unbounded twin-width~\cite{tww2}, but  does not transduce the class of all graphs as it is nowhere dense (see below).

 Both the above result for cliquewidth and our results are obtained using appropriate grid theorems.

\paragraph{Grid theorems.}
Grid theorems are dichotomy results in structural graph  theory which state that 
either a structure has a small width with respect to a considered width parameter, or otherwise, a grid-like obstruction can be found in the structure.
For example, this applies to the treewidth parameter and planar grids occurring as minors~\cite{ROBERTSON198692}.  It also applies to cliquewidth and planar grids being definable in CMSO~\cite{COURCELLE200791}.
 As CMSO formulas can  define 
the relations of being in the same row/column in a planar grid, this allows to encode any graph in a subgraph of a sufficiently large grid, implying the result mentioned earlier about classes of bounded cliquewidth.

Our main result also proves an appropriate grid theorem for classes of ordered graphs of bounded twin-width,
as made precise by the notion of regular semigrids.
For such semigrids there is a first-order formula which defines the relation 
of being in the same interval $I_1,\ldots,I_n$ (`same row'), and the relation of having the same distance from the beginning of an interval (`same column').
From this it follows that if a hereditary class has unbounded twin-width then it interprets the class of all graphs.

There are other known grid theorems. The Marcus-Tardos result itself proves a result of this kind, stating that if a square $0-1$ matrix has sufficiently many $1$'s then it must contain a large grid-like pattern formed by $1$'s.
Finally,~\cite{tww1} use the Marcus-Tardos result to exhibit a grid-like structure in classes of unbounded twin-width
(as a subdivision of the adjacency matrix).
However, the grids obtained this way are not sufficiently structured to allow to define the `same row' and `same column' relations.
The large semigrids that we exhibit provide the ultimate grid theorem for classes of unbounded twin-width.

\paragraph{Monadic NIP.}
\todo{rework}
Model theory classifies typically infinite structures according to the combinatorial complexity of families of definable sets. This is usually done through the introduction of tameness properties. The most important such notion is that of stability. A structure is stable if no formula $\phi(x, y)$ encodes arbitrary large half-graphs ($\le$-graphs as in Fig.~\ref{fig:halfgraphs}),
which roughly means that there is no definable order on large subsets of the structure.
Stability captures the tameness properties of families of algebraic sets.
\todo{simple interpretations}
 A related, weaker, notion is that of NIP: a structure is NIP if every definable family of sets has finite VC-dimension. This captures the tameness properties of families of sets arising from geometric settings (for instance families of semi-algebraic sets of bounded complexity have finite VC-dimension).
 
  The notion of monadically NIP is a much stronger requirement which says that the structure is NIP even if every subset of the domain can be used as a unary predicate.
 
 This notion is closely related to notions which are studied in theoretical computer science and structural graph theory.

As mentioned, a class of structures $\CC$ is monadically NIP if and only if it does not transduce the class of all graphs.
Hence, Theorem~\ref{thm:intro} proves that a class  of ordered graphs is monadically NIP if, and only if it has bounded twin-width.

Examples of monadically NIP graph classes include all {nowhere dense} classes.
A class of graphs  $\CC$ is \emph{nowhere dense}
if for all $r\in\N$ there is some $t\in\N$ such that the $r$-subdivision of the $t$-clique is not a subgraph of any graph in~$\CC$.
Examples include the class of graphs with maximum degree bounded by a constant (those classes have unbounded twin-width~\cite{tww2}),
as well as every proper minor-closed graph class (here the twin-width is bounded).
A subgraph-closed class of graphs  is nowhere dense if and only if it is monadically NIP~\cite{adler2014interpreting}.

Monadically NIP classes are closed under transductions, so any transduction of a nowhere dense graph class is also monadically NIP,
but not necessarily nowhere dense. 
 
Theorem~\ref{thm:summary-mt0} provides a grid theorem for any monadically NIP  class of structures -- not necessarily ordered, binary, nor finite.
\todo{about stability theory}

\paragraph{Tractability of model-checking.}
Testing if a given sentence $\phi$ of first-order logic holds in a given graph $G$ 
takes time $\Oof(|V(G)|^{|\phi|})$ using the 
naive algorithm, and it is conjectured that the exponential dependency on $|\phi|$ cannot be avoided. More precisely,
it is conjectured that model-checking first-order logic is not fixed-parameter tractable on the class of all graphs, which is equivalent to the conjecture FPT$\neq$AW[$\ast$] from parameterized complexity theory~\cite{10.5555/1121738}.

There are several known classes $\CC$ of structures  for which model-checking first-order logic is fixed-parameter tractable.
To the best of our knowledge, all known tractable hereditary  classes
are monadically NIP\footnote{Tractable classes that are not hereditary include for example the class of all finite abelian groups~\cite{bova2015first}}.

\medskip
Altogether, the following picture emerges, assuming FPT$\neq$AW[$\ast$]:
\begin{enumerate}
	 \item if $\CC$ is a subgraph-closed class of graphs  then $\CC$ is  monadically NIP if, and only if model-checking first-order logic is fixed-parameter tractable over $\CC$~\cite{grohe2014deciding},
	 \item if $\CC$ is a hereditary class of ordered graphs 
	 then $\CC$ is  monadically NIP if, and only if model-checking first-order logic is fixed-parameter tractable over $\CC$.
\end{enumerate}
The last item is by the equivalence~\eqref{iit:tww}$\leftrightarrow$\eqref{iit:fpt} in our main result, Theorem~\ref{thm:intro}, both implications being new.

\medskip
There are a few hereditary graph classes there are known to be tractable for first-order logic, but are not covered by the items above.
Those include:
\begin{itemize}
	\item  classes of structures of bounded twin-width for which there is an effectively computable certificate of having bounded twin-width (this includes all the examples from the first paragraph),
	\item map graphs~\cite{10.1007/978-3-662-55751-8_17},
	\item first-order transductions of graph classes with bounded maximum degree~\cite{10.1145/3383206}.
\end{itemize}

A conjecture~\cite[Conj. 8.2]{10.1145/3383206} implies that if a class  $\CC$ does not transduce all graphs, equivalently, $\CC$ is monadically NIP, then 
model-checking first-order logic is fixed-parameter tractable on $\CC$
(the conjecture there is actually even stronger). This would generalize all the above.

Theorem~\ref{thm:intro} confirms this conjecture in the case of ordered binary structures. 
Theorem~\ref{thm:summary-mt0} provides a tool for approaching the conjecture for all monadically NIP classes, by providing a grid theorem for such classes.


\paragraph{Independent work.}
A few days before this paper was submitted,  Bonnet, Giocanti,  Ossona de Mendez, and Thomass\'e
reported~\cite{bonnet2021twinwidth}  essentially the same result as our Theorem~\ref{thm:intro} and Theorem~\ref{thm:approx}.
Our results were obtained independently, and most likely using different methods (we were unable to verify this in the short time before submission).
The key differences between our results and the results of~\cite{bonnet2021twinwidth} are summarized below:
\begin{itemize}
	\item Part of our proof is not effective,
	as it uses tools from model theory, including the compactness theorem for first-order logic.
	In particular, we do not obtain an explicit upper bound on the twin-width of a graph avoiding a fixed regular $n\times n$-semigrid. We do, however, obtain a computable bound.	
	It should not be difficult to extract some elementary bound from our proof, but this  falls short of the explicit bounds obtained in~\cite{bonnet2021twinwidth}.

	\item The lower bound $\lowerbound$ in the growth of hereditary classes of unbounded twin-width~\eqref{iit:not-big} in Theorem~\ref{thm:intro} is improved to $n!$ in~\cite{bonnet2021twinwidth}.	
	Proving this requires substantially more effort than our  bound, which falls trivially out of the existence of large regular semigrids.
\item On the other hand,~\cite{bonnet2021twinwidth} does not provide 
any result analogous to our second main result, Theorem~\ref{thm:summary-mt0}, which  applies to all classes of structures -- possibly unordered, of arbitrary arity, and infinite.
\end{itemize}

\todo{organization}

\section{Preliminaries}\label{sec:prelims}

For $n\in\N$ denote $[n]=\set{0,\ldots,n-1}$. 

\paragraph{Order.}Order means total order.
A  subset $C$ of an ordered set $(X,\le)$ is 
\emph{convex}, or an \emph{interval}, if $a,b\in C$ and $a\le c\le b$ imply $c\in C$, for all $a,b,c\in X$.
A \emph{convex partition} of $X$ is a partition into convex subsets.

\paragraph{Structures.}
We consider the setting of relational structures.
Fix a relational signature $\Sigma$.

Let $\str S$ be a $\Sigma$-structure. We identify $\str S$ with its domain
when writing e.g. $a\in \str S$ or $A\subset \str S$. Tuples of elements of $\str S$ are denoted $\tup a\in\str S^k$ or $\tup a\in\str S^{\tup x}$, where $\tup x$ is a finite set of variables. In the latter case, the tuple $\tup a$ can be seen as a valuation of the variables $\tup x$.

If $\str S$ and $\str T$ are two $\Sigma$-structures then $\str S$ is an \emph{induced substructure} of $\str T$  if the domain of $\str S$ is contained in the domain of $\str T$ and for every relation symbol $R\in\Sigma$ of arity $k$ and $\tup a\in \str S^k$, 
$R(\tup a)$ holds in $\str S$ if and only if $R(\tup a)$ holds in $\str T$.

\paragraph{Formulas.}
We only consider first-order $\Sigma$-formulas in what follows.

If $\alpha$ is a formula and $\tup x$ is a finite set of variables then we may write $\alpha(\tup x)$ to indicate that the free variables of $\alpha$ are among $\tup x$.

For 
a formula $\alpha(\tup x)$ and tuple $\tup a\in\str S^{\tup x}$
write $\str S\models \alpha(\tup a)$ to denote that $\tup a$ satisfies $\alpha$ in~$\str S$.
Denote \[\alpha(\str S)=\setof{\tup a\in \str S^{\tup x}}{\str S\models\alpha(\tup a)}.\]

If $\alpha(\tup x)$ is a formula and $\tup x$ is partitioned as $\tup x=\tup u\cup \tup v$
then we may write $\alpha(\tup u;\tup v)$ to signify this partition, and that $\tup v$ will be treated as variables ranging over \emph{parameters}. For instance, if $\str S$ is a structure and $\tup b\in\str S^{\tup v}$ is a tuple then $\theta(\tup x;\tup b)$ denotes the \emph{formula with parameters}  obtained from $\theta$ by replacing the variables $\tup v$ by  $\tup b$, treated as constants.
Denote
\[\alpha(\str S;\tup b)=\setof{\tup a\in\str S^{\tup u}}{\str S\models\alpha(\tup a;\tup b)}.\]

\paragraph{Atomic types.}
An \emph{atomic formula} is a formula of the form $R(x_1,\ldots,x_k)$ where $R\in\Sigma$ is of arity~$k$. 
For a $k$-tuple $\tup a\in \str S^k$, 
its \emph{atomic type}, denoted $\atp(\tup a)$, 
 is conjunction of all  formulas $\alpha(x_1,\ldots,x_k)$
such that $\str S\models \alpha(\tup a)$
and $\alpha$ is either an atomic formula, or its negation.
Up to bijection, $\atp(\tup a)$ is uniquely determined by the isomorphism type of the substructure $\str S'$ of $\str S$ induced by $\set{a_1,\ldots,a_k}$, expanded with $k$ constants interpreted as $a_1,\ldots,a_k$.

An \emph{atomic type} with variables $\tup x$ is a conjunction which contains 
as a conjunct every atomic formula
$R(x_1,\ldots,x_k)$ for $R\in\Sigma$ and $x_1,\ldots,x_k\in \tup x$, or its negation $\neg R(x_1,\ldots,x_k)$. 

\paragraph{$\theta$-types over a set.}

	Let $\theta(\tup u;\tup v)$ be a formula and $\str S$ a structure. For 
	a tuple  $\tup a\in \str S^{\tup u}$ and a set $B\subset \str S$ of \emph{parameters} define the \emph{type of $\tup a$ over $B$} as:
	 \[\tp^\theta(\tup a/B)=\setof{\tup b\in B^{\tup v}}{\str S\models\theta(\tup a;\tup b)}.\]		
Equivalently -- up to bijection -- $\tp^\theta(\tup a/B)$ is the set of formulas $\theta(\tup u;\tup b)$ with parameters $\tup b$ from $B$ that are satisfied by $\tup a$ in $\str S$. 

For a set $A\subset \str S$, denote
\[\Types[\theta](A/B):=\setof{\tp^\theta(\tup a/B)}{\tup a\in A^{\tup u}}.\]
Note that if $\Types[\theta](A/B)\le k$ then
$\Types[\hat\theta](B/A)\le 2^k$,
where $\hat\theta(\tup v; \tup u)=\theta(\tup u;\tup v)$.
\medskip

If $\Sigma$ is a binary signature then denote
\[\Types[\Sigma](A/B):=\setof{\tp^R(a/B)}{\tup a\in A, R\in\Sigma}.\]
Similarly as above, if $\Types[\Sigma](A/B)\le k$ then $\Types[\Sigma](B/A)\le 2^{|\Sigma| k}$.
Moreover, if $\Types[R](A/B)\le k_R$ for all $R\in\Sigma$ then $\Types[\Sigma](A/B)\le \prod_{R\in \Sigma}k_R$.
\paragraph{Homogeneity.}
Let $\str S$ be a structure over a binary relational signature $\Sigma$ and let 
$X,Y\subset \str S$. 
 The pair $X$ and $Y$ is \emph{homogeneous} 
if \[|\Types[\Sigma](X/Y)|=|\Types[\Sigma](Y/X)|=1.\]


 \begin{example}
	 In a graph $G$, a pair $X,Y\subset G$ is homogeneous 
	 if and only if it is homogeneous in the sense considered in Section~\ref{sec:intro}.
	If $G$ is an ordered graph then a pair $X,Y\subset G$ is homogeneous if $X$ and $Y$ are homogeneous in $G$ as an unordered graph, and moreover, either all elements of $X$ are strictly smaller than all elements of $Y$, or vice-versa.
 \end{example}

\paragraph{Interpretations and transductions.}
Let $\Sigma$ and $\Gamma$ be relational signatures. 
An \emph{interpretation} $I\from\Sigma\to\Gamma$ consists of a $\Sigma$-formula $\delta(x)$ and for each symbol $R\in \Gamma$ of arity $k$, a $\Sigma$-formula $\phi_R(x_1,\ldots,x_k)$.	
Given a $\Sigma$-structure $\str S$, define $I(\str S)$ as the $\Gamma$-structure with domain $D=\delta(\str S)$ equipped with the relation $\phi_R(\str S)\cap D^k$ interpreted as $R$, for each $R\in\Gamma$ of arity $k$.
A  \emph{transduction} $T\from\Sigma\to\Gamma$ 
is a $\widehat\Sigma$-interpretation for some $\widehat\Sigma$ extending $\Sigma$ by unary predicates.
Given a $\Sigma$-structure $\str S$, define $T(\str S)$ as the set of $\Gamma$-structures $I(\widehat{{\str S}})$, where $\widehat{\str S}$ ranges over all $\widehat\Sigma$-structures expanding $\str S$ by arbitrarily interpreting the unary predicates in $\widehat\Sigma\setminus\Sigma$ in $\widehat{\str S}$.

Interpretations and transductions are closed under compositions: if 
\[\Sigma_1\stackrel{I_1}\longrightarrow\Sigma_2\stackrel{I_2}\longrightarrow\Sigma_3\] are two interpretations (resp. transductions) then there is an interpretation (resp. transduction)   $I_2\circ I_1\from \Sigma_1\to\Sigma_3$ such that $(I_2\circ I_1)(\str S)=I_2(I_1(\str S))$, for every $\Sigma$-structure $\str S$.

\begin{remark}
	What is defined above is a special case of a more general notion of interpretations considered in model theory. 
	The restricted interpretations that we consider here are sometimes called \emph{simple} interpretations, but we drop this qualifier in this paper. 
\end{remark}

\begin{definition}\label{def:transduction}
    A class of $\Gamma$-structures $\DD$ is an \emph{interpretation} (resp. \emph{transduction}) of a class of $\Sigma$-structures $\CC$ if there is an interpretation (resp. transduction) $T\from\Sigma\to\Gamma$ 
    such that for every structure $\str D\in\DD$ there is some $\str C\in\CC$ with $\str D\in T(\str C)$.In this case we also say that $\CC$ \emph{interprets} (resp. \emph{transduces}) $\DD$.
\end{definition}

\paragraph{Classes of structures.}
We consider classes $\CC$ of structures that are all over a common relational signature $\Sigma$. 
We assume that $\CC$ is closed under isomorphisms, that is, if $\str A$ and $\str B$ are isomorphic structures then either both belong to $\CC$ or both do not belong to $\CC$.
A class of structures is \emph{hereditary} if it is closed under taking induced substructures.

We say that $\CC$ is a class of \emph{binary structures} if $\Sigma$ is a relational signature with relation symbols of arity at most two.
We do not restrict only to classes of finite structures. However, graphs are always assumed to be finite (specifically, when we say that $\CC$ interprets the class of all graphs). By a class of \emph{ordered structures} we mean a class of structures over a signature including the symbol $\le$ which is interpreted as an order in each structure from the class.

\subsection*{Twin-width}

A partition $\Pp$ of the domain of a structure $\str S$ has \emph{red-degree}
at most $d$ (for $d\in\N$) if for every $X\in \Pp$ there are at most $d$ sets $Y\in\Pp$ other than $X$ such that $X$ and $Y$ are not homogeneous.
A \emph{contraction sequence} in a finite structure $\str S$
is a sequence of partitions $\Pp_1,\Pp_{2},\ldots, \Pp_n$ of $\str S$ such that:
	\begin{itemize}
		\item $\Pp_1$ has one part,	
		\item $\Pp_n$ is the partition into singletons, 		
		\item $\Pp_{i+1}$ is obtained from $\Pp_{i}$ by splitting one of the parts of $\Pp_i$ into two, for $i=1,\ldots,n-1$.		
	\end{itemize}

\begin{definition}[Twin-width]Fix $d\in\N$.
	A finite structure $\str S$ over a binary signature has \emph{twin-width} at most~$d$ if
	it has a contraction sequence consisting of partitions of red-degree at most~$d$.
\end{definition}

A class $\CC$ of finite, binary structures  has \emph{bounded twin-width} 
if there is some $d\in\N$ such that every $\str S\in\CC$ has twin-width at most $d$.

\begin{example}
	A class of graphs has twin-width $0$ if and only if it is a class of cographs, that is, graphs that can be obtained from one-vertex graphs by the operations of disjoint union and edge-complement. 
\end{example}

\begin{example}[\cite{bonnet2020twin}]
	The following classes have bounded twin-width:
\begin{itemize}
	\item every proper minor-closed class of finite graphs,
	\item every class of graphs of bounded clique-width,
	\item every class of posets of bounded width,
	\item every class of permutations that excludes a fixed permutation pattern.
\end{itemize}
On the other hand, the following classes do not have bounded twin-width:
\begin{itemize}
	\item the class of all finite graphs,
	\item the class of all bipartite graphs,
	\item the class of all graphs of degree at most $3$ \cite{tww2}.
\end{itemize}
\end{example}

The following is a consequence of \cite{bonnet2020twin}, Theorems 10 and~14.
\begin{fact}[\cite{bonnet2020twin}]\label{fact:grid-char}
	The following conditions are equivalent for a class of finite binary structures $\CC$:
	\begin{enumerate}
	\item $\CC$ has bounded twin-width,
	
	\item There is some $t\in\N$ such that for every structure $\str S\in \CC$ there is an order on $\str S$ such that for every two convex partitions $\Ll$ and $\Rr$ of $\str S$ into $t$ parts, there are  $L\in\cal L$ and $R\in\cal R$ such that 
	$|\Types[\Sigma](L/R)|=1$  or $|\Types[\Sigma](R/L)|=1$.
	 \end{enumerate}
\end{fact}
In the case of graph classes $\CC$, the second condition above means the following:
for every $G\in \CC$ there is an order $\le$ on $G$ such that in the adjacency matrix $M$ of $G$ along this order, for every partition $\cal R$ of the rows of $M$ and partition $\cal C$ of the columns of $M$ into convex intervals, there is some pair $R\in\cal R$ and $C\in \cal C$ such that the submatrix $M[R\times C]$ has either all rows equal, or all columns equal.

\begin{fact}[\cite{bonnet2020twin}]\label{fact:properties}
	Classes of bounded twin-width are closed under:
	\begin{enumerate}
		\item\label{prop:monadic} \emph{expanding by unary predicates:}
		if $\CC$ is a class of bounded twin-width and $\CC'$ is a class of structures such that each $\str S'\in\CC'$ is an expansion of some $\str S\in\CC$ by unary predicates, then $\CC'$ has bounded twin-width.
		\item\label{prop:induced} \emph{taking induced substructures}:
        if $\CC$ has bounded twin-width and $\CC'$ is a class such that every structure in $\CC'$ is an induced substructure of some structure in $\CC$, then $\CC'$ has bounded twin-width.
        \item\label{prop:transductions}\emph{first-order transductions:} if $\CC$ is a class of bounded twin-width and $T$ is a first-order transduction, then $T(\CC)$ is a class of bounded twin-width;
		\item\label{prop:orders} \emph{expanding by compatible  orders:}
		if $\CC$ is a class of bounded twin-width then there  
		is an ordered expansion $\CC'$ of $\CC$ (that is, a class of ordered structures $\CC'$ such that for each structure in $\str S\in\CC$, the structure $\str S$ with some order belongs to $\CC'$)
		such that $\CC'$ has bounded twin-width.
	\end{enumerate}	
\end{fact}
\begin{proof}
	\eqref{prop:transductions} follows from~\cite[Theorem 39]{bonnet2020twin}. \eqref{prop:monadic} and~\eqref{prop:induced} are special cases (also \eqref{prop:induced} is immediate by definition:
	when taking a substructure induced by $A$,
	the same contraction sequence, restricted to $A$, works). 
	
	We show \eqref{prop:orders}.
	Let $\str S$ be a structure and let $\cal P_1,\ldots,\cal P_n$ be its contraction sequence. There is an order $\le$ on $\str S$ such that each partition $\cal P_i$ is convex.
	Indeed, order the parts of $\cal P_i$ by induction on~$i$:
	$\cal P_1$ has just one part, so there is nothing to do,
	whereas $\cal P_{i+1}$ is obtained from $\cal P_i$ by merging some two part $X_1,X_2$ of $\cal P_i$ into one  parts $X_1\cup X_2$ of $\cal P_{i+1}$. By induction, $X_1$ and $X_2$ are already ordered, and order $X_1\cup X_2$ by declaring (arbitrarily) that $X_1<X_2$. 

	As each partition $\cal P_i$ is convex with respect to $\le$, a pair of its parts that was homogeneous in $\str S$ remains homogeneous in $(\str S,\le)$. Hence, the contraction sequence $\cal P_1,\ldots,\cal P_n$ has the same red-degree in $\str S$ as in $(\str S,\le)$.

	Given a class of structures $\CC$ of twin-width bounded by $d$, pick a contraction sequence of red-degree at most $d$ for each structure $\str S\in\CC$ and then expand $\str S$ by the order as defined above. This yields a class of ordered structures $\CC'$ of red-degree at most $d$.
\end{proof}
Morally, an ordered expansion $G_\le$ of $G$ with $\tww(G_\le)$ with twin-width bounded by a function of $\tww(G)$
    is the same as $G$ together with a contraction sequence.
    In one direction, given $G$ together with a contraction sequence we can compute (in linear time) a total order $\le$ on $G$ such that $\tww(G_\le)=\tww(G)$, described as above.
    The converse direction is not that clear, however.
	We will prove (cf. Section~\ref{sec:algorithms})
	there is a polynomial-time algorithm which inputs a totally ordered structure $G_\le$ and outputs its contraction sequence 
    of red-degree $f(\tww(G_\le))$, for some computable function $f$. A contraction sequence of $G_\le$ is also a contraction sequence of $G$.

\section{Proof outline}\label{sec:outline}
Our proof proceeds in several steps,
which utilise various notions which turn out to be equivalent to bounded twin-width. 
We explain those notions and outline the proof below. Figure~\ref{roadmap-small} may be helpful in tracing the implications.

\begin{figure}
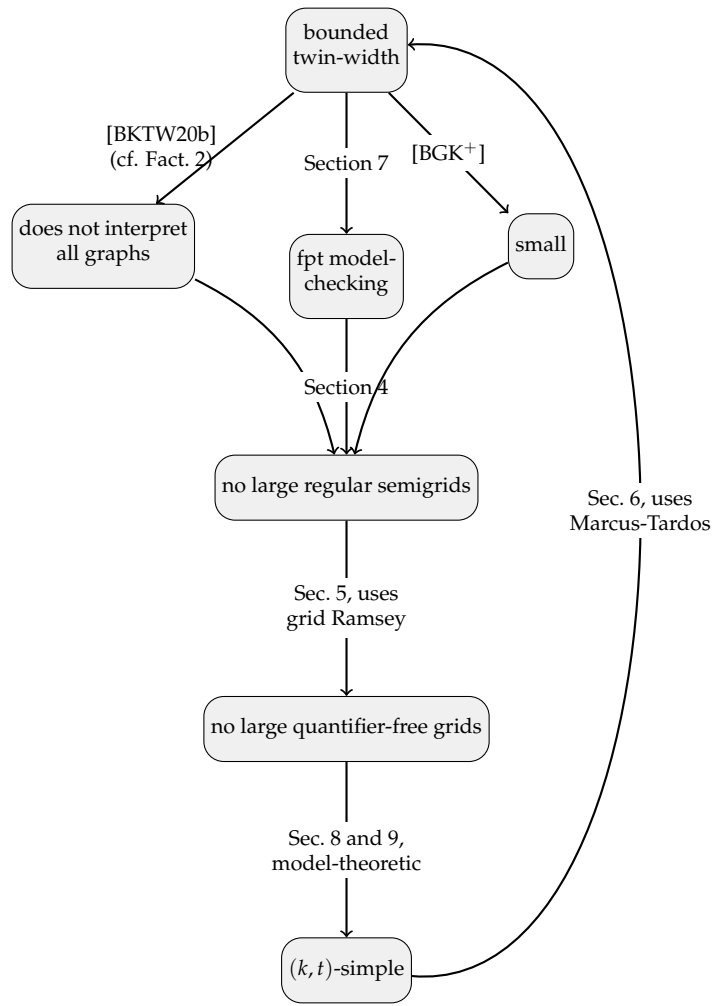

	\begin{centering}
		\footnotesize
	\ctikzfig{tikz-small}
	\end{centering}\caption{Implications among the conditions in Theorem~\ref{thm:intro}.
	}\label{roadmap-small}
	\end{figure}

\medskip


The first notion is as follows.

\begin{definition}[$(k,t)$-simple]An ordered binary structure $\str S$ is \emph{$(k,t)$-simple} if for every  pair of convex partitions $\cal L,\cal R$ with $|\cal L|=|\cal R|\ge t$ of $\str S$, there are  $L\in \cal L$ and $R\in \cal R$ such that  $|\Types[\Sigma](L/R)|\le k$ and 
	${|\Types[\Sigma](R/L)|\le k}$.
\end{definition}
For an ordered graph $G$ this means that if the set of rows and set of  columns of the adjacency matrix of $G$ are  partitioned into $t$ intervals each, then 
there there is a submatrix induced by some two of those intervals which has at most $k$ distinct rows and at most $k$ distinct columns. This perspective can be adapted to binary structures,
by suitably defining their adjacency matrices (cf. Sec.~\ref{sec:grid-and-tww}).


As a first step, we prove that $(k,t)$-simple classes have bounded twin-width.
\begin{theorem}\label{thm:k-t-simple}
	Every $(k,t)$-simple finite ordered binary structure has twin-width $2^{2^{\Oof(kt)}}$.
\end{theorem}
This is shown in Section~\ref{sec:grid-and-tww}, and follows the ideas present in~\cite{bonnet2020twin}. The proof is combinatorial, and uses the Marcus-Tardos result about 0-1 matrices with many entries equal to $1$, which is the cornerstone  of twin-width theory. Our proof generalizes a result and a construction of Bonnet et al.~\cite[Theorem 14]{bonnet2020twin}.

\medskip
Another key notion is that of defining grids by a first-order formula. It makes sense in any class of structures over any signature $\Sigma$.

Fix a first-order formula $\phi(\tup x,\tup y, z)$, where $\tup x$ and $\tup y$ are sets of variables and $z$ is a single variable. An \emph{$m\times n$ grid}  defined by $\phi$  in a structure $\str S$
is a triple of sets $A\subset \str S^{\tup x}$, $B\subset  \str S^{\tup y}$ and $C\subset\str S$ with $|A|=m$, $|B|=n$ and $|C|=m\times n$, such that the relation 
\[\setof{(\tup a,\tup b,c)\in A\times B\times C}{\str S\models\phi(\tup a,\tup b,c)}\]
is the graph of a bijection between $A\times B$ and $C$. 
More explicitly, for each $(\tup a,\tup b)\in A\times B$ there is a unique $c\in C$ such that $\str S\models \phi(\tup a,\tup b,c)$, and conversely, for each $c\in C$ there is a unique $(\tup a,\tup b)\in A\times B$ such that $\str S\models \phi(\tup a,\tup b,c)$.

\begin{definition}[Defining large grids]\label{def:grids}
	A class of structures $\CC$ \emph{defines large grids} if there is a formula $\phi(\tup x,\tup y,z)$ such that for all $n\in\N$  $\phi$ defines an  $n\times n$ grid in some structure $\str S\in \CC$. 
\end{definition}
\begin{example}
	Consider an $m\times n$-semigrid as discussed in the introduction. It consists of $m+1$ intervals $I_0,\ldots,I_m$ of size $n+1$ each, and with $I_1<\ldots<I_m$.
	We show that there is a formula $\phi(x_1,x_2; y_1,y_2;z)$  which defines an $m\times n$-grid in an $m\times n$-semigrid $G$ of a given type $R\in\set{\le,\ge,=,\neq}$.

	The set $A\subset G^2$ consists of the pairs $\tup a_i=(\min I_i,\max I_i)$ of endpoints of the intervals $I_i$, for $i=1,\ldots,n$.
	Suppose the semigrid has type $=$.
	Then we take $B=I_0\subset G$,
	 $C=I_1\cup\cdots\cup I_n$, and 
	 $\phi(x_1,x_2;y;z)\equiv (x_1\le z\le x_2)\land E(x,z).$
	 The reader is invited to check that this formula defines an $n\times (n+1)$-grid $(A,B,C)$. 
	 
	 The case of a $\le$-semigrid is only slightly different. Note that in a $\le$-graph formed by two sets $X$ and $Y$ (cf. Fig.~\ref{fig:halfgraphs}), the $i$th vertex $y\in Y$ in the bottom can be characterized in terms of $X$ by saying that it is adjacent to the $i$th vertex in $X$ and non-adjacent to the $(i+1)$st vertex in $X$. Because of this, we now take $B\subset G^2$ to consist of pairs of consecutive elements of $I_0$,
	 and for $A$ and $C$ as above, use the formula $\psi(x_1,x_2;y_1,y_2;z)\equiv (x_1\le z\le x_2)\land (E(y_1,z)\land\neg E(y_2,z))$. Then $\psi$ defines an $n\times n$-grid $(A,B,C)$.

	 The cases of $\ge$ and $\neq$-semigrids are similar.
\end{example}

The core of our proof is the following:
\begin{theorem}\label{thm:core}
	If $\CC$ is a class of ordered binary  structures which is not $(k,t)$-simple for any $k,t\in\N$, then some quantifier-free formula $\phi(\tup x,\tup y,z)$ defines large grids in $\CC$.
\end{theorem}
This is proved using model-theoretic methods. In particular, we use compactness  to construct an infinite ordered binary structure $\str N$ which defines infinite grids.
	This is done in Sections~\ref{sec:equivalences} and \ref{sec:main-proof}.

	Model-theoretic ideas are at the heart of our approach. We provide further characterisations of bounded twin-width classes in terms of notions originating from model-theory in Section~\ref{sec:equivalences}.
	In particular, some of the implications in Theorem~\ref{thm:intro} are stated and proved in much greater generality, for arbitrary classes of structures.

	\medskip

Finally, we prove:

\begin{theorem}\label{thm:grids}Let $\CC$ be a hereditary class of ordered binary structures.
	If $\CC$ defines large grids using a quantifier-free formula, then $\CC$ contains arbitrarily large regular $n\times n$-semigrids.
\end{theorem}
This uses a result from Ramsey theory to regularize the obtained grids. It crucially relies on the grids being defined by a quantifier-free formula. 
	This is proved in Section~\ref{sec:regularise}.

	As a consequence of our proof, we obtain Theorem~\ref{thm:approx}, yielding a polynomial-time approximation algorithm for the twin-width of a given ordered binary structure. This in turn yields Corollary~\ref{cor:fpt}, about fixed-parameter tractability of model-checking of first-order logic on ordered structures of bounded twin-width.
	This is shown in Section~\ref{sec:algorithms}.

	\medskip

From here, it is all downhill. The following result is relatively straightforward.
\begin{proposition}\label{prop:easy}
	Let $\CC$ be a hereditary class of ordered graphs which contains arbitrarily large regular $n\times n$-semigrids.
	Then:
	\begin{itemize}
		\item $\CC$ interprets the class of all graphs,
		\item model-checking first-order logic on $\CC$ is as hard as on the class of all graphs,
		\item $\CC$ contains at least $\lowerbound$ non-isomorphic structures with $n$ vertices.
	\end{itemize}
\end{proposition}

 Conversely, by the results of~\cite{tww1} and~\cite{tww2} classes   of ordered graphs of bounded twin-width:
 \begin{itemize}
	 \item are preserved under transductions,	 
	  \item have at most $2^{\Oof(n)}$ structures of size $n$, up to isomorphism,
 \end{itemize}

 This closes the loop (cf. Fig.~\ref{roadmap-small}) and proves the equivalence of all the notions considered above, yielding our first main result,
Theorem~\ref{thm:intro}. 

\medskip
Our second main result, Theorem~\ref{thm:summary-mt0}, applies to arbitrary  classes of structures, not necessarily finite, ordered or binary. Similarly as above, at its core is a suitable grid theorem. We define a notion \emph{regular} classes $\CC$, which generalizes the notion of $(k,t)$-simple classes above to arbitrary structures. 
We then show that if a class $\CC$ is not regular then it defines large grids.
This is proved in Section~\ref{sec:equivalences}, alongside  with Theorem~\ref{thm:core}. In the process, we exhibit various properties of monadically NIP classes of structures,
in particular, that they are \emph{1-dimensional}.

\section{Regular semigrids}
\label{sec:bad}
In this section, we define the notion of regular semigrids,
and prove that they exhibit bad behaviour,
proving Proposition~\ref{prop:easy}.
In particular, they do not have bounded twin-width.

We start with defining regular semigrids.
The definition below applies only to ordered graphs, rather than to ordered binary structures.
The complete definition of regular semigrids for ordered binary structures is deferred to Section~\ref{sec:regularise}. Proposition~\ref{prop:easy} will be easily lifted to arbitrary binary signatures (cf. Corollary~\ref{cor:fails-sigma}).

\newcommand{\direction}{\mathrm{dir}}

\medskip
Fix an $m\times n$-semigrid 
consisting of disjoint intervals $I_0,\ldots,I_m$ of length $n+1$ each, with $I_1<\ldots<I_m$.
 It will be convenient to identify 
 the elements of the $m\times n$ semigrid with the points in $[m+1]\times [n+1]$ so that for each $i\in[m+1]$, the elements of $I_i$ 
	 are identified with the elements in the $i$th row $\set i\times [n+1]$ in increasing order 
	 (that is, the smallest element in $I_i$ is identified with $(i,0)$, the second smallest with $(i,1)$, etc.).
	 In other words, the intervals $I_0,\ldots,I_m$ are 
	 arranged in an $(m+1)\times (n+1)$ matrix by putting $I_0$ in the first row,
	 $I_1$ in the next one, and so on.
	 \todo{pic} Note that the order $\le$ on $I_1\cup\cdots\cup I_m$ (with the interval $I_0$ omitted) agrees with the lexicographic order on $\set{1,\ldots,m}\times\set{0,\ldots,n}$.
	 For $p,q\in I_1\cup\cdots\cup I_m$, with $p=(i,j)$ 
	 in the $i$th row and $j$th column and $q=(i',j')$ in the $i'$th row and $j'$th column, 
	 and $p< q$ lexicographically, let
	 $\direction(p,q)\in\set{\rightarrow,\downarrow,\searrow,\swarrow}$ be equal to:	 
	 \begin{quote}
		 $\rightarrow$ if $i=i'$, $j<j'$,\qquad
		$\downarrow$ if $j=j'$, $i<i'$,
	   \qquad $\searrow$ if $i<i'$,$j>j'$,\qquad
		$\swarrow$ if $i<i'$, $j'<j'$.
	 \end{quote}		 	
An $m\times n$-semigrid is~\emph{regular}~if:
 \begin{enumerate} 
	\item
 $I_0<I_1<\ldots<I_m$ or  $I_1<\ldots<I_m<I_0$, 
 \item $I_0$ forms a clique or an independent set, and
 \item for $p,q\in I_1\cup\cdots\cup I_n$ with $p< q$, the adjacency between $p$ and $q$
 depends only $\direction(p,q)$.
\end{enumerate}
\newcommand{\schemes}{\mathrm{Schemes}}
\newcommand{\regsigma}{\mathscr R_\sigma}
	A \emph{scheme} of a semigrid consists of: the type $R\in\set{\le,\ge,\neq,=}$, a parameter $S\in \set{<,>}$ 
corresponding to the choice (1) above,
a parameter $T\in\set{\text{clique},\text{independent}}$ corresponding to the choice (2)  above,
and a subset $X\subset \set{\rightarrow,\downarrow,\swarrow,\searrow}$ corresponding to the choices (3) above, where a pair $p,q$ with $p<q$ is adjacent if and only if $\direction(p,q)\in X$.
 Denote by \[\schemes=\set{\le,\ge,\neq,=}\times \set{<,>}\times \set{\text{clique},\text{independent}}\times 2^{\set{\rightarrow,\downarrow,\swarrow,\searrow}}\] the set of all possible schemes, so $|\schemes|=2^8$. 
For each $\sigma\in \schemes$, let $\regsigma$ denote the class of all regular $m\times n$ semigrids conforming to $\sigma$, for all $m,n\in\N$. 

Note that a regular $m\times n$-semigrid $G\in\regsigma$
contains as induced substructures 
every regular $m'\times n'$-semigird $G'\in \regsigma$, for all $m'\le m$ and $n'\le n$ (by removing some rows and columns of $G$),
so if a hereditary class $\CC$ contains arbitrarily large regular $n\times n$-semigrids then $\CC\supseteq\regsigma$ for some $\sigma\in\schemes$.

The definition of regular semigrids is generalized to arbitrary binary structures in Section~\ref{sec:regularise}.

\medskip
Proposition~\ref{prop:easy} is restated below.
\begin{proposition}\label{prop:fails}Let $\CC$ be a hereditary class of ordered graphs.
	If $\CC$ contains arbitrarily large regular semigrids then:
	\begin{enumerate}
		\item\label{it:fails-inter} $\CC$ interprets  the class of all graphs,
		\item\label{it:fails-mc} model-checking first-order logic on $\CC$ is AW[$\ast$]-hard,
		\item\label{it:fails-small} $\CC$ contains at least  $\lowerbound$ non-isomorphic structures with $n$ elements.
	\end{enumerate}
\end{proposition}
\begin{proof}

	Fix any $\sigma\in\schemes$ such that $\CC$ contains $\regsigma$. 
	Without loss of generality we may assume that $\sigma$ imposes an independent set on $I_0$ (otherwise we replace edges by non-edges in all graphs in $\CC$ and  replacing `$\mathrm{clique}$' by `$\mathrm{independent}$' in $\sigma$).
	Also without loss of generality we may assume that $\sigma$ imposes $I_0<I_1$ (otherwise reverse the order in all structures in $\CC$ and replace $<$ by $>$ in $\sigma$).

	\newcommand{\rows}{\cal R}
	\newcommand{\cols}{\cal C}

	Fix a semigrid $G\in\regsigma$; recall that its 
	elements are identified with $[m+1]\times[n+1]$.
	Let  $\ast=(0,0)$ denote the smallest element in $G$. This is also the smallest element of $I_0$, by the assumption above. 
	Let $\rows=\set{\min I_1,\ldots,\min I_m}=\set{(1,0),(2,0),\ldots,(m,0)}$ be the set of smallest elements of each of the intervals $I_1,\ldots,I_m$ ($I_0$ is omitted). Those elements represent the rows of $M=\set{1,\ldots,m}\times \set{1,\ldots,n}$.
	Let $\cols=\set{(0,1),\ldots,(0,n)}=I_0\setminus\set\ast$; those elements represent the columns of $M$.\todo{pic}

	For a given $S\subset M$,  denote by $G^S$ the ordered subgraph of $G$ induced by $\set\ast\cup \cols\cup\rows\cup S$. 
	The following is immediate.
\begin{claim}\label{cl:effective-semigrid}
	There is a polynomial-time algorithm which inputs numbers $m,n\in\N$, a relation $S\subset \set{1,\ldots,m}\times\set{1,\ldots,n}$ and in polynomial time outputs the ordered graph $G^S\in\regsigma$.
\end{claim}

	This will be the basis of our reduction from the model-checking problem over the class of arbitrary binary relations to that over $\CC$.
	
	\medskip
	We now show how to interpret the binary relation $S$ in $G^S$ by a first-order formula.
	The following lemma is proved by a simple case analysis, depending on $\sigma\in\schemes$.

\begin{lemma}\label{lem:semigrid-formulas}
	There are first-order formulas $\phi_\cols(x),\phi_\rows(x),\pi_1(x,y),\pi_2(x,y),\rho(x,y)$ 
	such that for every $m\times n$-semigrid $G\in\regsigma$ and $S\subset \set{1,\ldots,m}\times\set{1,\ldots,n}$:
	\begin{align*}		
	\phi_\cols(G^S)&=\cols\subset G^S,\\
	\phi_\rows(G^S)&=\rows\subset G^S,\\
	\pi_1(G^S)&=\setof{((i,j),(i,0))}{(i,j)\in S}\subset S\times \rows\subset  G^S\times G^S\\
	\pi_2(G^S)&=\setof{((i,j),(0,j))}{(i,j)\in S}\subset S\times \cols\subset G^S\times G^S\\
	\rho(G^S)&=\setof{((i,0),(0,j))}{(i,j)\in S}\subset \rows\times\cols\subset  G^S\times G^S.
\end{align*}	
In words, $\pi_1$ and $\pi_2$ define the graphs of the two projections from $S\subset \set{1,\ldots,m}\times \set{1,\ldots,n}$ onto 
$\rows\simeq\set{1,\ldots,m}$ and $\cols\simeq\set{1,\ldots,n}$, while 
$\rho$ defines the relation $S\subset\set{1,\ldots,m}\times \set{1,\ldots,n}$ as a relation on $\rows\times\cols$.

\end{lemma}
\begin{proof}


	If $\sigma$ specifies a type among $\le,\ge,=$, then $\cols$ is the open interval between $\ast$ and the smallest neighbor of $\ast$.
	If $\sigma$ specifies the type $\neq$, then $\cols$ is 
	the open interval between $\ast$ and the smallest neighbor of the successor of $\ast$. In either case, $\cols$ can be described by a first-order formula $\phi_\cols(x)$ not depending on $m,n\in\N$ and $S$.

	If $\sigma$ specifies the type  $=$ or $\ge$, then $\rows$ is the set of neighbors of $\ast$.
	If $\sigma$ specifies the type $\le$ then $\rows$ is the neighborhood of $\ast$ minus the neighborhood of the successor of $\ast$. Finally, if $\sigma$ specifies the type $\neq$ then $\rows$ is  the complement of $\cols\cup \set\ast \cup N(\ast)$ (the set of neighbors of $\ast$). In all cases, $\rows$ can be described by a first-order formula $\psi_\rows(x)$ not depending on $m,n$ and $S$.

	Note that an element $p\in S$ lies in the $i$th row if and only if $(i,0)<p<(i+1,0)$, which can be expressed by a first-order formula. More precisely, consider the first-order formula $\pi_1(x,y)$ which holds if and only if $y\in\rows$, and $y<x<y'$ where $y'$ is the least $y'\in\rows$ with $y> y'$ (unless no such $y'$ exists, then just take $y<x$). Then $\pi_1(p,q)$ holds if and only if $p$ and $q$ are in the same row and $p\in\rows$, defining the required relation.

	Dually, we can define that an element in $S$ is in a column $c\in\cols$. Namely, let $\pi_2(x,y)$ be the first-order formula expressing that $y\in\cols$, 
	and:
	\begin{description}
	\item[(case $=$)] $x$ is a neighbor of $y$,
	\item[(case $\neq$)] $x$ is a non-neighbor of $y$,
	\item[(case $\le$)] $y$ is a neighbor of $x$ but not of the predecessor  of $x$,
	\item[(case $\ge$)] $x$ is a neighbor $y$ but not of the predecessor of $y$.
\end{description}
Finally, the formula $\rho(x,y)$ expresses that $x\in\rows,y\in\cols$ and there exists $z$ such that $\pi_1(z,x)$ and $\pi_2(z,y)$ hold. 
\end{proof}

	\newcommand{\BB}{\mathscr B}
Let $\BB$ be the class of all finite, ordered bipartite graphs $H=(V_1\cup V_2,E,\le_1,\le_2)$ with parts $V_1$ and $V_2$, where $V_1$ is totally ordered by $\le_1$ and $V_2$ is totally ordered by $\le_2$, and $E\subset V_1\times V_2$.
	\begin{lemma}\label{lem:universal-interpretation}
		There is an interpretation $I_\sigma$ such that ${I_\sigma(\CC)=\BB}$ for every hereditary class of ordered graphs $\CC$ containing $\regsigma$.
		More precisely, for $S\subset \set{1,\ldots,m}\times\set{1,\ldots,n}$,
		the structure $I_\sigma(G^S)$ is the bipartite graph with parts $\rows=\set{1,\ldots,m}$ and $\cols=\set{1,\ldots,n}$, each ordered by $\le$, where $i\in\rows$ and $j\in\cols$ are adjacent if and only if $(i,j)\in S$. 	
\end{lemma}
	\begin{proof}
		$I_\sigma$ consists of:
 \begin{itemize}
	 \item the the domain formula $\delta(x)\equiv \phi_\rows(x)\lor \phi_\cols(x)$ defining the domain $\rows\cup\cols$, 
	 \item the formula $\phi_E(x,y)\equiv \rho(x,y)$ defining the edge relation $E$,  
	 \item 
	 the formula $\phi_\rows(x)\land\phi_\rows(y)\land x\le y$ defining the order $\le_1$ on $\rows$, and
	 \item the formula 
	 $\phi_\cols(x)\land\phi_\cols(y)\land x\le y$ defining the order $\le_2$ on $\cols$.
 \end{itemize}
The statement follows from Lemma~\ref{lem:semigrid-formulas}.
\end{proof}

As the class of all graphs interprets in the class of all bipartite graphs via some first-order interpretation $T$, 
the interpretation $I_\sigma$ can be further composed with $T$ obtaining an interpretation $T\circ T_\sigma$ such that $(T\circ T_\sigma)(\CC)$ is the class of all finite graphs. This proves~\eqref{it:fails-inter}.

This also gives a polynomial-time reduction of the model-checking problem on the class of bipartite graphs to the  model-checking problem on $\CC$:
	a sentence $\psi$ holds in a bipartite graph $H=(V_1,V_2,E)$ with $V_1=\set{1,\ldots,m}, V_2=\set{1,\ldots,n}$ if and only if the sentence obtained from $\psi$ by replacing each atom $E(x,y)$ by $\rho(x,y)$ holds in the structure $G^S$, where $G\in \regsigma$
	is the $m\times n$-semigrid and $S=E\cap (V_1\times V_2)$. 
	As $G^S$ can be computed in polynomial time given $H$ (cf. Claim~\ref{cl:effective-semigrid}),
	this yields the required polynomial-time reduction. Since model-checking first-order logic on the class of all bipartite graphs is AW[$\ast$]-hard~\cite{10.5555/1121738}, this yields the same lower bound for $\CC$, proving~\eqref{it:fails-mc}.
	
	\medskip
	We now prove~\eqref{it:fails-small}.
 For a $k\times l$ matrix $M\in\set{0,1}^{k\times l}$ with 0-1 entries 
 let $|M|$ denote $k+l$ plus the number of nonzero entries in $M$.
For $n\in\N$ define 
\[f(n)=\max\setof{|M|}{k,l\in\N, M\in \set{0,1}^{k\times l}}.\]

It is easy to see that $f(n)\ge \lfloor \frac n 3 \rfloor!$.
Indeed: for $n=3k$ 
consider a
$k\times k$ permutation matrix $M$,
having exactly one nonzero entry in each row and each column. Then $|M|=n$, and there are exactly $k!$ such matrices.

This bound can be easily improved:
instead of permutation matrices we can consider all $k\times k$ matrices with exactly one nonzero entry in  each column,
yielding $f(n)\ge k^k$ for $n=3k$.
Further, taking all $k\times k$ matrices with exactly $k$ nonzero entries yields $f(n)\ge {k^2\choose k}$ for $n=3k$.
We expect those lower bounds can be improved further. 
We use the first lower bound 
as it yields the simplest expression among the three.

\begin{lemma}
	$\CC$ contains
	at least $f(n)$ non-isomorphic structures with $n$ vertices, for each $n\in\N$.
\end{lemma}
\begin{proof}
	
	For all $k,l\in\N$ and matrix $M\in \set{0,1}^{k,l}$, consider the $k\times l$-regular semigrid $G\in\regsigma$ and denote by $F(k,l,M)$ the ordered graph $G^S\in \CC$ as constructed above for $S$ being the set of non-zero entries in $M$:\[S=\setof{(i,j)}{1\le i\le k,1\le j\le l, M_{ij}=1}.\]

	It follows from the above that $F$ 
maps distinct matrices to non-isomorphic structures in $\CC$. Indeed: for the transduction $T$ defined above, $T(F(m,n,M))$ uniquely determines $m,n$ and $M$.
This proves the lemma.
\end{proof}

This  completes the proof of Proposition~\ref{prop:fails}, and hence Proposition~\ref{prop:easy}. 
(See Corollary~\ref{cor:fails-sigma} below for a generalization to arbitrary binary signatures).
\end{proof}

\begin{corollary}
	$\CC$ has unbounded twin-width.
\end{corollary}
\begin{proof}
	If $\CC$ has bounded twin-width then it does not interpret the class of all graphs~\cite{tww1} and is  small~\cite{tww2}. 
\end{proof}
\section{From quantifier-free grids to regular semigrids}\label{sec:regularise}
In this section we prove 
 Theorem~\ref{thm:grids},
which we recall below, generalized to arbitrary binary signatures $\Sigma$:
\begin{theorem}\label{thm:grids-binary}
    Let $\CC$ be a hereditary class of ordered binary structures graphs.
	If $\CC$ defines large grids using a quantifier-free formula, then $\CC$ contains arbitrarily large regular $n\times n$-semigrids.
\end{theorem}
In particular, we need to define the notion of a regular semigrid for an arbitrary binary signature $\Sigma$.


Before doing this, we introduce the relevant result from Ramsey theory.
Let an \emph{ordered $m\times n$ grid}
be the relational structure with domain
$[m]\times[n]$ and two quasi-orders $\le_1$ and $\le_2$, where two points $(i,j),(i',j')$ 
are related by $\le_1$ 
if and only if $i\le i'$ and are related by $\le_2$ if and only if $j\le j'$.
A \emph{pair coloring} of such an ordered $m\times n$ grid is a function $c\from 
([m]\times[n])^2\to \Gamma$ to a fixed finite set $\Gamma$ of \emph{colors}.
Such a coloring is \emph{homogeneous}
if the color $c(p,q)$ depends only on the atomic type of the pair $(p,q)$. In other words, the color of a pair $(p,q)$ depends only on whether  $p=q$, and if not, on the quadrant or principle semi-axis in $\Z^2$ to which the vector $p-q$ belongs to.
There are $9$ possible types, corresponding to the pairs in $\set{<,=,>}^2$ describing the relations between $p$ and $q$ in each of the two dimensions.

The following result is a special case of the so-called \emph{product Ramsey theorem}
(see e.g.~Proposition 3 in~\cite{bodirsky_2015}
in the special case of the full product of two copies of $(\mathbb Q,\le)$. See also the historical comment following it).
\begin{lemma}\label{lem:grid-ramsey}
Fix a finite set of colors $\Gamma$.
For every $m,n$ there are $m',n'$ such that 
for every coloring $c$ of the ordered $m'\times n'$ there is an induced substructure isomorphic to the $m\times n$ grid for which  the coloring induced by $c$ is homogeneous.
\end{lemma}

Fix a formula $\phi(\tup x,\tup y,z)$,
an ordered structure $\str S$ and an $m\times n$ grid $(A,B,C)$ defined by $\phi$ in~$\str S$.
The sets $A\subset \str S^{\tup x}$ and $B\subset \str S^{\tup y}$ are ordered lexicographically with respect to the order on~$\str S$, and some fixed enumeration of the tuples $\tup x$ and $\tup y$.
Those two orders induce two quasi-orders $\le_1,\le_2$ on $C$, via the bijection $\alpha \from A\times B\to C$ defined by the formula $\phi$. Namely, $\alpha(\bar a,\bar b)\le_1 \alpha(\bar a',\bar b')$ if and only if $\bar a\le_\textrm{lex} \bar a'$ 
and $\alpha(\bar a,\bar b)\le_2 \alpha(\bar a',\bar b')$ if and only if $\bar b\le_\textrm{lex} \bar b'$.

Say that the grid $(A,B,C)$ is \emph{homogeneous} if for every $(\bar a,\bar b),(\bar a',\bar b')\in A\times B\subset \str S^{\tup x\tup y}$ the atomic type of $(\bar a,\bar b,\bar a',\bar b',\alpha(\bar a,\bar b),\alpha(\bar a',\bar b'))$ in $\str S$ depends only on the atomic type of $(\alpha(\bar a,\bar b),\alpha(\bar a',\bar b'))$
in $(C,\le_1,\le_2)$.

\begin{lemma}\label{lem:homogenize}
Fix a finite signature $\Sigma$ containing $\le$ and  a $\Sigma$-formula $\phi(\tup x,\tup y,z)$.
For every $m,n$ there are $m',n'$ such that 
if $\phi$ defines an $m'\times n'$ grid in an ordered $\Sigma$-structure $\str S$ then $\phi$ also defines an $m\times n$ homogeneous grid in $\str S$.
\end{lemma}
\begin{proof}
Let $\Gamma$ be the set of atomic $\Sigma$-types of tuples of length $2(|\tup x|+|\tup y|+1)$.
Apply Lemma~\ref{lem:grid-ramsey} yielding numbers $m',n'$.

Suppose $\phi$ defines an $m'\times n'$ grid $(A,B,C)$ in $\str S$, and let $\alpha\from A\times B\to C$ be the bijection defined by $\phi$.
Consider $C$ with the two orders $\le_1$ and $\le_2$ as described earlier, then $(C,\le_1,\le_2)$ is isomorphic to the ordered $m'\times n'$ grid.
Color each pair $(\alpha(\tup a,\tup b),\alpha(\tup a',\tup b'))\in C^2$ by the atomic type of the tuple $(\bar a,\bar b,\bar a',\bar b',\alpha(\bar a,\bar b),\alpha(\bar a',\bar b'))$ in $\str S$.

By Lemma~\ref{lem:grid-ramsey}, $(C,\le_1,\le_2)$ contains a copy $C'\subset C$ of a homogeneous ordered $m\times n$-grid. 
Then $\alpha^{-1}(C')\subset A\times B$ and $C'$ form a homogeneous  $m\times n$ grid defined by $\phi$.
\end{proof}

We now generalize the notion of a regular semigrid from Section~\ref{sec:bad} to
an arbitrary relational binary signature $\Sigma$
containing the order symbol $\le$.

Fix $m,n\in\N$ and consider the ordered grid $(m+1)\times (n+1)$-grid $(C,\le_1,\le_2)$ with domain $C=[m+1]\times[n+1]$ as described above.
Let $B=[0]\times[n+1]\subset C$,
 and let $\pi\from C\to B$ be the projection
 mapping $(i,j)\in C$ to $(0,j)$.

A \emph{regular $m\times n$-semigrid} is an ordered $\Sigma$-structure $\str S$ with domain $C$   such that the order $\le$ of $\str S$ agrees with the natural lexicographic order on $C=[m+1]\times [n+1]$ or its inverse,
and for all $a,b\in B$ and $p,q\in C\setminus B$:
\begin{enumerate}    
\item   the atomic type of $(a,b)$ in $\str S$ depends only on the atomic type of $(a,b)$ in $(B,\le)$;
\item the atomic type of  $(p,q)$ in $\str S$ depends only on the atomic type of $(p,q)$ in $(C,\le_1,\le_2)$;
\item  the atomic type of  $(p,b)$ in $\str S$ depends only on the atomic type of $(\pi(p),b)$ in $(B,\le)$, and is non-constant.
\end{enumerate}
A regular semigrid is completely specified by the dimensions $m\times n$,
the information ($<$ or $>$) about whether or not to reverse the lexicographic order,
as well as a tuple of functions describing the dependencies as above.
For example, the dependencies (1) and (3) above 
are each specified by a function from $\set{<,=,>}$ to the set $T$ of atomic $\Sigma$-types of pairs, whereas the dependency (2) is specified by a function from $\set{<,=,>}^2$ to $T$. In total, all the dependencies are described by a tuple  in $T^{3+3+9}=T^{15}$ (with some tuples being excluded due to the non-constant requirement, and some being unsatisfiable in an ordered structure). 
Let $\schemes=\set{<,>}\times T^{15}$ denote the set of possible choices, called \emph{schemes}.

As  $|T|\le 4^{|\Sigma|}$, we get the following.

\begin{lemma}\label{count-semigrids}
For every $n\ge 1$ there are at most $2^{\Oof(|\Sigma|)}$ distinct regular $n\times n$-semigrids, and they can be enumerated in time $2^{\Oof(|\Sigma|)}\cdot \textit{poly}(n)$.
\end{lemma}

\begin{lemma}
Assume $\Sigma=\set{E,\le}$ consists only of the edge symbol $E$ and order symbol $\le$.
Then an ordered graph is a regular semigrid 
in the sense of Section~\ref{sec:bad} if and only if it is a regular semigrid as defined above.
\end{lemma}
\begin{proof}[Proof sketch]
    In one direction, suppose $G$ is a regular semigrid in the sense of Section~\ref{sec:bad} with intervals $I_0,\ldots,I_m$ of size $n+1$,
    and assume $I_0<I_1<\ldots<I_m$ (in the other case reverse the order of $G$).
    Arrange $I_0,\ldots,I_m$ in a grid $C=[m+1]\times[n+1]$ 
    by placing $I_0$ in the first row, $I_1$ in the second, etc. Then the order of $G$ on $I_0\cup\ldots\cup I_m$ agrees with the lexicographic order on $C=[m+1]\times[n+1]$, and the conditions (1)-(3) in the definition above hold.

    The other direction is similar.
\end{proof}

\begin{lemma}\label{lem:map-semigrids}
    For any class $\CC$ of ordered $\Sigma$-structures which contains arbitrarily large $n\times n$-semigrids 
there is a quantifier-free interpretation $I\from\Sigma\to \set{E,\le}$ such that $I(\CC)$ is a class of ordered graphs containing arbitrarily large regular $n\times n$-semigrids. 
\end{lemma}
\begin{proof}[Proof sketch]
    We use the notation from the definition of regular semigrids above.
    Fix $\sigma\in\schemes$ such that there are arbitrarily large $n\times n$-semigrids in $\CC$ conforming to $\sigma$.
    Then $\sigma$ specifies that 
    the atomic type of $(p,q)$ in $\str S$ depends in the same way on the atomic type of $(\pi(p),b)$ in $(B,\le)$;
    in particular, this dependency is non-constant. Hence there are two distinct atomic types $\tau_1,\tau_2$ which are realized as atomic types of pairs of the form $(\pi(p),b)$ for $b\in B,p\in C\setminus B$.

    The interpretation $I$ is the  interpretation 
    with domain formula $\delta(x)=(x=x)$,
    order formula $\phi_\le(x,y)=(x\le y)$ and
     edge formula $\phi_E(x,y)=\tau_1(x,y)\lor \tau_1(y,x)$. 
     
     If $\str S$ is a regular $m\times n$-semigrid of scheme $\sigma$ then $I(\str S)$ is a regular $m\times n$-semigrid which is an ordered graph.
     Indeed, since $I$ is quantifier-free, the atomic type of $(x,y)$ in $I(\str S)$ depends only on the atomic type of $(x,y)$ in $\str S$, and hence the conditions in the definition of regularity are met. The non-constancy condition is also satisfied by the choice of $\tau_1\neq\tau_2$.
\end{proof}

We now exhibit large regular semigrids in any class in which some quantifier-free formula defines large grids.
\begin{proposition}
Let $\CC$ be a hereditary class of ordered structures over a finite binary signature $\Sigma$ and let $\phi(\tup x,\tup y,z)$ be a quantifier-free formula which defines large grids in $\CC$. Then $\CC$ contains arbitrarily large regular $n\times n$-semigrids.
\end{proposition}

\begin{proof}
Fix a number $n$.
By Lemma~\ref{lem:homogenize} $\phi$ defines a homogeneous $(n+1)\times (n+1)$ grid $(A,B,C)$ in some structure $\str S\in \CC$. In particular, the atomic type $\atp(\tup a,\tup b)$ does not depend on the choice of $\tup a\in A$ and $\tup b\in B$.

\newcommand{\lex}{\le_{\mathrm{lex}}}
As in the definition of homogeneity of the grid $(A,B,C)$, order $A\subset \str S^{\tup x}$ and $B\subset \str S^{\tup y}$ lexicographically with respect to the order $\le$ on $\str S$ and the fixed enumeration of $\tup x$ and $\tup y$.
Let $\le_A$ and $\le_B$ be the resulting orders of $A$ and $B$.
This induces, via the bijection $\alpha\from A\times B\to C$ defined by $\phi$, 
eight possible
lexicographic orders on $C$:
\begin{itemize}
    \item we may either order $A$ using  $\le_A$ or its inverse $\ge_A$,
     \item we may either order $B$ using $\le_B$ or its inverse $\ge_B$,
    \item we may then order $C\simeq A\times B$ lexicographically using first the chosen order on $A$ then the chosen order on $B$, or the other way around.
\end{itemize}

\begin{claim}
    The order $\le$ on $C$ coincides with one of the $8$ orders above.
\end{claim}
\begin{proof}
    For $\tup a\in A,\tup b\in B$ denote $[\tup a,\tup b]:=\alpha(\tup a,\tup b)$.
    Pick points $\tup a<_A\tup a'\in A$
and $\tup b<_B\tup b'<_Bb''\in B$.
Assume $[\tup a,\tup b]<[\tup a,\tup b']$; otherwise replace $\le_B$ with $\ge_B$ in the following.
Assume $[\tup a,\tup b']<[\tup a,\tup b']$; otherwise replace $\le_A$ with $\ge_A$ in the following.

Let $C'=\set{[\tup a,\tup b],[\tup a,\tup b'], [\tup a,\tup b''], [\tup a',\tup b']}\subset C\subset A\times B$.\todo{pic}

Compare $[\tup a',\tup  b']$ with $[\tup a,\tup  b'']$.
If $[\tup a',\tup b']<[\tup a,\tup b'']$ then $\le$ coincides on $C'$ with the lexicographic order which first orders according to $\le_B$ and then according to $\le_A$. Otherwise, if $[\tup a,\tup b'']<[\tup a',\tup b']$ then $\le$ coincides on $C'\subset C\simeq A\times B$ with the lexicographic order which first orders according to $\le_A$ and then according to $\le_B$. Let $\lex$ denote this among those two lexicographic orders which agrees with $\le$ on $C'$.

By construction of $C'$,
 for every pair of points $c_1=[\tup a_1,\tup b_1],c_2=[\tup a_2,\tup b_2]$ in $C$, there is a pair $c_1',c_2'$ in $C'$ such that  $[c_1',c_2']$ and $[c_1,c_2]$ have equal atomic types in $(C,\le_1,\le_2)$.
It follows from homogeneity that $\le$ agrees of $\lex$ on all of $C$.
\end{proof}

Suppose that the order $\le$ on $C$ coincides with  one of the four lexicographic orders where $A$ has higher priority than $B$
(otherwise replace $B$ and $A$ in the argument),
and the projection from $C\simeq A\times B$ to $A$ is monotone rather than anti-monotone 
(otherwise replace $\le$ with $\ge$ in the argument).

\medskip
Pick points $\tup a<\tup a'$ in $A$
and $\tup b<\tup b'$ in $B$ and let $c=[\tup a,\tup b]\in C$, so that $\phi(\tup a,\tup b,c)$ holds.

\begin{claim}
 $\atp(\tup b,c)\neq\atp(\tup b',c)$.
\end{claim}
\begin{proof}
 Note that $\atp(\tup a,\tup b)=
\atp(\tup a,\tup b')$ by homogeneity.
If
$\atp(\tup b,c)=\atp(\tup b',c)$
then we would have $\atp(\tup a,\tup b,c)=\atp(\tup a,\tup b',c)$
since the signature is binary.
This is a contradiction since $\phi$ is quantifier-free and holds of $(\tup a,\tup b_,c)$ but not of $(\tup a,\tup b',c).$
\end{proof}

\medskip
As $\Sigma$ is binary, there is a variable $y\in \tup y$ such that 
$\atp(\tup b(y),c)\neq\atp(\tup b'(y),c)$.
Let $B'=\pi_y(B)\subset \str S$ be the projection of 
$B\subset \str S^{\tup y}$ onto the $y$-component of $\tup y$.

It is now straightforward to verify that the substructure of $\str S$ induced by $C\cup B'$ is (isomorphic to) a regular $(n+1)\times n$-semigrid. In particular, it contains a regular $n\times n$-semigrid as an induced substructure.
\end{proof}

\begin{corollary}\label{cor:fails-sigma}
If $\CC$ is a hereditary class of ordered $\Sigma$-structures containing arbitrarily  large regular $n\times n$-semigrids then:
\begin{itemize}
    \item $\CC$ interprets all graphs,
    \item $\CC$ is not small,
    \item model-checking first-order logic is AW[$\ast$]-hard on $\CC$,
    \item $\CC$ has unbounded twin-width.
\end{itemize}
\end{corollary}
\begin{proof}
    Let $I$ be the quantifier-free interpretation from Lemma~\ref{lem:map-semigrids}. Then $I(\CC)$ is a class of ordered graphs that contains arbitrarily large semigrids, hence satisfies the above conditions by Proposition~\ref{prop:fails}. This implies the same properties for $\CC$.
\end{proof}

\section{$(k,t)$-simplicity}\label{sec:grid-and-tww}
In this section we prove Theorem~\ref{thm:k-t-simple}, which is repeated below:
\begin{theorem*}[\ref{thm:k-t-simple}]
	Every $(k,t)$-simple finite ordered binary structure has twin-width $2^{2^{\Oof(kt)}}$.
\end{theorem*}

We follow the general scheme of the proof of the implication  (2)$\rightarrow$(1) in Fact~\ref{fact:grid-char}, presented in~\cite[Section 5]{bonnet2020twin}.

Instead of working with binary structures, we will consider their adjacency matrices. An ordered $\Sigma$-structure $\str S$ induces a square matrix, called the \emph{adjacency matrix} of $\str S$, whose rows and columns correspond to the elements of $\str S$, and where the entry at $(a,b)$ is the atomic type $\atp(ab)$. 
In the case of ordered graphs, there are five  possible atomic types, depending on whether $a=b,a<b$ or $a>b$, and on the adjacency between $a$ and $b$.
All matrices below have ordered sets of rows and columns.

\begin{definition}[$(k,t)$-mixed minor]
	Fix $k,t\in\N$.
A \emph{$(k,t)$-mixed minor} in a matrix $M$ is a convex partition $\cal R$ of its rows
	and a convex partition $\cal C$ of its columns,
	each with $t$ parts, 
	such that	 
	 for all $R\in \cal R$ and $C\in \cal L$, the submatrix $M[R\times C]$ of $M$ has at least $k$ different columns or at least $k$ different columns.
\end{definition}

Note  that  if a class $\CC$ is $(k+1,t)$-simple  then the adjacency matrix of every $\str S\in \CC$ has no $(k,t)$-mixed minor.

\medskip

We will use the following result due to Marcus and Tardos, concerning a related notion, of \emph{grid minors} in $0$-$1$ matrices with many $1$'s.
\begin{definition}[grid minor]
	 A \emph{$t$-grid minor} in a $0$-$1$ matrix $M$ is a convex partition $\cal R$ of its rows
	 and a convex partition $\cal C$ of its columns, each with $t$ parts, such that
	for all $R\in \cal R$ and $C\in\cal C$,
	the submatrix $M[R\times C]$ has an entry $1$.
\end{definition}

\begin{theorem}\cite{MARCUS2004153}
	For every $t\in\N$ there is a constant $c_t$ such that 
	 every $m\times n$ $0$-$1$ matrix $M$ 
	with at least $c_t\cdot \max(m,n)$ entries $1$  has a $t$-grid minor. 
\end{theorem}
It is known~\cite{DBLP:journals/corr/CibulkaK16} that $c_t=2^{\Oof(t)}$.

\medskip
A row interval $R$ and a column interval $C$ are \emph{homogeneous} if the submatrix  $M[R\times C]$ 
of $M$ has all entries equal to each other. 
We will consider pairs consisting of a convex partition $\cal R$ of the rows and of a convex partition $\cal C$ of the columns.
Such a pair $(\cal R,\cal C)$ has \emph{red-degree} at most $d$ if for each row part $R\in\cal R$ there are at most $d$ column parts $C\in\cal C$ such that $R$ and $C$ are not homogeneous, and symmetrically, for each column part $C\in\cal C$ there are at most $d$ row parts $R\in\cal R$ such that $R$ and $C$ are not homogeneous.

We prove the following, asymmetric variant of Theorem~\ref{thm:k-t-simple}.

\begin{proposition}\label{prop:two-steps}	
	Fix $k$ and $t\in\N$. There is a constant $d=2^{2^{\Oof(tk)}}$ 
	such that for every rectangular matrix $M$
	with no $(k,t)$-mixed minor there is a
	sequence
	\begin{align}\label{seq:step1}
		(\cal R_0,\cal C_0),(\cal R_1,\cal C_1),\ldots, (\cal R_l,\cal C_l)
		\end{align}
	of pairs of partitions (of the rows and columns of $M$, respectively) that is maximal under refinement, and consists of pairs of red-degree at most $d$.
\end{proposition}
Here, maximality under refinement means:
\begin{itemize}
	\item $\cal R_0$ and $\cal C_0$ are both partitions into singletons, 
	\item $\cal R_l$ and $\cal C_l$ are both partitions with one part, and 
	\item if $1\le i<l$ then
	$(\cal R_{i+1},\cal C_{i+1})$ 
	either $\cal R_{i+1}$ is obtained  by merging two parts in $\cal R_i$ into one and leaving $\cal C_{i+1}=\cal C_i$, or vice-versa.
\end{itemize} 

Note that the partitions in Proposition~\ref{prop:two-steps}	  are not necessarily convex.

\medskip

We proceed to the proof of Proposition~\ref{prop:two-steps}.
\textbf{For the rest of Section~\ref{sec:grid-and-tww}, fix constants $k$ and $t$ and a matrix $M$ with no $(k,t)$-mixed minor.}
We proceed in two steps.
The first step is:
\begin{lemma}\label{lem:step1}
	There are constants $b=2c_{tk}$ and $c=k^{c_{tk}}$ depending on $k$ and $t$ only such $M$ has a maximal (under refinement) sequence 
	\begin{align}\label{seq:step1}
	(\cal R_0,\cal C_0),(\cal R_1,\cal C_1),\ldots, (\cal R_l,\cal C_l)
	\end{align}
	of pairs of convex partitions of the rows and columns of $M$ such that the following conditions hold for $0\le i\le l$:
	\begin{itemize}
		\item 
		for each row interval $R\in\cal R_i$ there is a set $B\subset \cal C_{i}$ of at most $b$ column intervals such that $M$ has at most $c$ distinct rows in $R\times \bigcup (\cal C_{i}-B)$, and
		\item symmetrically for columns and rows exchanged.
	\end{itemize}
	\end{lemma}
To construct the sequence~\eqref{seq:step1}, we start with the partitions $\cal R_0$ and $\cal C_0$ of the rows and columns into singletons, and then proceed by repeatedly merging either two 
adjacent parts of the row partition 
or of the column partition
so that the condition in the lemma above is maintained.
This is done in a greedy way, and we show that if the process is blocked at some point, then $M$ must have a $(k,t)$-mixed minor, by the Marcus-Tardos result.

The key lemma is:
\begin{lemma}\label{lem:invariant}
	Let $\cal R$ be a convex partition of 
	the rows and $\cal C$ be a convex partition of the columns of $M$, such that $2|\cal R|\ge |\cal C|$.	
	Then for some row interval $R\in \cal R$ there is a set 
	$B\subset \cal C$ of at most $2c_{tk}$ column intervals such that $M$ has at most $k^{c_{tk}}$ distinct rows in $R\times \bigcup (\cal C-B)$.
	
\end{lemma}
\begin{proof}
	Fix a row interval $R\in\cal R$.
	A \emph{minimal bad interval} for $R$ is an inclusion-minimal interval of columns $I$ such that $M[R\times I]$ contains at least $k$ distinct rows. The first column in a minimal bad interval for $R$ is a \emph{bad column} for $R$.
	
	Consider the $\cal R\times \cal C$ matrix $N$ 
	whose $(R,C)$-entry is $1$ if $C$ contains a bad column for $R$, and $0$ otherwise.
	By the Marcus-Tardos theorem, at least one of the following holds:
	\begin{enumerate}[label=(\emph{\alph*})]
		\item some row $R\in\cal R$ in $N$ has less than $2c_{tk}$ many $1$'s, or
		\item $N$ has a $tk$-grid minor.
	\end{enumerate}
	
Assume  case (\emph{a}) holds.
Let $B\subset \cal C$ be the set of columns $C$ such  that 
the $(R,C)$-entry is a $1$. By assumption, $|B|<2c_{tk}$.
Consider a maximal interval in $I\subset \cal C$ not containing an element of $B$. Then $I$ does not contain a minimal bad interval for $R$. In particular, $M[R\times I]$ has fewer than $k$ distinct rows. Since $|B|<2c_{tk}$, there are at most $2c_{tk}$ such intervals $I$. It follows that $M[R\times (\bigcup \cal C-B)]$ has fewer than $k^{2c_{tk}}$ distinct rows, yielding the conclusion of the lemma.
It remains to prove that case (\emph{b}) cannot hold.

\medskip
We show that case (\emph{b}) would yield  
 a $(k,t)$-mixed minor in $M$, contrary to the assumption. 
Assume that $N$ has a $tk$-grid minor.
This minor corresponds to a convex coarsening $\cal R'$ of $\cal R$ and a convex coarsening $\cal C'$ of $\cal C$ such that for all $R\in\cal R'$ and $C\in \cal C'$,   a bad column for $R$ can be found in $C$.
Moreover, $\cal C'$ and $\cal R'$ have $t k$ parts each.

Let $\cal C''$ be the coarsening of $\cal C'$ obtained by grouping parts of $\cal C'$ by groups of size $k$, so that each interval in $\cal C''$ is a union of $k$ intervals in $\cal C'$.
Then $\cal C''$ has $t$ parts.

Fix a row interval $R\in\cal R'$ and 
a column interval $C\in\cal C''$.
We claim that $M[R\times C]$ has at least $k$ different columns.

By definition of $\cal C''$, there are $k$ columns $c_1<\ldots<c_k$ in $C$ such that for each $1\le i\le k$,
column $c_i$ is bad for some row interval $R_i\in \cal R$ that is contained in $R$. Let $I_i$ be the corresponding bad interval that starts at column $c_i$.
If some $I_i$ is contained in $C$ then we are done, since $I_i$ witnesses that $M[R_i\times I_i]$ has at least $k$ different columns, so even more so $M[R\times C]$.

Otherwise, every interval $I_i$ starts at $c_i\in C$ and extends beyond $C$. We claim that then the columns $c_1,\ldots,c_k$ of $M[R\times C]$ are pairwise distinct. To this end, pick $1\le i<j\le k$; we show that the columns $c_i$ and $c_j$ 
differ already within $M[R_i\times C]$.

Note that $c_j\in I_i$ since $I_i$ starts at $c_i$ and extends beyond $C$, whereas $c_i<c_j\in C$.
Since $I_i$ is a minimal bad interval for $R_i$, the columns $c_i$ and $c_j$ of $M[R_i\times C]$ differ -- otherwise $M[R_i\times (I_i\setminus \set {c_i})]$ would have the same set of columns as $M[R_i\times I_i]$, contradicting minimality of $I_i$.

This proves that $M[R\times C]$ contains at least $k$ distinct columns, for all $R\in \cal R'$ and $C\in\cal C''$.
Coarsening $\cal R'$ by grouping by $k$ parts, we get a convex partition $\cal R''$ into $t$ parts. The partitions $\cal R''$ and $\cal C''$ witness that $M$ has a $(k,t)$-mixed minor, contradicting the assumption. Hence, case (\emph{b}) cannot hold.
\end{proof}


	\begin{proof}[Proof of Lemma~\ref{lem:step1}]
We construct a sequence $(\cal R_0,\cal C_0),(\cal R_1,\cal C_1),\ldots, (\cal R_i,\cal C_i)$ of pairs of convex partitions,  satisfying the  two conditions in the lemma, as follows. $\cal R_0$ is the partition of the rows into singletons, and $\cal C_0$ is the partition of the columns into singletons. For $i>0$, construct $(\cal R_{i},\cal C_i)$ from $(\cal R_{i-1},\cal C_{i-1})$ as follows. Assume that $|\cal R_{i-1}|\ge |\cal C_{i-1}|$, the other case being symmetric. Group the parts of $\cal R_{i-1}$ by two,
yielding a partition $\cal R$ with $2|\cal R|\ge |\cal C_{i-1}|$. Apply Lemma~\ref{lem:invariant},
yielding some row interval $R\in \cal R$ for which there is a set $B\subset \cal C_{i-1}$ of at most $2{c_{tk}}$ column intervals such that $M$ has at most $k^{c_{tk}}$ distinct rows in $R\times \bigcup (\cal C_{i-1}-B)$. The row interval $R\in\cal R$ corresponds to two adjacent row intervals in $\cal R_{i-1}$. Let $\cal R_i$ be obtained from $\cal R_{i-1}$ by replacing those two adjacent row intervals by their union $R$, and let $\cal C_i=\cal C_{i-1}$. 
The above construction terminates once both $\cal R_i$ and $\cal C_i$ have one part each. 
	\end{proof}

In the second step, we improve the sequence~\eqref{seq:step1}, yielding a sequence of pairs of \emph{not necessarily convex} partitions
with the properties described in the following reformulation of Proposition~\ref{prop:two-steps}.
Below, $b=2c_{tk}$ and $c=k^{c_{tk}}$ are the constants from Lemma~\ref{lem:step1}.
\begin{lemma}\label{lem:step2}
	There
	is a maximal (under refinement) sequence 
	\begin{align}\label{seq:step2}
	(\cal R_0'',\cal C_0''),(\cal R_1'',\cal C_1''),\ldots, (\cal R_l'',\cal C_l'')
	\end{align}
	of pairs of partitions of the rows and columns of $M$ of red-degree at most $d=2bc^2$.
\end{lemma}

\begin{proof}
We start with the sequence~\eqref{seq:step1} obtained from Lemma~\ref{lem:step1}.

Fix $0\le i\le l$ and a row interval $R\in \cal R_i$.
Let $B\subset \cal C_i$ be as in the statement of Lemma~\ref{lem:step1}, with $|B|\le b=2 c_{tk}$. Then $M$ has $c=k^{c_{tk}}$ distinct rows in $R\times \bigcup(\cal C_i-B)$.
Partition $R$ into at most $c$ parts, where two rows $r,r'\in R$ are in the same part if they coincide on each column in $\cal C_i-B$. 

By doing this for every row interval $R$, we obtain a refinement $\cal R_i'$  of $\cal R_i$, in which every interval in $\cal R_i$ is partitioned into at most $c$ parts in $\cal R_i'$.
We proceed symmetrically with the columns, yielding a refinement $\cal C_i'$ of $\cal C_i$.

We show that the obtained sequence of pairs of partitions 
satisfies the conditions of Lemma~\ref{lem:step2}, apart from being a maximal chain under refinement. 

By construction of $\cal R_i'$ and $\cal C_i'$, we get:
\begin{claim}
	For each row part $R\in\cal R_i'$ there is a set $B'\subset \cal C_{i}'$ of at most $t'=b\cdot c$ column parts such that for all $C\in \cal C_{i}'-B'$, all entries in $M[R,C]$ are equal.
	Symmetrically for columns and rows exchanged.
\end{claim}

Say that a partition $\cal P_1$ is an $r$-\emph{coarsening} of a partition $\cal P_2$ 
if every part of $\cal P_1$ is union of at most $r$ parts of $\cal P_2$
\begin{claim}$\cal R_i'$ is a $2c$-coarsening of 
	$\cal R_{i-1}'$ and $\cal C_i'$ is a $2c$-coarsening of $\cal C_{i-1}'$. 
\end{claim}
This claim, as well as the rest of the proof, proceed as in~\cite[Proof of Theorem 10]{bonnet2020twin}. In particular, applying Lemma 8 from there yields the sequence \eqref{seq:step2} of red-degree $d=2c\cdot bc$.
\end{proof}

Lemma~\ref{lem:step2} yields Proposition~\ref{prop:two-steps}.
\medskip

To finish the proof of Theorem~\ref{thm:k-t-simple},
we observe that the sequence~\eqref{seq:step2} can be made so that the partitions in each pair are equal, just as in~\cite[Proof of Theorem 14]{bonnet2020twin}.

Observe that the construction above is effective, and yields the following:
\begin{corollary}\label{cor:algo}
	Fix a binary signature $\Sigma$.
	There is an algorithm which,
given an ordered binary structure $\str S$ and numbers $k,t\in\N$,
terminates in time polynomial in $|\str S|$ with one of the following outcomes:
\begin{itemize}
	\item either a contraction sequence  of $\str S$ of red-degree bounded by $2^{2^{\Oof(kt)}}$, or 
	\item a $(k,t)$-mixed minor in the adjacency matrix of $\str S$.
\end{itemize}
\end{corollary}

\section{Approximating twin-width}\label{sec:algorithms}
In this section we assume that 
we already know that every hereditary  class $\CC$ of binary structures of unbounded twin-width contains arbitrarily large regular semigrids.
This will be proved in the following sections using entirely different tools, but in this section we focus only on the algorithmic application of this fact. We prove:

\begin{theorem}[\ref{thm:approx}]
    There is a computable function $f\from\N\to\N$ and an  algorithm which inputs an ordered binary structure $\str S$, and computes in time polynomial in $|\str S|$ a contraction sequence of $\str S$ of  red-degree at most $f(\tww(\str S))$.
\end{theorem}
\begin{corollary}\label{cor:fpt}
    Fix a binary signature $\Sigma$. There is a computable function $g\from\N\times \N\to\N$, a constant $c\in\N$ and an algorithm which given a binary ordered structure $\str S$ and a first-order sentence $\phi$,
    determines 
    whether $\str S\models\phi$
    in time $g(|\phi|,\tww(\str S))\cdot |\str S|^c$.
    In other words, model-checking first-order logic is fixed-parameter tractable on ordered, binary structures, with the parameter being the formula $\phi$ and the twin-width of the input structure.
\end{corollary}
\begin{proof}
    Follows from Theorem~\ref{thm:approx} and~\cite[Theorem 1]{tww1} according to which a there is an algorithm which given a first-order sentence $\phi$ and a binary structure $\str S$ together with a contraction sequence of red-degree at most $d$ determines if $\str S\models\phi$ in time $h(d,\phi)\cdot |\str S|^c$, for some computable function $h$ and constant $c\in\N$.
\end{proof}

\newcommand{\maxsemigrid}{\mathrm{max\text{-}semigrid}}
\newcommand{\simplicity}{\mathrm{simplicity}}

Fix  an ordered binary structure $\str S$ and define the following two parameters:
\begin{itemize}
    \item $\maxsemigrid(\str S)$: the maximal number $n$ such that some regular $n\times n$-semigrid is an induced substructure of $\str S$,
    \item $\simplicity(\str S)$: the least number $k$ such that $\str S$ is $(k,k)$-simple.
\end{itemize} 
Clearly, both functions are computable.

For  two  functions $p,q$
mapping $\Sigma$-structures to numbers,
write $p\preceq q$ if there is a computable function $\alpha\from\N\to\N$ such that $p(\str S)\le \alpha(q(\str S))$ for all finite ordered $\Sigma$-structures $\str S$.

\begin{lemma}
    We have:
     \[\tww\preceq \simplicity\preceq \maxsemigrid\preceq \tww.\]
\end{lemma}
\begin{proof}
    The first inequality is by Theorem~\ref{thm:k-t-simple}.
    
    We show the second inequality.
    Fix $k\in\N$.
    Let $\cal F_{k}$ denote the family of all  minimal (under induced substructure)
    structures $\str S$ which are not $(k-1,k-1)$-simple. In particular,  every $\str S$ with $\simplicity(\str S)\ge k$ contains at least one structure in $\cal F_k$ as an induced substructure.

    \begin{claim}The family $\cal F_{k}$ is finite (up to isomorphism) and its representatives are effectively computable, given $k$.
    \end{claim}
    \begin{proof}
        A $(k,t)$-mixed minor in a matrix $M$ is exhibited by a set $R$ of 
         at most $kt^2$ rows and a set $C$ of at most $kt^2$ columns 
        (in each of the $t^2$ zones leave $k$ distinct rows and $k$ distinct columns).
        If $M$ is the adjacency matrix of a binary ordered structure $\str M$,
        then the rows and columns are indexed by the same set $\str M$, so $R\cup C\subset \str M$.
        Then the substructure of $\str M$ induced by $R\cup C$ has a $(k,t)$-mixed minor and has size $|\str M|\le 2kt^2$. 
        
        This shows that every structure that is not  $(k,t)$-simple contains an induced substructure of size at most $2kt^2$ which is not $(k,t)$-simple. Hence, the family of minimal structures that are not $(k,t)$-simple is finite and can be effectively enumerated by checking all structures of size $2kt^2$.
    \end{proof}
    
    Denote
    \[\alpha(k):=\min\setof{\maxsemigrid(\str S)}{\str S\in \cal F_{k}}.\]
    The function $\alpha$ is computable by the claim above. By construction, if $\simplicity(\str S)=k$ then $\str S$ has some structure in $\cal F_k$ as an induced substructure, and hence also some semigrid of size $\alpha(k)$. This proves $\maxsemigrid(\str S)\ge \alpha(\simplicity(\str S))$.
    Hence, $\simplicity(\str S)\le \alpha^{-1}(\maxsemigrid(\str S))$, proving $\simplicity\preceq \maxsemigrid$.

    \medskip
    Finally, we prove $\maxsemigrid\preceq \tww$. 

    By Lemma~\ref{lem:universal-interpretation}  and Lemma~\ref{lem:map-semigrids} there is a fixed transduction $T$ such that for any regular $k\times k$-semigrid $\str K$ over the signature $\Sigma$, the set of outputs $T(\str K)$ contains the set $B_k$ of all bipartite graphs with two parts of size $k$ The transduction $T$ nondeterministically guesses a subset of a given $k\times k$-semigrid $\str K$ corresponding to some set $S\subset\set{1,\ldots,k}\times\set{1,\ldots,k}$ and guesses the scheme $\sigma$ (now in the signature of ordered graphs) and then applies the interpretation $I$ from Lemma~\ref{lem:map-semigrids} (converting the regular semigrid over $\Sigma$ to a graph semigrid) followed by $I_\sigma$ from Lemma~\ref{lem:universal-interpretation}, converting a regular $k\times k$-semigrid $G$ with a chosen set $S\subset\set{1,\ldots,k}\times\set{1,\ldots,k}$ into the bipartite graph corresponding to $S$.

    We can furthermore assume that $T(\str S')\subset T(\str S)$ for any $\str S$ and its induced substructure $\str S'$, as $T$ may first restrict the domain of the input structure using a unary predicate.

By~\cite[Theorem 39]{bonnet2020twin}, there is a computable function $\beta$ (depending on $T$) such that $\tww(T(\str S))\le \beta(\tww(\str S))$ (in the left-hand side $T(\str S)$ is a set $X$ of structures, so $\tww(X)=\max_{\str S'\in X}\tww(
    \str S'))$.

Suppose $\maxsemigrid(\str S)=k$.
    Then $\str S$ contains a $k\times k$-regular semigrid $\str K$,
    so $T(\str S)\supseteq B_k$ and:
\[\beta(\tww(\str S))\ge \tww(T(\str S))\ge \tww(B_{k})=:\gamma(k),\]
where $\gamma(k)$ is a computable function which is unbounded, as the class of all bipartite graphs has unbounded twin-width.

Hence $k\le \gamma^{-1}(\beta(\tww(\str S)))$,
proving $\maxsemigrid\preceq\tww$.
\end{proof}

We now prove Theorem~\ref{thm:approx}.

\begin{proof}[Proof of Theorem~\ref{thm:approx}]
    Given a structure $\str S$, the algorithm proceeds as follows.
    For each $k=1,2,\ldots,n$,
     run the algorithm from Corollary~\ref{cor:algo}, until encountering the smallest number $k$   for which the algorithm returns a contraction sequence of red-degree $d=2^{2^{\Oof(k^2)}}$.
This contraction sequence is the result of our algorithm. We now provide an upper bound on $d$ in terms of $\tww(\str S)$.

We have that $\simplicity(\str S)\ge k$,
since the algorithm did not succeed for the value $k-1$.
As $\simplicity\preceq \tww$, this shows that $\tww(\str S)\ge \alpha(k)$ for some computable function, and hence $k\le \alpha^{-1}(\tww(\str S))$, and so $d=2^{2^{\Oof(k^2)}}$ 
is bounded by a computable function of $\tww(\str S)$.
\end{proof}

\section{Model-theoretic characterisations}
\label{sec:equivalences}
In this section, we present our model-theoretic characterisations of classes of bounded twin-width,
as well as prove more general results concerning arbitrary classes of structures.
Among others, this will prove Theorem~\ref{thm:core} and Theorem~\ref{thm:summary-mt0}.

We start with introducing the relevant notions from model theory.

\subsection{Monadically NIP classes}
\paragraph*{Monadically NIP classes}
A formula $\phi(\tup x,\tup y)$ defines in a structure $\str S$ a bipartite graph with parts $\str S^{\tup x}$ and $\str S^{\tup y}$,
in which two tuples $\tup a\in\str S^{\tup x}$ and $\tup b\in\str S^{\tup y}$ are adjacent if $\phi(\tup a,\tup b)$ holds in $\str S$.
\begin{definition}\label{def:nip}
	A class of structures $\CC$ is  \emph{NIP}
if for every formula $\phi(\tup x,\tup y)$ there is a finite bipartite graph $H$ which does not occur as an induced subgraph of $\phi(\str S)$, for all $\str S\in\CC$.
\end{definition}
For example, the class of all cliques is NIP, whereas the class $\CC$ of all finite graphs is not, since the edge formula $E(x,y)$ may define arbitrary bipartite graphs in $\CC$.
As another example, consider  the class of all cliques with each edge subdivided: an additional vertex placed  in the middle  of each edge. This class is NIP -- intuitively, one cannot define much more in such a graph than in a usual clique.
However, this changes if one can remove some of the newly inserted vertices and the edges incident to them -- this way, we can obtain any graph $G$ with each edge subdivided, and hence, very complex graphs can be obtained.
The same is true if instead of removing the vertices, we can color some of the vertices by adding unary predicates, and use those predicates in the formulas.

\begin{definition}[Monadic NIP]\label{def:NIP}
	A class of structures $\CC$ is \emph{monadically NIP} if every expansion $\CC'$ of $\CC$ by unary predicates is NIP (here $\CC'$ is such that every structure from $\str S'\in \CC'$ is obtained from some structure in $\str S\in \CC$ by equipping $\str S$ with arbitrarily many unary predicates). 
	\end{definition}

In particular, if $\CC$ is monadically NIP then the hereditary closure of $\CC$ is also monadically NIP, since removing vertices can be simulated by coloring them.
The following fact essentially says that it is enough to consider formulas $\phi(x,y)$ with just two variables when considering monadic NIP.


\begin{definition}[Finitely monadically NIP]
	A class of structures $\CC$ is \emph{finitely monadically NIP}
	(\emph{fmNIP})
	if every expansion of $\CC$ by finitely many unary predicates that are interpreted in structures from $\CC$ as \emph{finite sets} is NIP.
\end{definition}
The following  lemma is immediate.
\begin{lemma}\label{lem:fin-ip}\label{lem:fin-mnip-trans}
	\begin{enumerate}
		\item The class of all finite graphs is not NIP.		
		\item A class of finite structures is monadically NIP if and only if it is fmNIP.
	\end{enumerate} 
\end{lemma}

The following proposition is a restatement of a result of the first author~\cite{nip-dim-1}, 
see Appendix~\ref{app:nip} for more details.
\begin{proposition}[\cite{nip-dim-1}]\label{prop:nip-2}
	The following conditions are equivalent for a class of structures $\DD$:
	\begin{itemize}
		\item $\DD$ is not NIP,
		\item there is a formula $\phi(x,y;\tup z)$ 
		such that for every $n$ there is a structure $\str M\in\DD$ and a tuple $\tup c\in\str M^{\tup z}$ such that $\phi(\str M;\tup c)\subset \str M^{2}$ defines a binary relation of VC-dimension at least $n$;
	\end{itemize}
\end{proposition}

\begin{corollary}
	\label{cor:nip}
	The following conditions are equivalent for 
	a class of structures $\CC$:
	\begin{enumerate}				
		\item $\CC$ does not transduce the class of all finite graphs;
		 \item $\CC$ is not monadically NIP.
	\end{enumerate}
\end{corollary}
Corollary~\ref{cor:nip} also follows from results of Baldwin and Shelah~\cite{DBLP:journals/ndjfl/BaldwinS85}.

Together with Fact~\ref{fact:properties}  this gives:
\begin{corollary}\label{cor:tww-monNIP}
	Every class of bounded twin-width is monadically NIP.
\end{corollary}

Finitely monadically NIP classes do not define large grids (cf. Def.~\ref{def:grids}):
\begin{lemma}\label{lem:no-large-grids}
	If $\CC$ is fmNIP then $\CC$ does not define large grids.
\end{lemma}
\begin{proof}Observe that if $\phi(\tup x,\tup y,z)$ defines a grid $(A,B,C)$ in a structure $\str S$ then for every binary relation $R\subset A\times B$ there is a unary predicate $U\subset C$ such that the formula
	 $$\psi(\tup x,\tup y):=\exists z.\phi(\tup x,\tup y,c)\land U(z)$$ defines $R$, so that for all $\tup a\in A$ and $\tup b\in B$, 
	 \[(\tup a,\tup b)\in R\quad\Leftrightarrow\quad \str S\models \psi(\tup a,\tup b).\]
	 
	 In particular, 
	if $|A|=n$ and $|B|=2^n$ then for a suitably chosen $U\subset C$,
	the formula $\psi$ defines a relation of VC-dimension $n$ between $A$ and $B$.

Consider the expansion of $\CC$ by a unary predicate $U$ interpreted as finite sets.
Then the formula $\psi(\tup x,\tup y)$ above has unbounded VC-dimension on $\CC$, proving that $\CC$ is not fmNIP.
\end{proof}

\subsection{1-dimensionality}
We now introduce a wholly model-theoretic notion characterizing bounded twin-width.
For this, we first recall some basic notions from model theory.

By  a \emph{model} we mean a structure which is typically infinite, as opposed to the structures considered earlier, which were typically finite.
We give a brief account of basic notions from model theory in Appendix~\ref{app:model-theory}, although they are not needed to follow the main text below.

The \emph{elementary closure} of a class of structures
$\CC$ is the class of all models $\str M$ that satisfy every sentence $\phi$ that holds in all structures $\str S\in\CC$. 
In particular, if $\CC$ does not define large grids,
then neither does its elementary closure. This is because for any fixed $n\in\N$ the existence of an $n\times n$-grid defined by a fixed formula $\phi(\tup x,\tup y,z)$ can be expressed by a first-order sentence $\phi'$
which existentially quantifies $(|\tup x|+|\tup y|)\cdot n+n^2$ variables, corresponding to sets $A,B,C$ of $\tup x$-tuples, $\tup y$-tuples and single vertices,
and then checks that $\phi$ defines a bijection between $A\times B$ and $C$.

By the compactness theorem (cf. Thm.~\ref{thm:compactness}), if $\CC$ defines large grids, then its elementary closure contains a structure that defines a grid $(A,B,C)$
with $A$ and $B$ of arbitrarily large infinite cardinalities.


\begin{definition}[Elementary extension]
	Let $\str M,\str N$ be two models.
	Then $\str N$ is an \emph{elementary extension} of $\str M$,
	written $\str M\prec \str N$, 
	if the domain of $\str M$ is contained in the domain of $\str N$, and  for  every 
	formula $\phi(\tup x)$ and tuple $\tup a\in\str M^{\tup x}$ of elements of $\str M$,
	\[\str M\models\phi(\tup a)\text{\quad if and only if \quad}\str N\models \phi(\tup a).\]	
\end{definition}
 In other words, it doesn't matter if we evaluate formulas in $\str M$ or in $\str N$. In particular, $\str M$ and $\str N$ satisfy the same sentences.

 \medskip
 A \emph{formula $\phi(\tup x)$ with parameters} from $C\subset \str N$
is a formula using constant symbols denoting elements from $C$.
Such a formula can be evaluated in $\str N$
on a tuple $\tup a\in\str N^{\tup x}$, as expected.
Note that if $\str M\prec \str N$ and $\phi(\tup x)$ is a formula with parameters from $\str N$ and $\tup a\in\str M^{\tup x}$ then it is not necessarily the case that $\str M\models \phi(\tup a)$ if and only if $\str N\models \phi(\tup a)$, although this does hold for formulas with parameters from $\str M$.
\medskip
\begin{definition}[Independence]\label{def:independence}
	Let $\str M$ be a model and $\str N$ its elementary extension. For a tuple $\tup a\in \str N^{\tup x}$ and a set $B\subset \str N$ say that $\tup a$ is \emph{independent} from $B$ over $\str M$, denoted $\tup a\ind[\str M] B$,
	if for every formula $\phi(\tup x)$ with parameters  from  $B\cup \str M$ such that $\str N\models \phi(\tup a)$ there is some $\tup c\in \str M^{\tup x}$ such that $\str N\models \phi(\tup c)$.
\end{definition}
Abusing notation, if $B$ is enumerated by a tuple $\tup b$, 
then we may write $\tup a\ind[\str M]\tup b$.
We write $\nind[\str M]$ for the negation of the relation
$\ind[\str M]$.


\begin{example}\label{ex:order-ind}
	Let $\str N$ be $(\mathbb R,\le)$
    and let $\str M$ be 
    the union  of the open intervals 
    $]0,1[$ and $]8,9[$, equipped with the relation $\le$.
	Then $\str M\prec \str N$. This is easy to derive from the fact that $(\mathbb R,\le)$ has quantifier elimination, that is, every formula $\phi(\tup x)$ is equivalent to a quantifier-free formula. 
	\begin{figure}[h!]
		\centering
		\includegraphics[scale=0.35,page=1]{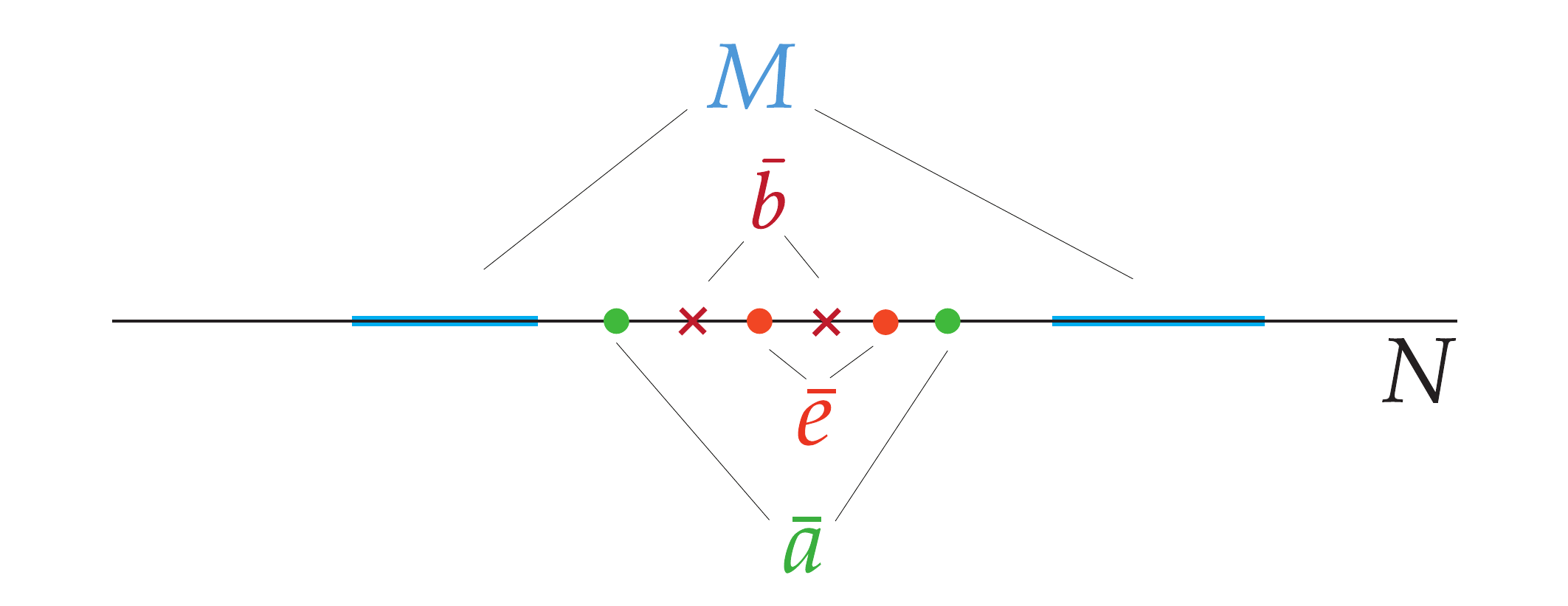}
		\caption{The structures $\str M\prec \str N$, a set $B$ and two tuples, $\tup a,\tup e$, with 
		$\tup a\ind[\str M]B$ and $\tup e\nind[\str M]B$.
		}\label{fig:ind1}
	\end{figure}
	Figure~\ref{fig:ind1} illustrates independence over  $\str M$. 
	


	

	

\end{example}

\begin{example}Let
	$\str M\prec \str N$.
	Then  $\tup a\ind[\str M]{}\str M$ for every $\tup a\in\str N^{\tup x}$
	(cf. Lemma~\ref{lem:satisfaction}).
\end{example}

\begin{definition}[1-dimensionality]\label{def:1-dim}
	A model $\str M$ is \emph{1-dimensional}
	if for every $\str M\prec \str N$, tuples  $\bar a,\bar b$ of elements of $\str N$ and $c\in\str N$ a single element, 
if $\tup a\ind[\str M]{}\tup b$ then  $\tup ac\ind[\str M]{}\tup b$ or $\tup a\ind[\str M]{}\tup bc$.
A class $\CC$ of structures is \emph{1-dimensional} if every model in the elementary closure of $\CC$ is 1-dimensional.
\end{definition}

\begin{example}
	Any total order $(X,\le)$ is 1-dimensional.
As an illustration, in the situation in Fig.~\ref{fig:ind1}, consider the tuples $\tup a,\tup b$ marked therein. Then $\tup a\ind[\str M]\tup b$.
Let $c\in\str N$. If $c$ belongs to the interval $]b_1,b_2[$
then $\tup a\ind[\str M]\tup bc$. Otherwise, $\tup ac\ind[\str M]\tup b$.
\end{example}

\begin{example}\label{ex:grid-2-dim}
	Let $\str N=(\R\times \R, \sim_1,\sim_2)$ where for $i=1,2$, the relation $\sim_i$ denotes equality of the $i$th coordinates.
	Let $\str M$ be the induced substructure of $\str N$ with domain $I\times I$ for some infinite subset $I\subset \str N$. Then $\str M\prec \str N$. In the situation depicted in Fig.~\ref{fig:ind2}, $a\ind[\str M]b$ but both $ac\nind[\str M]b$ and $a\nind[\str M]bc$. So $\str M$ is not $1$-dimensional. 

	\begin{figure}[h!]
		\centering
		\includegraphics[scale=0.35,page=2]{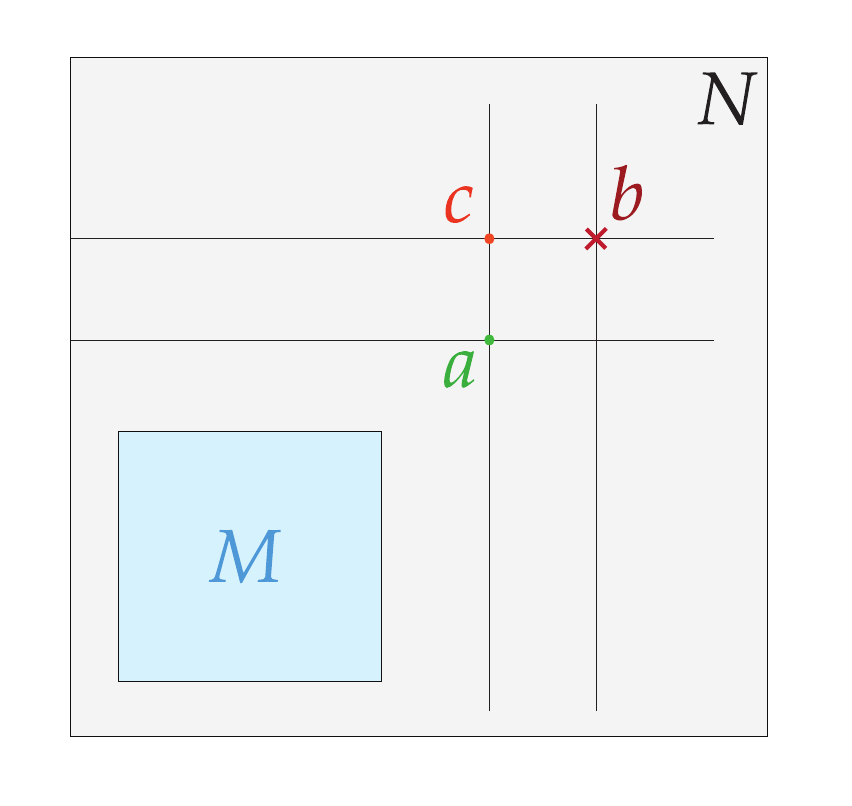}
		\caption{The structures $\str M\prec \str N$ and elements $a,b,c\in\str N$ with 
		$a\ind[\str M]b, ac\nind[\str M]b$ and $a\nind[\str M]bc$.
		}\label{fig:ind2}
	\end{figure}
\end{example}

The  following result  is essentially \cite[Lemma 2.2]{shelah:hanfnumbers}.

\begin{proposition}\label{prop:fs dichotomy}
	If a  model $\str M$  does not define large grids the it is 1-dimensional. More precisely,
	if there is a formula $\phi(\tup x,\tup y,z)$ with parameters from $\str M$ which is satisfied by $\tup a,\tup b,c$ and such that $\phi(\tup x,\tup b,z)$ and $\phi(\tup x,\tup b,c)$ are not satisfiable in $\str M$ then a boolean combination $\psi$ of instances of $\phi$ defines large grids in~$\str M$.	
\end{proposition}
See Appendix~\ref{app:fmNIP} for a proof.
In particular, every class $\CC$ that does not define large grids is 1-dimensional. 

\subsection{Regular classes}
We now provide a characterization in terms of the number of types. The definition below generalizes the definition of $(k,t)$-simplicity for ordered structures. 

Let $\phi(x,\tup y)$ be a formula and $\str S$ a structure.
A \emph{$\phi$-definable disjoint family} 
is a family $\cal R$ of pairwise disjoint of subsets of $\str S$, where for each $R\in\cal R$ there is $\tup b\in\str S^{\tup y}$ with $R=\phi(\str S,\tup b)$.
For example, if $\str S$ is an ordered structure and 
$\cal R$ is a convex partition of $\str S$, then $\cal R$ is a $\phi$-definable family of pairwise disjoint sets, for $\phi(x,y_1,y_2)=y_1\le x\le y_2$.

\begin{definition}[Regularity]\label{def:regular}
A  class $\CC$ of structures is 
\emph{regular}
if the following condition holds.
Let $\phi(x;\bar y)$, $\psi(x;\bar z)$ and $\theta(\bar u;\bar v)$ be $L$-formulas. Then there are natural numbers $t$ and $k$ such that for any $\str S\in \CC$ and any
$\phi$-definable disjoint family $\cal R\subset P(\str S)$  and $\psi$-definable disjoint family $\cal L\subset P(\str S)$ with  $|\cal R|=|\cal L|\ge t$ there are $R\in \cal R$ and $L\in \cal L$ with $|\Types[\theta](R/L)|\le k$.
\end{definition}
Note that $|\Types[\theta](R/L)|\le k$ implies 
$|\Types[\hat\theta](L/R)|\le 2^k$, where $\hat\theta(\tup v;\tup u)=\theta(\tup u;\tup v)$.
So in the definition of regularity, we could 
equivalently require
both $|\Types[\theta](R/L)|\le k$ and $|\Types[\hat\theta](L/R)|\le k$,
as $k$ can be increased appropriately.

It is easy to see that 
if $\CC$ is regular class of ordered graphs 
then $\CC$ is $(k,t)$-simple for $k,t$ obtained from the definition of regularity applied to the formulas $\phi=y_1\le x\le y_2$ and $\psi=z_1\le x\le z_2$ and $\theta$ being the edge relation. This also holds for arbitrary classes of ordered binary structures, as $|\Types[\Sigma](R/L)|\le \prod_{S\in \Sigma}|\Types[S](R/L)|$.

\subsection{Main model-theoretic result}

\begin{figure}
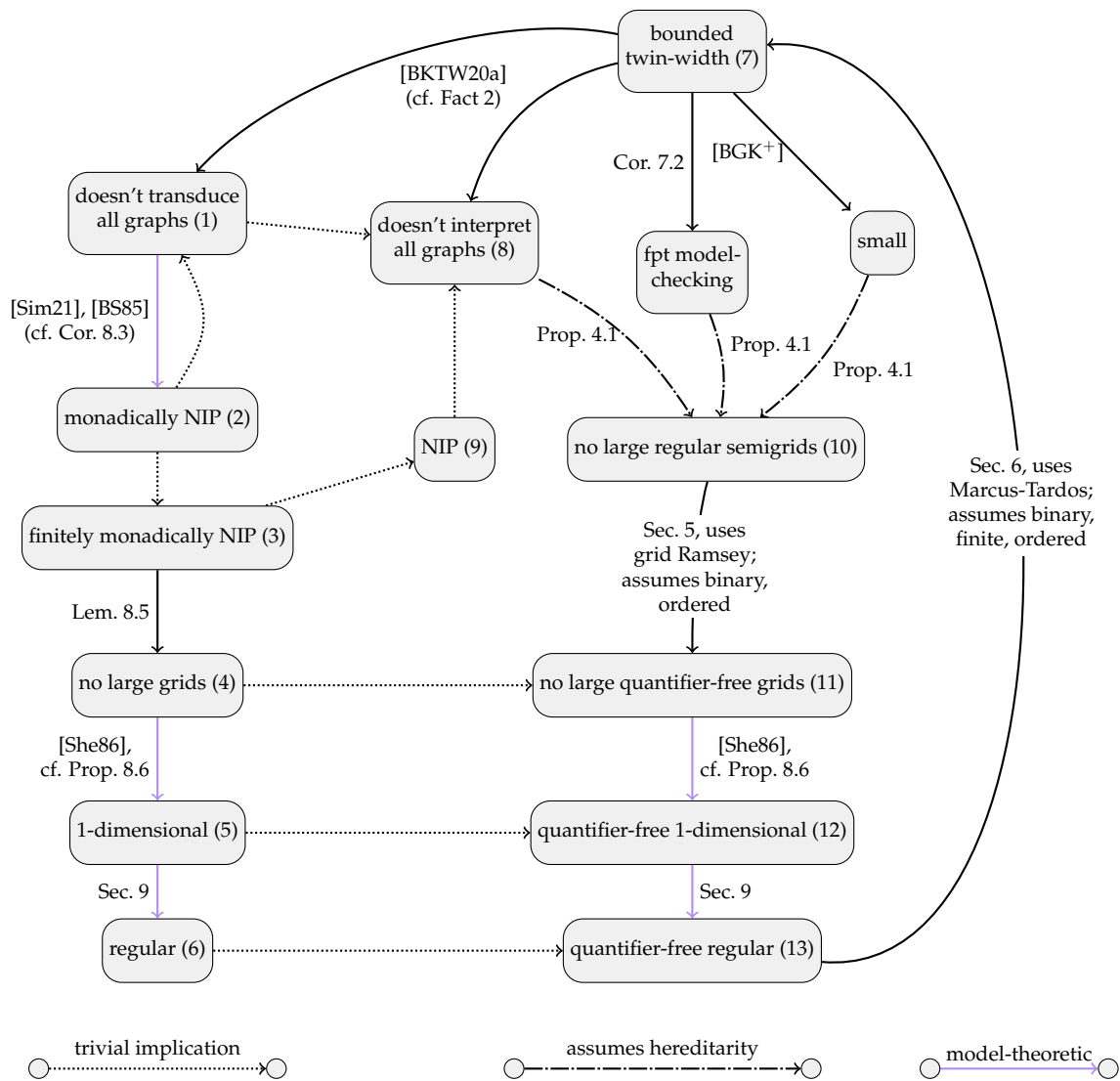

\begin{centering}
	\footnotesize
\ctikzfig{tikz}
\end{centering}\caption{Implications among the conditions in Theorems~\ref{thm:main2} and~\ref{thm:main3}. 
All implications apply to arbitrary classes of structures,
unless specified otherwise.
}\label{roadmap}
\end{figure}
We can finally state 
our main model-theoretic result. It extends Theorem~\ref{thm:summary-mt0}
stated in the introduction.
\begin{theorem}\label{thm:main2}
	Let $\CC$ be any class of structures over a relational signature and consider the following conditions:
	\begin{enumerate}
		
		\item\label{lit:all-graphs} $\CC$ does not transduce the class of all finite graphs,
		\item\label{lit:lin-monNIP} $\CC$ is monadically NIP,
		
		\item\label{lit:lin-fmNIP} $\CC$ is finitely monadically NIP,

		\item\label{lit:grid-or} $\CC$ does not define large grids,
		\item\label{lit:1-dim}$\CC$ is $1$-dimensional,
		
		\item\label{lit:regular}$\CC$ is regular.
	\end{enumerate}
	 Then the implications
	\eqref{lit:all-graphs}$\leftrightarrow$%
\eqref{lit:lin-monNIP}$\rightarrow$%
\eqref{lit:lin-fmNIP}$\rightarrow$%
\eqref{lit:grid-or}$\rightarrow$%
\eqref{lit:1-dim}$\rightarrow$%
\eqref{lit:regular} hold.
\end{theorem}

\begin{proof}
The implications~\eqref{lit:all-graphs}$\leftrightarrow$\eqref{lit:lin-monNIP}$\rightarrow$\eqref{lit:lin-fmNIP}
follow from Corollary~\ref{cor:nip} and Lemma~\ref{lem:fin-ip}.
Implication \eqref{lit:lin-fmNIP}$\rightarrow$\eqref{lit:grid-or} is by Lemma~\ref{lem:no-large-grids}.
The implication~\eqref{lit:grid-or}$\rightarrow$\eqref{lit:1-dim} is Proposition~\ref{prop:fs dichotomy}.
The implication \eqref{lit:1-dim}$\rightarrow$\eqref{lit:regular} is relegated to Section~\ref{sec:main-proof} below.
\end{proof}


For a class $\CC$ of ordered structures over a binary signature we may
additionally prove the equivalence of all the above with $\CC$ having bounded twin-width and furthermore
improve some of the conditions above to obtain quantifier-free formulas, as follows.
See Fig.~\ref{roadmap} for a diagram of all the various considered conditions and implications among them.
\begin{theorem}		\label{thm:main3}
	Let $\CC$ be a class of ordered structures over a binary signature $\Sigma$ and $\CC'$ be the hereditary class of all its finite induced substructures.
	Then the following conditions are all equivalent to each other and to the conditions~\eqref{lit:all-graphs}-\eqref{lit:regular}:
		\begin{enumerate}\setcounter{enumi}{6}
			 \item\label{lit:tww} $\CC'$ has bounded twin-width,
			 
			 \item\label{lit:all-graphs'}$\CC'$  does not interpret the class of all finite graphs,
			 
			 \item\label{lit:nip}$\CC'$ is NIP,
			\item\label{lit:semigrids} $\CC'$ does not contain arbitrarily large regular semigrids,

			 \item\label{lit:grid-or'} no quantifier-free formula defines large grids in $\CC$,
		
			 \item(\emph{quantifier-free 1-dim}.)\label{lit:1-dim'}		let $\str M\prec \str N$ be models in the elementary closure of $\CC$, $\tup a\in\str N^{\tup x},\tup b\in\str N^{\tup y}$ tuples with  $\tup a\ind[\str M]\tup b$, and $c\in\str N$; then there is no quantifier-free formula $\phi(\tup x,\tup y,z)$ such that $\str N\models\phi(\tup a,\tup b,c)$ and such that both $\phi(\tup x,\tup b,z)$ and $\phi(\tup x,\tup b,c)$ are not satisfiable in $\str M$,
		\item\label{lit:grid-and}($(k,t)$-\emph{simple}) there are some $k,t\in\N$ such that for every structure $\str S\in \CC$ and two convex partitions $\Ll$ and $\Rr$ of $\str S$ into $t$ parts, there are $L\in\cal L$ and $R\in \cal R$ such that $|\Types[\Sigma](L/R)|\le k$ and 
		 ${|\Types[\Sigma](R/L)|\le k}$.
	\end{enumerate}
	
	\end{theorem}


\begin{proof}
The implication \eqref{lit:tww}$\rightarrow$\eqref{lit:all-graphs'} is by Fact~\ref{fact:properties},
the implication
\eqref{lit:all-graphs'}
$\rightarrow$\eqref{lit:semigrids}
is by Corollary~\ref{cor:fails-sigma},
\eqref{lit:semigrids}
$\rightarrow$\eqref{lit:grid-or'} is by Theorem~\ref{thm:grids-binary}. 
The implication~\eqref{lit:grid-or'}$\rightarrow$\eqref{lit:1-dim'} 
is by the refined statement in Proposition~\ref{prop:fs dichotomy}.
The implication~\eqref{lit:1-dim'}$\rightarrow$\eqref{lit:grid-and} is proved in Section~\ref{sec:main-proof} below, together  with the implication \eqref{lit:1-dim}$\rightarrow$\eqref{lit:regular}.
The implication \eqref{lit:grid-and}$\rightarrow$\eqref{lit:tww} is by Theorem~\ref{thm:k-t-simple}. 
This proves the equivalence of the conditions~\eqref{lit:tww}-\eqref{lit:grid-and}.

Equivalence with the conditions \eqref{lit:all-graphs}-\eqref{lit:regular} follows since
\eqref{lit:tww}$\rightarrow$\eqref{lit:all-graphs} by Fact~\ref{fact:properties},
\eqref{lit:lin-fmNIP}$\rightarrow$\eqref{lit:nip}$\rightarrow$\eqref{lit:all-graphs'} are trivial,
and
 \eqref{lit:regular}$\rightarrow$\eqref{lit:grid-and} is immediate by the remark following Definition~\ref{def:regular},
as a convex partition of $\str S$ is a $\phi$-definable disjoint family of subsets of~$\str S$, for $\phi(x,y_1,y_2)=y_1\le x\le y_2$.

	\end{proof}
	

	Theorem~\ref{thm:main3} yields Theorem~\ref{thm:core}, and thus completes the proof of our two main results, Theorem~\ref{thm:intro} and  Theorem~\ref{thm:summary-mt0}.
It remains to prove implications \eqref{lit:1-dim}$\rightarrow$\eqref{lit:regular} and \eqref{lit:1-dim'}$\rightarrow$\eqref{lit:grid-and}.

\section{1-dimensional classes}\label{sec:main-proof}
In this section we prove that every 1-dimensional class is regular,  and prove the implication \eqref{lit:1-dim'}$\rightarrow$\eqref{lit:grid-and} in Theorem~\ref{thm:main3}.
We start with very briefly introducing some notions from model theory. A more comprehensive overview is presented in Appendix~\ref{app:model-theory}.

\subsection{On independence}

Let $\str M\prec \str N$.
For two sets $A,B\subset \str N$, write $A\ind[\str M] B$ if $\tup a\ind[\str M] B$ for all tuples $\tup a$ of elements of $A$.
Note that $A\ind[\str M]B$ if and only if $\tup a\ind[\str M]\tup b$ for all tuples $\tup a$ in $A$ and $\tup b$ in $B$.

\begin{lemma}\label{lem:type determination}
	Fix a formula $\theta(\tup u;\tup v)$.
	Suppose $A\ind[\str M]B$. Then for all $\tup b\in B^{\tup u}$, the type $\tp^\theta(\tup b/A)$ depends only on the type $\tp^\theta(\tup b/\str M)$. More precisely, 
	there is a function $f\from \Types[\theta](B/\str M)\to \Types[\theta](B/\str A)$ such that 
\[\tp^\theta(\tup b/A)=f(\tp^{\theta}(\tup b/\str M)).\]
\end{lemma}
\begin{proof}
	We show that if $\tp^\theta(\tup b/A)\neq \tp^\theta(\tup b'/A)$ then $\tp^\theta(\tup b/\str M)\neq \tp^\theta(\tup b/\str M)$.
	If $\tp^\theta(b/A)\neq \tp^\theta(b'/A)$ then there is some $\tup a\in A^{\tup v}$ such that $\theta(\tup b,\tup a)\triangle \theta(\tup b',\tup a)$. By $A\ind[\str M]B$ there is some $\tup m\in \str M^{\tup v}$ such that $\theta(\tup b,\tup m)\triangle \theta(\tup b',\tup m)$, implying $\tp^\theta(\tup b/A)\neq \tp^\theta(\tup b'/A)$.
\end{proof}

As $|\Types(B/\str M)|\le 2^{|\str M|}$, we get a bound on the number of types.
\begin{corollary}
	\label{cor:counting fs}
		Fix a model $\str M$ and its elementary extension $\str N$.
		Let $A,B\subset \str N$ be 
	such that $A\ind[\str M]B$. For any formula $\theta(\tup u;\tup v)$
	the set $\Types[\theta](B/A)$
	has cardinality at most $2^{|\str M|}$.
\end{corollary}

The above bound can be leveraged to give a finite bound on the size of a set of types, as follows.

\begin{lemma}\label{lem:finite st}
    Fix a model $\str M$ and its elementary extension $\str N.$ Let $\phi(x)$ and $\psi(x)$ be two formulas with parameters from $\str N$.
    Suppose that for every elementary extension $\str N'$ of $\str N$,
     \[\phi(\str N')\ind[\str M]\psi(\str N').\]
    Then for every formula
     $\theta(\tup u;\tup v)$ the set
    $\Types[\theta](\psi(\str N)/\phi(\str N))$ is finite.
    \end{lemma}
    \begin{proof}Fix a formula $\theta(\tup u;\tup v)$.
		Let $E_{\theta}(\tup u;\tup u')$ be the formula defining the equivalence relation on $\psi(\str N)^{\tup u}$ such that 
		\[E_\theta(\tup a,\tup a')\quad\Leftrightarrow\quad\tp^\theta(\bar a/\phi(\str N)) = \tp^\theta(\bar a'/\phi(\str N))\qquad\qquad\text{for  $\bar a,\bar a'\in\psi(\str N)^{\tup u}$}.\] More precisely, 
        \[E_\theta(\tup u;\tup u')\quad\equiv\quad \psi(\tup u)\land \psi(\tup u')\land \forall \tup v. \phi(\tup v)\implies (\theta(\tup u;\tup v)\iff \theta(\tup u';\tup v)).\]

        Suppose $E_\theta$ defines infinitely many classes in $\str N$.
        By compactness  there is an elementary extension $\str N'$ of $\str N$ such that  $E_\theta$ induces more than $2^{{|M|}}$ equivalence classes in~$\str N'$ (cf. Lemma~\ref{lem:many-classes}). 
By assumption, 
        \[\phi(\str N') \ind[\str M] \psi(\str N').\]
        By Corollary~\ref{cor:counting fs}, $\Types[\theta](\psi(\str N')/\phi(\str N'))$ has cardinality at most $2^{{|M|}}$, 
        which is in contradiction with the number of classes of $E_\theta$ in $\str N'$.
    \end{proof}
    
\begin{remark}\label{rem:triangle}
	To arrive at the conclusion of  Lemma~\ref{lem:finite st}, the assumption $\phi(\str N') \ind[\str M]\psi(\str N')$ can be weakened as follows:
	 for all $\tup b,\tup b'\in\str M^{\tup u}$ and	
	$\sigma(\tup v)=\theta(\tup b;\tup v)\triangle \theta(\tup b';\tup v)$, if $\sigma(\tup v)$ is satisfiable in $\phi(\str N')$ then it is satisfiable in $\str M$. The same proof as above works, since $\sigma$ is the only formula that appears in the proof of Lemma~\ref{lem:type determination}.
\end{remark}

\subsection{Proof of \eqref{lit:1-dim}$\rightarrow$\eqref{lit:regular} and \eqref{lit:1-dim'}$\rightarrow$\eqref{lit:grid-and}}
Let $\CC$ be a class which is not regular. We  show that $\CC$ is not 1-dimensional.
We first construct a structure $\str M$ in the elementary closure of $\CC$ that exhibits the lack of regularity in a handy way. 

\medskip

Say that a class $\CC$ is not regular \emph{as witnessed} by formulas 
$\phi(x,\tup y),\psi(x,\tup z)$, $\theta(\tup u,\tup v)$
if for all $k,t\in\N$ there are:
\begin{itemize}
	\item a structure $\str S\in\CC$,  
	\item a $\phi$-definable disjoint family $\cal R\subset P(\str S)$ with $|\cal R|\ge t$,
	\item a $\psi$-definable disjoint family $\cal L\subset P(\str S)$ with $|\cal L|\ge t$,
\end{itemize}
such that
\[|\Types[\theta](R/L)|>k\qquad\text{for all $R\in\cal R$ and $L\in\cal L$.}\]
If $\CC$ is not regular then there are some $\phi,\psi,\theta$ that witness it.
\begin{remark}\label{rem:qf-1}
	Suppose $\CC$ is a class of ordered  binary structures which is not $(k,t)$-simple for all $k,t\in\N$.
	That is,
	for all $k,t\in\N$ there is a structure $\str S\in \CC$ and two convex partitions $\Ll$ and $\Rr$ of $\str S$ into $t$ parts, there are no $L\in\cal L$ and $R\in \cal R$ such that $|\Types[\Sigma](L/R)|\le k$ and 
	${|\Types[\Sigma](R/L)|\le k}$.
	Then $\CC$ is not regular as witnessed by  $\phi(x,y_1,y_2)\equiv y_1\le x\le y_2$ and $\psi(x,z_1,z_2)\equiv z_1\le x\le z_2$ and $\theta(u,v)$ a quantifier-free $\Sigma$-formula with two variables.
\end{remark}

\begin{lemma}\label{lem:preparation}
	Suppose $\CC$ is not regular, as witnessed by formulas $\phi(x,\tup y),\psi(x,\tup z),\theta(\tup u,\tup v)$.
	Then there exist:
    \begin{itemize}
        \item  a structure $\str M$ in the elementary closure of $\CC$, 
        \item an elementary extension $\str N$ of $\str M$,
        \item tuples $\tup a_0,\tup a_1\in \str N^{\tup y}$ and $\tup b_0\in\str N^{\tup z}$,
    \end{itemize}
    such that the following properties hold:
    \begin{enumerate}
        \item  $\tp(\tup a_0/\str M)=\tp(\tup a_1/\str M)$,
        \item\label{it:many-types} 
        the set $\Types[\theta](B/A)$ is infinite, where $A=\phi(\str N,\tup a_1)$ and $B=\psi(\str N,\tup b_0)$,
        \item\label{it:ind-base} $\tup a_1 \ind[\str M]\tup a_{0}\tup b_{0}$,
        \item  $\phi(x;\tup a_0)\land \phi(x,\tup a_1)$ has no solution in $\str N$,
        \item $\psi(x;\tup b_0)$ has no solution in $\str M$.
    \end{enumerate}    
\end{lemma}

The proof of the lemma is a  standard application of basic tools from model theory: compactness, (mutually) indiscernible sequences and Morley sequences, which are recalled in Appendix~\ref{app:model-theory}. The proof of Lemma~\ref{lem:preparation} is in Appendix~\ref{app:preparation-lemma}. Using the lemma, we now show that $\CC$ is not $1$-dimensional.



We use the notation from Lemma~\ref{lem:preparation}.
Let $\str M$ and $\str N$ be as in Lemma~\ref{lem:preparation}. 
By (2) and  Lemma~\ref{lem:finite st} there is an elementary extension $\str N'$ of $\str N$ 
and tuples $\tup a\in \phi(\str N',\tup a_1)$,
$\tup b\in\psi(\str N',\tup b_0)$
such that $\tup a\nind[\str M]\tup b$.
We show that if $\str M$ is $1$-dimensional then
$\tup a_1\tup a\ind[\str M]\tup a_0\tup b_0\tup b$, implying $\tup a\ind[\str M]\tup b$, contrary to what was just stated.



\begin{claim}\label{cl:2}
	Let $a\in \phi(\str N',\tup a_1)$. Then $\tup a_1\nind[\str M]\tup a_0 a$.
\end{claim}
\emph{Proof}:
Consider the formula 
\begin{align}\label{eq:zeta}
\zeta(\tup y;a,\tup a_0) := \phi(a;\tup y)\land \neg\exists x. \phi(x; \bar y)\land \phi(x; \bar a_0).
\end{align}
Then $\zeta(\tup a_1;a,\tup a_0)$ holds since $\phi(x;\tup a_1)$ and $\phi(x;\tup a_0)$ are inconsistent by (4). Assume that there is some $\tup a' \in \str M^{\tup y}$ such that $\zeta(\tup a';a,\tup a_0)$ holds. Then \[\exists x.\phi(x;\tup a')\wedge \phi(x;\tup a_1)\] holds in $\str N'$, as witnessed by $x=a$. By property (1) and as $\tup a'$ is in $\str M$, this implies that \[\exists x.\phi(x;\tup a')\wedge \phi(x;\tup a_0)\] holds in $\str N'$, contradicting  $\zeta(\tup a';a,\tup a_0)$. Thus  $\zeta(\tup y;a,\tup a_0)$ has no solution in $\str M$. 
In particular, $\tup a_1\nind[\str M]\tup a_0a$, proving the claim.

\begin{claim}\label{cl:3}Suppose $\str M$ is $1$-dimensional and let $\tup a$ be a tuple in $\phi(\str N';\tup a_1)$ and $\tup b$  a tuple in $\psi(\str N';\tup b_0)$. Then 
    \[\tup a_1\tup a\ind[\str M]\tup a_{0}\tup b_0\tup b.\]    
\end{claim}

\emph{Proof}: We show the result by induction on the length of $\tup a$ and $\tup b$. The base case where $\tup a$ and $\tup b$ are empty is given by  property~(3).

Assume we know the result for $\tup a, \tup b$ and we want to add an element $b\in\psi(\str N';\tup b_0)$ to $\tup b$. By 1-dimensionality, one of the two cases holds: \[\tup a_1\tup a\,  b \ind[\str M]\tup a_0\tup b_0\tup b\qquad\text{or}\qquad \tup a_1\tup a   \ind[\str M]\tup a_0\tup b_0\tup b\,b.\] 
Note that property (5) implies $b\nind[\str M]\tup b_0$, excluding the first case,
so the second case must hold, as required.


Now assume we want to add $a\in\phi(\str N';\tup a_1)$ to $\tup a$. 
By  1-dimensionality, 
\[\tup a_1 \tup a\, a \ind[\str M]\tup a_0\tup b_0 \tup b \qquad\text{or}\qquad \tup a_1\tup a\ind[\str M]\tup a_0\tup b_0\tup b a,\]
but the second possibility is excluded by Claim~\ref{cl:2}, and the first one concludes the inductive step. This proves Claim~\ref{cl:3}.

\medskip
Since there are $\tup a\in \phi(\str N';\tup a_1)$ and $\tup b\in\psi(\str N';\tup b_0)$ with
$\tup a\nind[\str M]\tup b$ by (2) and Corollary~\ref{cor:counting fs}, Claim~\ref{cl:3} implies $\str M$ is not $1$-dimensional.

\medskip
This finishes the proof of the implication~\eqref{lit:1-dim}$\rightarrow$\eqref{lit:regular} in Theorem~\ref{thm:main2}, and completes the proof of the theorem. 
We now finish the proof of Theorem~\ref{thm:main3}.

\begin{proof}[Proof of Theorem~\ref{thm:main3}]It remains to prove 
	 the implication~\eqref{lit:1-dim'}$\rightarrow$\eqref{lit:grid-and} in Theorem~\ref{thm:main3}.
		
	 We just proved that if $\CC$ is not regular, then there are $\str M\prec \str N'$ in the elementary closure of $\CC$ and tuples $\tup a,\tup b\in \str N'$  and a single element $c\in\str N'$ such that:
	\[
		\tup a\ind[\str M]\tup b,\qquad
		\tup a c\nind[\str M]\tup b,\qquad
	    \tup a \nind[\str M]\tup b c,\]
	so $\str M$ is not $1$-dimensional. The formulas that exhibit $\tup a c\nind[\str M]\tup b$ and 
	$\tup a \nind[\str M]\tup b c$ 
	depend only on $\theta,\psi,\phi$ (and not on $\CC$), and
	are among the following:
	\begin{enumerate}
		\item the formula $\theta(\tup u;\tup b)\triangle \theta(\tup u;\tup b')$, cf. Remark~\ref{rem:triangle},
		\item the formula $\psi(x;\tup b_0)$, cf. property (5) in Lemma~\ref{lem:preparation},
		\item the formula $\zeta(\tup y; a,\tup e)\equiv 
		\phi(a;\tup y)\land \neg\exists x. \phi(x; \bar y)\land \phi(x; \tup e)
		$, cf.~\eqref{eq:zeta}.
	\end{enumerate}
In particular, if $\phi(a;\tup y)$ is the formula $y_1\le a\le y_2$, where $\le$ defines a total order in $\CC$, then $\zeta$ is equivalent to a quantifier-free formula, namely $(y_1\le a\le y_2)\land ((\tup y_1> e_2)\lor (\tup y_2< e_1))$.
 Hence, 
 in the setting of Remark~\ref{rem:qf-1},
 all the above formulas are quantifier-free $\Sigma$-formulas.
\todo{gap}
 This proves the implication~\eqref{lit:1-dim'}$\rightarrow$\eqref{lit:grid-and} in Theorem~\ref{thm:main3}, and completes its proof.
\end{proof}


\appendix

\section{Model theoretic preliminaries}\label{app:model-theory}

\subsection{Basic notions from model theory}\label{sec:basic}

\paragraph{Models and theories.} In model theory, structures are called \emph{models},
and we will therefore denote them $\str M,\str N$, etc.
They will typically be infinite.

A (first-order) \emph{theory} is a set $T$ of sentences over a fixed signature. A \emph{model of a theory} $T$ is a model $\str M$ (finite or not) which satisfies all the sentences in $T$, which is denoted $\str M\models T$. We say that $T$ \emph{has a model} if there is some model $\str M$ of $T$.

\emph{The theory} of a class of structures $\CC$ 
is the set $T$ of all sentences $\phi$ such that $\str S\models \phi$ for all $\str S\in \CC$. Trivially, every structure in $\CC$ is a model of $T$, but typically, $T$ has also other models. 
Those can be constructed using the compactness theorem:

\begin{theorem}[Compactness of first-order logic]\label{thm:compactness}
    Let $T$ be a theory such that  every finite subset $T'\subset T$ has a model. Then $T$ has a model.
\end{theorem}
For example, let $\CC$ be a class of structures over a signature $\Sigma$, and assume that 
$\CC$ contains structures of arbitrarily large finite size.
Then the models of the theory of $\CC$ also include infinite models of arbitrarily large cardinality.
To see this, consider the theory $T$ of $\CC$
and let $\Sigma'$ extend the signature of $\CC$ 
by an arbitrary set  $C$ of constant symbols. For $c,d\in C$,
let $\phi_{cd}$ be the $\Sigma'$-sentence $c\neq d$.
Then $T\cup\setof{\phi_{cd}}{c,d\in C,c\neq d}$ satisfies the assumption of the compactness theorem, so it has a model $\str M$, and this model has at least the cardinality of $C$.

\paragraph{Elementary extensions.}
Let $\str M,\str N$ be two models such that the domain of $\str M$ is contained in the domain of $\str N$.
Then $\str N$ is an \emph{elementary extension} of $\str M$,
written $\str M\prec \str N$,
if  for  every 
formula $\phi(\tup x)$ and tuple $\tup a\in\str M^{\tup x}$ of elements of $\str M$,
\[\str M\models\phi(\tup a)\text{\quad if and only if \quad}\str N\models \phi(\tup a).\] In other words, it doesn't matter if we evaluate formulas in $\str M$ or in $\str N$.

A typical way of constructing an elementary extension of $\str M$ 
is by considering the following theory, called the \emph{elementary diagram} of $\str M$. Let $\Sigma$ be the signature of $\str M$,
and let $\Sigma'=\Sigma\cup\str M$, where the elements of $\str M$ are viewed as constant symbols.

     For a $\Sigma$-formula $\phi(\tup y)$ and tuple $\tup a\in \str M^{\tup y}$ write $\phi(\tup a)$ for the $\Sigma'$-sentence 
     obtained by replacing the variables in $\tup y$ by constants in $\str M$, according to $\tup a$. 
     Let $T$ be the $\Sigma'$-theory consisting of all sentences $\phi(\tup a)$, for all $\Sigma$-formulas $\phi(\tup x)$ and tuples $\tup a$ such that $\str M\models\phi(\tup a)$.

Pick a model $\str N'$ of $T$, and 
let $\str N$ denote the $\Sigma$-structure obtained from $\str N'$ by forgetting the constants in $\str M\subset \Sigma'$.
    The interpretation of the constants $m\in\str M$ of $\Sigma'$ in $\str N'$ yields a function $i\from \str M\to \str N$.
    By the definition of $T$,
    for any formula $\phi(\tup y)$ and tuple  $\tup a\in\str M^{\tup y}$, 
    \[\str M\models\phi(\tup a)\text{ if and only if }\str N\models \phi(i(\tup a)).\]
    Therefore, we may view (identyfing each $m\in\str M$ with  $i(m)\in\str N$) the $\Sigma$-structure $\str N$ as an elementary extension of $\str M$.
    
    Reassuming, models of the elementary diagram of $\str M$ correspond precisely to elementary extensions of $\str M$. In particular, by extending the elementary diagram of $\str M$ by an arbitrary set of constants, from compactness we get that $\str M$ has elementary extensions of arbitrarily large cardinality (unless $\str M$ is finite). More generally, we have the following.

\begin{lemma}\label{lem:many-classes}
    Let $\str M$ be a model and let $\alpha(\tup x,\tup x')$ be a formula with $|\tup x|=|\tup x'|$ defining an equivalence relation in $\str M$ with infinitely many classes. Then for every cardinality $\mathfrak n$
    there is an elementary extension $\str N\succ\str M$
    in which 
    $\alpha$ defines an equivalence relation with at least $\mathfrak n$ equivalence classes. 
\end{lemma}
\begin{proof}
To simplify notation, assume that $|\tup x|=|\tup x'|=1$.
The  case of $|\tup x|=|\tup x'|=k>1$ proceeds similarly, or can be deduced from the case $k=1$ by extending the domain of $\str M$ by $\str M^k$ and the $k$ projection functions.

Let $\alpha(x,x')$ be  formula defining an equivalence relation ${\sim}$ in $\str M$ with infinitely many classes.
Let $\Sigma$ be the signature of $\str M$.
Fix any set of constants $C$ and let $\Sigma'$ extend  $\Sigma$  by   $C\cup\str M$,
where all the added elements are  constant symbols. For any $c,d\in C$ consider the $\Sigma'$-sentence $\phi_{cd}=\neg\alpha(c, d)$. 
Let $T$ be the $\Sigma'$-theory consisting of:
\begin{itemize}
    \item the sentences $\phi_{cd}$, for all $c\neq d$ in $C$,
    \item the elementary diagram of $\str M$.
\end{itemize} 
We show  that every $T'\subset T$ containing finitely many sentences of the form $\phi_{cd}$ has a model. Let $C'\subset C$ be the finite set of constants appearing in the sentences $\phi_{cd}\in T'$. Let $\str M'$  be the model $\str M$ together with 
each constant $c$ in $\str M\subset \Sigma'$ interpreted as the corresponding element  $c\in \str M$,
and
constants in $C'$ interpreted as pairwise $\sim$-inequivalent elements of $\str M$, and constants in $C\setminus C'$ interpreted as arbitrary elements of $\str M$. This can be done, since there are infinitely many pairwise $\sim$-inequivalent elements in $\str M$. This shows that $T'$ has a model.

By compactness, $T$ has a model $\str N'$. 
This model can be seen as an elementary extension of $T$ together with 
a set of $|C|$ elements which are pairwise inequivalent with respect to the equivalence relation defined by $\alpha$ in $\str N$. Since $C$ was taken arbitrary, this proves the lemma.
\end{proof}



\paragraph{Parameters.}
Let $\str M$ be a model over a signature $\Sigma$ and let $A\subset \str M$ be a set of elements. We may view $\str M$ 
as a model over a signature $\Sigma\cup A$, where the elements of $A$ are seen as constant symbols, interpreted in $\str M$ in the expected way: a constant $a\in A$ is interpreted as the element $a\in\str M$. We call the elements of $A$ \emph{parameters} in this context.
A $\Sigma$-formula \emph{with parameters} from $A$ is a formula over the signature $\Sigma\cup A$.

\paragraph{Types.}
A \emph{type} with variables $\tup x$ and parameters from $A$,
or a \emph{type over} $A$
is a set $p$ of formulas $\phi(\tup x)$ with parameters from $A$. We may write $p(\tup x)$ to indicate that $p$ has variables~$\tup x$.

If $p(\tup x)$ is a type over $A$ and $B\subset A$ then $p|B$ denotes the subset of $p$ consisting of all formulas with parameters from $B$.
If $\tup b\in\str M^{\tup x}$ is a tuple of elements of $\str M$ then \emph{the type}  of $\tup b$ \emph{over} $A$  in $\str M$ 
is the set of formulas $\phi(\tup x)$ with parameters from $A$ that are satisfied by $\tup b$ in $\str M$. 
This type is denoted $\tp(\tup b/A)$ or $\tp_{\tup x}(\tup b/A)$. Note that 
$\tp(\tup b/A)$ is related to the notion of $\theta$-types
as follows, for every formula $\theta(\tup x;\tup y)$ and tuple $\tup a\in A^{\tup y}$:
\[\theta(\tup x;\tup a)\in \tp(\tup b/A)
\iff \tup a\in \tp^\theta(\tup b/A).\]
In particular, $\tp(\tup b/A)$ is uniquely determined by $\setof{\tp^\theta(\tup b/A)}{\theta(\tup x;\tup y)\text{ is a formula}}$.

A type $p(\tup x)$ is \emph{satisfiable in a set} $C$ if there is some tuple $\tup c\in C^{\tup x}$ which satisfies all the formulas in $p$. 
A type $p(\tup x)$ with parameters from $A\subset \str M$ is \emph{satisfiable} if it is satisfiable 
in some elementary extension $\str N$ of $\str M$.
By compactness, this is equivalent to saying that for any finite conjunction $\phi(\tup x)$ of formulas in $p(\tup x)$ we have $\str M\models\exists \tup x.\phi(\tup x)$.

A type $p(\tup x)$ with parameters from $A$ is 
\emph{complete} if it is satisfiable and for every formula
$\phi(\tup x)$ with parameters from $A$, either $\phi$ or $\neg \phi$ belongs to $p$. Equivalently, $p(\tup x)$ is the type over $A$ of  some tuple $\tup b\in \str N^{\tup x}$, for some elementary extension $\str N$ of $\str M$.
We sometimes say that a type is \emph{partial} to emphasise that it may not be complete.
We denote the set of complete types with variables $\tup x$ and parameters from $A$ by $\St[\tup x](A)$ or simply $\St(A)$, if  $\tup x$ is understood from the context.
Note that we have ommitted the model $\str M$ from the notation. Indeed,
 if $\str M'$ is a model containing the parameters $A$
and satisfying the same sentences with parameters from $A$
as $\str M$, then $\str M$ and $\str M'$ have 
identical sets of complete types $p(\tup x)$ with parameters from $A$.
 Hence, $\St[\tup x](A)$ does not depend on $\str M$, but only on the set of sentences satisfied by $A$ in $\str M$.

\subsection{Finite satisfiability}
A (partial) type $p(\tup x)$ with parameters from $A$ is \emph{finitely satisfiable} in $C$ if every finite subset $p'\subset p$  is satisfiable in $C$. 
Note that $\tup a\ind[\str M]B$ (cf. Def.~\ref{def:independence}) if and only if $\tp(\tup a/\str MB)$ is finitely satisfiable in~$\str M$.

\begin{lemma}\label{lem:satisfaction}
    A type $p(\tup x)$ with parameters from $\str M$ 
    is finitely satisfiable in $\str M$ if and only if it is satisfiable. Consequently, $\tup a\ind[\str M]\str M$ for all $\tup a$ in an elementary extension of $\str M$.
\end{lemma}
\begin{proof}
    For the right-to-left implication, assume that $p$ is satisfied by some tuple $\tup c\in\str N^{\tup x}$ for some elementary extension $\str N$ of $\str M$. Pick a finite $p'\subset p$, and suppose $p'=\set{\phi_1,\ldots,\phi_k}$. 
    Consider the formula $\psi:=\phi_1\land\cdots\land \phi_k$.
    Note that $\psi$ may use some parameters from $\str M$.
    So we may write $\psi$ as $\psi=\psi'(\tup x,\tup a)$ 
    where $\psi'(\tup x,\tup z)$ is a formula and $\tup a\in \str M^{\str z}$. 
    
    The formula $\exists_{\tup x}\psi'(\tup x,\tup a)$ holds in $\str N$, as witnessed by $\tup c$. As $\str N$ is an elementary extension of $\str M$, this formula also holds in $\str M$.
    So there is some $\tup m\in \str M$ satisfying $\psi'(\tup b,\tup a)$. Therefore, $p'$ is satisfied by $\tup m$ in $\str M$, proving that $p$ is finitely satisfiable in $\str M$.

    \medskip
    The left-to-right implicaiton is a basic application of the compactness theorem. 
    
    Consider the signature $\Sigma'=\Sigma\cup \str M\cup\tup x$
    extending $\Sigma$ by constant symbols for each element of $\str M$ and each variable in $\tup x$.
    Let $T$ be the theory over $\Sigma'$ consisting of:
    \begin{itemize}
        \item For every formula $\phi(\tup x)\in p$, 
         the $\Sigma'$-sentence obtained from $\phi(\tup x)$ by viewing each parameter $a\in \str M$ as the constant $a\in\str M\subset \Sigma'$, and each variable $x\in\tup x$ as the constant  $x\in \tup x\subset \Sigma'$. 
        \item the elementary diagram of $\str M$.
    \end{itemize}

    Then every finite subset $T'$ of $T$ has a model.
    Indeed, let $p'$ be the set of formulas $\phi(\tup x)$ which occur (as $\Sigma'$-sentences) in $T'$.
    Since $p(\tup x)$ is finitely satisfiable in $\str M$,
    $p'(\tup x)$ is satisfied by some tuple $\tup m\in\str M^{\tup x}$. The pair $(\str M,\tup m)$ may be seen as a $\Sigma'$-structure, where a constant $m\in\str M$ is interpreted by the corresponding element of $\str M$,
    and a constant $x\in\tup x$ is interpreted as $\tup m(x)$.
    Then $(\str M,\tup m)$ is a model of $T'$.

    By compactness, $T$ has a model $\str N'$. 
    This model can be seen as an elementary extension $\str N$ of $\str M$ together with a tuple $\tup c\in\str N^{\tup x}$ of elements (obtained by the interpretation of the constants $\tup x$ in $\str N'$), such that 
$\str N\models \phi(\tup c)$ for every formula $\phi(\tup x)\in p$.
Hence, $\tup c$ satisfies $p(\tup x)$ in $\str N.$
\end{proof}

\paragraph{Finite satisfiability  and filters.}
Recall that a \emph{filter} on a set $U$ is a nonempty set $F\subset P(U)$ that is closed under taking supersets (if $A\subset B$ then $A\in F$ implies $B\in F$), under binary intersections, and does not contain the empty set. A filter is an \emph{ultrafilter} if for every $A\subset U$, either $A\in F$ or $U\setminus A\in F$. Every filter is contained in some ultrafilter, by the Kuratowski-Zorn lemma.


\medskip

Let $\str N$ be a model, $A\subset\str N$ be a set and $\tup x$ be a set of variables. Fix a filter $F$  on $A^{\tup x}$.
The \emph{average (partial) type}  over $\str N$ is the partial type denoted $\Av_{F}(\tup x)$ such that for every formula $\phi(\tup x)$ with parameters from $\str N$,
\[ \phi(\tup x)\in \Av_{F}(\tup x) \iff \{\tup a \in A^{\tup x} : \str N \models \phi(\tup a)\} \in F.\]

This is a consistent partial type: if say $\phi_1(\tup x),\ldots,\phi_n(\tup x) \in \pi(\tup x)$, then since any finitely many elements of $F$ have non-empty intersection, there is $\tup a\in A^{\tup x}$ which satisfies the conjunction $\phi_1(\tup x)\wedge \cdots \wedge \phi_n(\tup x)$. Hence this conjunction is consistent, indeed we have shown that $\Av_F(\tup x)$ is finitely satisfiable in $A$.

If $F$ is an ultrafilter on $A^{\tup x}$, then $\Av_F(\tup x)$ is a complete type: for every formula $\phi(\tup x)$, either $\phi(\tup x)\in \Av_F(\tup x)$ or $\neg \phi(\tup x)\in \Av_F(\tup x)$.

\begin{lemma}
Let $\pi(\tup x)$ be a partial type, then $\pi(\tup x)$ is finitely satisfiable in $A$  if and only if there is a filter $F$ on $A^{\tup x}$ such that $\pi(\tup x) \subseteq \Av_F(\tup x)$.
\end{lemma}
\begin{proof}
We have already observed that $\Av_F(\tup x)$ is finitely satisfiable in $A$. Conversely, assume that $\pi(\tup x)$ is finitely satisfiable in $A$, then define $F_0\subseteq  P(A^{\tup x})$ by: $F_0 = \{\phi(A) : \phi(\tup x)\in \pi(\tup x)\}$. The fact that $\pi(\tup x)$ is finitely satisfiable in $A$ implies that any finitely many elements of $F_0$ have non-empty intersection. Let $F$ be the filter generated by $F_0$. Then we have $\pi(\tup x)\subseteq \Av_F(\tup x)$.
\end{proof}

\begin{lemma}
Let $p(\tup x)\in \St(\str N)$ be a complete type, then $p(\tup x)$ is finitely satisfiable in $A$ if and only if there is an ultrafilter $F$ on $A^{\tup x}$ such that $p(\tup x) = \Av_F(\tup x)$.
\end{lemma}
\begin{proof}
We have already seen that if $F$ is an ultrafilter on $A^{\tup x}$, then $\Av_F(\tup x)$ is a complete type over $\str N$, which is finitely satisfiable in $A$. Conversely, if $p(\tup x)\in \St(\str N)$ is finitely satisfiable in $A$, then by the previous lemma, there is a filter $F_0$ on $A^{\tup x}$ such that $p(\tup x)\subseteq \Av_{F_0}(\tup x)$. Extend $F_0$ to an ultrafilter $F$ on $A^{\tup x}$. Then $p(\tup x)\subseteq \Av_{F_0}(\tup x)\subseteq \Av_{F}(\tup x)$. But since $p(\tup x)$ is a complete type, one cannot add any formulas to it without making it inconsistent. Since $\Av_F(\tup x)$ is consistent, we must have $p(\tup x)= \Av_F(\tup x)$.
\end{proof}

\begin{lemma}\label{lem:extending fs}
Let $\pi(\tup x)$ be a partial type finitely satisfiable in $A$. Then there is a complete type $p(\tup x)\in \St(\str N)$ finitely satisfiable in $A$ which extends $\pi(\tup x)$.
\end{lemma}
\begin{proof}
Let $F$ be a filter on $A^{\tup x}$ such that $\pi(\tup x)\subseteq \Av_F(\tup x)$. Let $F'$ be an ultrafilter extending $F$ and let $p(\tup x) = \Av_{F'}(\tup x)$. Then $p$ is finitely satisfiable in $A$ and extends $\pi$.
\end{proof}

\begin{lemma}\label{lem:fs invariant}
Let $p(\tup x)\in \St(\str N)$ be finitely satisfiable in $A$. Then $p$ is \emph{$A$-invariant}, that is: for any formula $\phi(\tup x;\tup y)$ and tuples $\tup b, \tup b'\in \str N^{\tup y}$, we have:
\[ \tp(\tup b/A)  = \tp(\tup b'/A) \Longrightarrow \phi(\tup x;\tup b)\in p \leftrightarrow \phi(\tup x;\tup b')\in p.\]
\end{lemma}
\begin{proof}
If $\tp(\tup b/A) = \tp(\tup b'/A)$, then the formula $\phi(\tup x;\tup b)\triangle \phi(\tup x;\tup b')$ has no solution in $A$. Since $p$ is finitely satisfiable in $A$ that formula cannot be in $p$. Hence as $p$ is a complete type, the formula $\phi(\tup x;\tup b)\leftrightarrow \phi(\tup x;\tup b')$ is in $p$ as required.
\end{proof}


\subsection{Indiscernible sequences}

\begin{definition}
Let $\str M$ be a structure and $A\subseteq \str M$. Let $I$ be a linear order. A sequence $(\tup a_i:i\in I)$ of tuples of $\str M$ is \emph{indiscernible over $A$} if for any $n<\omega$ and indices \[i_1 < \cdots < i_n \qquad i'_1 < \cdots < i'_n\] in $I$, we have \[ \tp(a_{i_1},\ldots, a_{i_n}/A) = \tp(a_{i'_1},\ldots,a_{i'_n}/A).\]
\end{definition}

Another way to state this is that the sequence $(\tup a_i:i\in I)$  is {indiscernible over $A$} if for any $n<\omega$, indices \[i_1 < \cdots < i_n \qquad i'_1 < \cdots < i'_n\] in $I$ and formula $\theta(\tup x_1,\ldots,\tup x_n)$ with parameters in $A$, we have  \[ (1) \qquad \str M\models \theta(\tup a_{i_1},\ldots, \tup a_{i_n}) \leftrightarrow \theta(\tup a_{i'_1},\ldots,\tup a_{i'_n}).\]

If $\Delta$ is a set of formulas with parameters in $A$, we will say that the sequence $(\tup a_i:i\in I)$ is \emph{$\Delta$-indiscernible} if $(1)$ holds for each $\theta$ in $\Delta$. If $\Delta$ and $I$ are both finite, then this is expressible by a single first order formula.

\begin{definition}
Two sequences $(\tup a_i:i\in I)$ and $(\tup b_j : j\in J)$ are \emph{mutually indiscernible} over $A$ if $(\tup a_i:i\in I)$ is indiscernible over $A\cup \{\tup b_j :j\in J\}$ and $(\tup b_j:j\in J)$ is indiscernible over $A\cup \{\tup a_i:i\in I\}$.
\end{definition}

An equivalent definition is that the sequences $(\tup a_i:i\in I)$  and $(\tup b_j : j\in J)$ are \emph{mutually indiscernible} over $A$  if for any $n<\omega$, indices \[i_1 < \cdots < i_n \qquad i'_1 < \cdots < i'_n\] and \[ j_1 < \cdots < j_n \qquad  j'_1 < \cdots < j'_n\]  in $I$ and any formula $\theta(\tup x_1,\ldots,\tup x_n;\tup y_1,\ldots,\tup y_n)$ with parameters in $A$, we have  \[ (2) \qquad \str M\models \theta(\tup a_{i_1},\ldots,\tup  a_{i_n};\tup b_{j_1},\ldots,\tup b_{j_n}) \leftrightarrow \theta(\tup a_{i'_1},\ldots,\tup a_{i'_n};\tup b_{j'_1},\ldots,\tup b_{j'_n}).\]

If $\Delta$ is a set of formulas with parameters in $A$, we will say that the sequences $(\tup a_i:i\in I)$ and $(\tup b_j:j<\omega)$ are  \emph{mutually $\Delta$-indiscernible} if $(2)$ holds for each $\theta$ in $\Delta$. If $\Delta$, $I$ and $J$ are finite, then this is again expressible by a single first-order formula.

In the following lemma, we use the notation $\Av_F | C$ to mean the restriction of the type $\Av_F$ to $C$. We also use the notation $\tup a_{<i}$ to mean $\{\tup a_j: j<i\}$.

\begin{lemma}
Let $A\subseteq B \subseteq \str M$. Let $F$ be an ultrafilter on $A^{\tup x}$. Let $I$ be a linear order and let $(\tup a_i:i\in I)$ be a sequence of tuples of $\str M$ such that:
\[\tup a_i \models \Av_F | B\tup a_{<i}.\]
Then the sequence $(\tup a_i:i\in I)$ is indiscernible over $B$.
\end{lemma}
\begin{proof}
Write $p=\Av_F$. Note that $p$ is finitely satisfiable in $A$ and a fortiori finitely satisfiable in $B$.

We prove by induction on $n$ that if $n<\omega$ and $i_1 < \cdots < i_n$, $j_1 < \cdots < j_n$ are in $I$, then $\tp(\tup a_{i_1},\ldots, \tup a_{i_n}/A) = \tp(\tup a_{j_1},\ldots,\tup a_{j_n}/A)$. For $n=1$ this follows from the fact that all $\tup a_i$ realize $\Av_F | B$, which is a complete type over $B$. Assume we know it for $n$ and take $i_1 < \cdots < i_n< i_{n+1}$, $j_1 < \cdots < j_n< j_{n+1}$ in $I$. By induction hypothesis, we have \[\tp(\tup a_{i_1},\ldots, \tup a_{i_n}/B) = \tp(\tup a_{j_1},\ldots,\tup a_{j_n}/B).\]
By Lemma \ref{lem:fs invariant}, for any formula $\theta(\tup x;\tup y_1,\ldots, \tup y_n)$ with parameters in $B$, we have:
\[ \theta(\tup x;\tup a_{i_1},\ldots,\tup a_{i_n}) \in p \iff \theta(\tup x;\tup a_{j_1},\ldots,\tup a_{j_n})\in p.\]
Now since $\tup a_{i_{n+1}} \models p| Ba_{i_1}\ldots a_{i_n}$, we have 
\[ \theta(\tup x;\tup a_{i_1},\ldots,\tup a_{i_n}) \in p \iff \str N \models \theta(\tup a_{i_{n+1}};\tup a_{i_1},\ldots,\tup a_{i_n}),\]
and similarly since $\tup a_{j_{n+1}}\models p| Ba_{j_1}\ldots a_{j_n}$, we have
\[ \theta(\tup x;\tup a_{j_1},\ldots,\tup a_{j_n}) \in p \iff \str N \models \theta(\tup a_{j_{n+1}};\tup a_{j_1},\ldots,\tup a_{j_n}).\]
Putting all of this together, we get
\[ \str N \models \theta(\tup a_{i_{n+1}};\tup a_{i_1},\ldots,\tup a_{i_n}) \iff \str N \models \theta(\tup a_{j_{n+1}};\tup a_{j_1},\ldots,\tup a_{j_n}).\]
Since the formula $\theta$ was an arbitrary formula with parameters in $B$, we deduce
\[ \tp(\tup a_{i_1},\ldots,\tup a_{i_{n+1}}/B) = \tp(\tup a_{j_1},\ldots,\tup a_{j_{n+1}}/B)\] as required.
\end{proof}

\begin{definition}
Let the type $p(\tup x)\in \St(\str M)$ be finitely satisfiable in $A\subset \str M$ and let $A\supseteq B\subset\str M$. A sequence $(\tup a_i :i\in I)$ of tuples in $\str M^{\tup x}$ such that $\tup a_i \models p|B\tup a_{<i}$ is called a \emph{Morley sequence} of $p$ over $B$.
\end{definition}

By the previous lemma, a Morley sequence of $p$ over $B$ is indiscernible over $B$.

\subsection{Building indiscernible sequences}

Indiscernible sequences are easy to find thanks to Ramsey's theorem.

\begin{definition}
Let $(\tup a_i:i<\omega)$ be a sequence of tuples in some structure $\str M$. A family $(\tup b_i:i\in I)$ indexed by a linear order $I$ is \emph{based on $(\tup a_i)_{i<\omega}$} if for any formula $\theta(x_1,\ldots,x_n)\in L$ and $i_1<\ldots <i_n$ in $I$, if $\str M\models \theta(b_{i_1},\ldots,\tup b_{i_n})$ then there are $j_1<\ldots <j_n<\omega$ such that $\str M\models \theta(\tup a_{j_1},\ldots,\tup a_{j_n})$.
\end{definition}

Note that if $(\tup a_i:i<\omega)$ is indiscernible and $(\tup b_i:i\in I)$ is based on it, then it is also indiscernible: indeed for any $i_1<\ldots <i_n$ in $I$ and any $j_1<\ldots <j_n<\omega$, we have \[\tp(\tup b_{i_1},\ldots,\tup b_{i_n}) = \tp(\tup a_{j_1},\ldots,\tup a_{j_n}).\]

\begin{proposition}\label{prop:ramsey}
Let $(\tup a_i:i<\omega)$ be a sequence of tuples in some structure $\str M$ and let $I$ be any linearly ordered set. There is an elementary extension $\str M\prec \str N$ and a sequence $(\tup b_i:i\in I)$ of tuples of $\str N$ that is {based on $(\tup a_i)_{i<\omega}$}.
\end{proposition}
\begin{proof}
Follows from Ramsey and compactness.
\end{proof}

We have analogues for two sequences.

\begin{definition}
Let $(\tup a_i:i<\omega)$ and $(\tup a'_i:i<\omega)$ be two sequences of tuples in $\str M$. Two families $(\tup b_i:i\in I)$, $(\tup b'_j:j\in J)$ indexed by linear orders $I$ and $J$ are \emph{based on $(\tup a_i)_{i<\omega}$ and $(\tup a'_i)_{i<\omega}$} if for any formula $\theta(\tup x_1,\ldots,\tup x_n;\tup y_1,\ldots,\tup y_m)\in L$ and $i_1<\ldots <i_n$ in $I$ and $j_1<\ldots<j_m$ in $J$, if $\str M\models \theta(\tup b_{i_1},\ldots,\tup b_{i_n};\tup b'_{j_1},\ldots,\tup b'_{j_n})$ then there are $k_1<\ldots <k_n<\omega$ and $k'_1<\ldots<k'_m<\omega$ such that $\str M\models \theta(\tup a_{k_1},\ldots,\tup a_{k_n};\tup a_{k'_1},\ldots,\tup a_{k'_m})$.
\end{definition}

Here is a finitary version of Proposition \ref{prop:ramsey} for two sequences.

\begin{lemma}\label{lem:ramsey 2}
Let $\Delta$ be a finite set of formulas and let $m,d<\omega$. Then there is some $m_* <\omega$ such that if $(\tup a_i:i<m_*)$ and $(\tup b_i:i<m_*)$ are two sequences of $d$-tuples of a structure $\str M$, then there are $(\tup a'_i:i<m)$ and $(\tup b'_i:i<m)$ subsequences of $(\tup a_i)_{i<m_*}$ and $(\tup b_i)_{i<m_*}$ respectively such that the sequences $(\tup a'_i:i<m)$ and $(\tup b'_i:i<m)$ are mutually $\Delta$-indiscernible.
\end{lemma}

\begin{proposition}\label{prop:ramsey 2}
Let $(\tup a_i:i<\omega)$ and $(\tup a'_i:i<\omega)$ be two sequences of tuples in $\str M$ and let $I,J$ be two linearly ordered sets. There is an elementary extension $\str M\prec \str N$ and sequences $(\tup b_i:i\in I)$ and $(\tup b'_j:j\in J)$ of tuples of $\str N$ which are {based on $(\tup a_i)_{i<\omega}$ and $(\tup a'_i)_{i<\omega}$}.
\end{proposition}
\begin{proof}
Follows from Lemma \ref{lem:ramsey 2} and compactness.
\end{proof}

\begin{lemma}
    Let $\str N$  be a model and $I=(\tup a_i:i<\omega)$ an indiscernible sequence of tuples of $\str N$. There is an elementary extension $\str N\prec \str N'$, a submodel $\str M\prec \str N'$ and an ultrafilter $F$ on $\str M^{\tup x}$ such that $I$ is a Morley sequence of $\Av_F$ over $\str M$.
\end{lemma}
\begin{proof}
	In an elementary extension of $\str N$, we can increase the sequence to $I + J$, where $J = (\tup b_i :i \in \mathbb Z)$ so that the sequence $I+J$ is indiscernible. Let $F_0$ be an ultrafilter on $\setof{\tup b_i}{i\in \mathbb Z}$ that contains all subsets of the form $\setof{\tup b_i}{i<n}$ for $n\in \mathbb Z$. It follows from indiscernibility that the sequence $I$ is a Morley sequence of $\Av_{F_0}$ over $\setof{\tup b_i}{i\in \mathbb Z}$. Possibly up to passing to a further elementary extension, we can find an elementary submodel $\str M$ such that $I$ is a Morley sequence of $\Av_{F_0}$ over $\str M$. One can see that by compactness, or alternatively, take $\str M_0$ any model containing $\setof{\tup a'_i}{i\in \mathbb Z}$, let $I' = (\tup a'_i:i<\omega)$ be a Morley sequence of $\Av_{F_0}$ over $M_0$. Now $I$ has the same type as $I'$ over $J$, so passing to an elementary extension, there is an automorphism $\sigma$ fixing $\setof{\tup b_i}{i\in \mathbb Z}$ pointwise and sending $I'$ to $I$. Then take $M = \sigma(M_0)$.
	
	Finally, define $F$ to be the unique ultrafilter on $M$ extending $F_0$ (so a set $A$ is in $F$ if and only if it contains a set in $F_0$). Then $I$ is a Morley sequence of $\Av_F$ over $M$.
\end{proof}

\begin{lemma}\label{lem:limit types}
Let $\str N$ be a structure and let $I=(\tup a_i:i<\omega)$ and $J= (\tup b_j:j<\omega)$ two mutually indiscernible sequences of tuples of $\str N$. There is an elementary extension $\str N\prec \str N'$, a submodel $\str M\prec \str N'$ two ultrafilters $F$ and $F'$ on $\str M^{\tup x}$ such that $I$ is a Morley sequence of $\Av_F$ over $\str M \cup \{\tup a_i:i<\omega\}$ and $J$ is a Morley sequence of $\Av_{F'}$ over  $\str M\cup \{\tup b_j:j<\omega\}$.
\end{lemma}
\begin{proof}
	The proof is very similar to the previous one. First, in an elementary extension, construct sequences $I' = (\tup a'_i:i\in \mathbb Z)$ and $J' = (\tup b'_j :j\in \mathbb Z)$ so that the two sequences $I+I'$ and $J+J'$ are mutually indiscernible. This is possible by compactness. Let $F$ be an ultrafilter on $\setof{\tup a'_i}{i\in \mathbb Z}$ containing all initial segments as in the previous proof and similarly for $F'$ on $\setof{\tup b'_j}{j\in \mathbb Z}$. Then $I$ is a Morley sequence of $\Av_F$ over $\setof{\tup a'_i}{i\in \mathbb Z} \cup \setof{\tup b_j}{j<\omega} \cup \setof{\tup b'_j}{j\in \mathbb Z}$ and $J$ is a Morley sequence of $\Av_{F'}$ over $\setof{\tup a_i}{i<\omega} \cup \setof{\tup a'_i}{i\in \mathbb Z} \cup \setof{\tup b'_j}{j\in \mathbb Z}$. One can then construct the model $M$ as above.
\end{proof}

\subsection{Proof of Proposition~\ref{prop:nip-2}}\label{app:nip}
\begin{definition}
    Say that a theory $T$ is \emph{NIP} if in every model $\str M\models T$, every formula $\phi(\tup x;\tup y)$ has finite VC-dimension on $\str M$.        
\end{definition}
By compactness, a class of structures $\CC$ is NIP 
according to Definition~\ref{def:NIP}
if and only if its theory is NIP.
The following is proved in~\cite{nip-dim-1}.

\begin{proposition}\label{prop:nip-2-variables}
    Fix a theory $T$.
    Suppose that for every $\str M\models T$ and  for every formula $\phi(x,y)$ with parameters from $\str M$,
$\phi(x,y)$ has finite VC-dimension on $\str M$.
Then $T$ is NIP.
\end{proposition}

We show how Proposition~\ref{prop:nip-2} follows.
We repeat its statement here.
\begin{proposition*}
	The following conditions are equivalent for a class of structures $\DD$:
	\begin{itemize}
		\item $\DD$ is not NIP,
		\item there is a formula $\phi(x,y;\tup z)$ 
		such that for every $n$ there is a structure $\str M\in\DD$ and a tuple $\tup c\in\str M^{\tup z}$ such that $\phi(\str M;\tup c)\subset \str M^{2}$ defines a binary relation of VC-dimension at least $n$;
	\end{itemize}
\end{proposition*}

\begin{proof}
    We show the top-down implication, the other being trivial.

    Assume $\DD$ is not NIP. Then the theory $T$ of $\DD$ is not NIP. By Proposition~\ref{prop:nip-2-variables}  there is a formula $\phi(x,y;\tup z)$, a model $\str M$ of $T$ and a tuple $\tup c\in \str M^{\tup z}$ of parameters from $\str M$ such that $\phi(\str M;\tup c)$ has infinite VC-dimension on $\str M$.
    
    Fix an arbitrary $n\in\N$. Consider the sentence $\psi_n$ expressing
    \begin{quote}
        ``there exists $\tup c$ such that $\phi(x,y;\tup c)$ defines a relation of VC-dimension at least $n$.''    
    \end{quote}
    Then $\str M$ satisfies $\psi_n$. In particular, the sentence $\neg\psi_n$ is not in $T$, so there is a model $\str S\in\DD$ satisfying $\psi_n$. This proves the top-down implication in Proposition~\ref{prop:nip-2}.    
\end{proof}

\section{Proof of Lemma~\ref{lem:preparation}}\label{app:preparation-lemma}
A family $(\phi_i(x))_{i\in I}$ 
of formulas with parameters from $\str N$ is \emph{pairwise inconsistent} if for any distinct $i,j\in I$, the formula $\phi_i(x)\land\phi_j(x)$ has no solution in $\str N$.
For a sequence $\setof{\tup a_i}{i<\omega}$ and for $i\le \omega$ 
by $\tup a_{<i}$ denote the set of elements in all the tuples $\tup a_j$ with $j<i$.

We prove a stronger variant of Lemma~\ref{lem:preparation}.
\begin{lemma*}
    Suppose $\CC$ is not regular, as witnessed by formulas 
     $\phi(x,\tup y),\psi(x,\tup z)$, $\theta(\tup u,\tup v)$. Then there exist:
    \begin{itemize}
        \item  a structure $\str M$ in the elementary closure of $\CC$, 
        \item an elementary extension $\str N$ of $\str M$,
        \item a sequence $(\tup a_i:i<\omega)$ of tuples in $\str N^{\tup y}$ and a sequence $(\tup b_j:j<\omega)$ of tuples in $\str N^{\tup z}$,
    \end{itemize}
    such that the following properties hold:
    \begin{enumerate}
        \item  the tuples $\tup a_0,\tup a_1,\ldots$ have equal types over $\str M$,
		and the tuples $\tup b_0,\tup b_1,\ldots$ have equal types over $\str M$,
        \item\label{it:many-types} for all $0<i,j<\omega$, 
        the set $\Types[\theta](A/B)$ is infinite, where $A=\phi(\str N,\tup a_i)$ and $B=\psi(\str N,\tup b_i)$,
        \item\label{it:ind-base} $\tup a_i \ind[\str M]\tup a_{<i}\tup b_{<\omega}$ for $i<\omega$,
        \item the formulas $\setof{\phi(x;\tup a_i)}{i<\omega}$ are pairwise inconsistent,
        \item the formulas $\setof{\psi(x;\tup b_j)}{j<\omega}$ are pairwise inconsistent.        
    \end{enumerate}
    \end{lemma*}
    It is clear that each of the properties (1)-(4)
    implies the corresponding property stated in Lemma~\ref{lem:preparation}.
    Properties (5) and (1) together imply that 
    $\psi(\str M,\tup b_j)=\emptyset$ 
    yielding property (5) in Lemma~\ref{lem:preparation}.
    We thus prove the statement above.

    \begin{proof}
Assume $\CC$ is not regular. We proceed in two steps. 

\medskip
\noindent
\underline{Step 1}. 
 There is a model $\str M$ in the elementary closure of $\CC$,
indiscernible sequences 
$(\tup a_i:i<\omega)$ in $\str M^{\tup y}$ and $(\tup b_j:j<\omega)$ in $\str M^{\tup z}$,
formulas $\phi(x,\tup y)$, $\psi(x,\tup z)$ and $\theta(\tup u,\tup v)$ such that:
\begin{itemize}
    \item the families $\setof{\phi(x,\tup a_i)}{i<\omega}$ and $\setof{\psi(x,\tup b_j)}{j<\omega}$ are both pairwise inconsistent
    \item for each $0\le i,j<\omega$ the set 
    $\Types[\theta](A/B)$ is infinite, where $A=\phi(\str M,\tup a_i)$ and $B=\psi(\str M,\tup b_j)$.    
\end{itemize}
\medskip

 As $\CC$ is not regular, for every natural number $m$, we can find a structure $\str M_m\in\CC$ sequences $(\tup a^m_i:i<m)$ and $(\tup b^m_j:j<m)$ of tuples of $\str M_m$ such that:
    
    $\bullet$ the two families $\{\phi(\tup x;\tup a^m_i):i<m\}$ and $\{\psi(x;\tup b^m_j):j<m\}$ are pairwise disjoint;
    
    $\bullet$ for every $i,j<m$, the set $\Types[\theta](A/B)$ has size at least $m$, where $A = \phi(\str M_m;\tup a^m_i)$ and $B = \psi(\str M_m;\tup b^m_j)$.
    
    Add constants to the signature to name two sequences $(\tup a_i:i<\omega)$ and $(\tup b_j:j<\omega)$. Consider the theory $T'$ in the extended language consisting of the following for every $m<\omega$:
    
    $\bullet_0$ all sentences which hold in all structures in $\CC$;
    
    $\bullet_{1,m}$ for every $i<j<m$, the two sets $\phi(\tup x;\tup a_i)$ and $\phi(\tup x;\tup a_j)$ are disjoint and the two sets $\psi(x;\tup b_i)$ and $\psi(x;\tup b_j)$ are disjoint;
    
    $\bullet_{2,m}$ the two sequences $(\tup a_i:i<m)$ and $(\tup b_j:j<m)$ are  indiscernible;
    
    $\bullet_{3,m}$ for every $i,j<m$ the set $\Types[\theta](A/B)$ has size at least $m$, where $A = \phi(\str M;\tup a^m)$ and $B = \psi(\str M;\tup b^m)$ and $\str M$ is the considered model.
    
    Note that all those conditions are expressible by first order formulas (infinitely many in the case of $\bullet_0$ and $\bullet_{2,m}$).
    
    We claim that $T'$ is consistent. Let $T_0\subseteq T'$ be finite. Then there is $m<\omega$ such that $T_0$ only contains formulas from $T$ along with formulas $\bullet_{1,m'}$, $\bullet_{2,m'}$ and $\bullet_{3,m'}$ for $m'\leq m$. Furthermore, there is a finite set $\Delta$ of formulas such that the formulas from $\bullet_2$ appearing in $T_0$ say at most that $(\tup a_i:i<m)$ and $(\tup b_j:j<m)$ are  $\Delta$-indiscernible.
    
    By Lemma \ref{lem:ramsey 2}, for $m_*<\omega$ is large enough, we can find a subsequences $(\tup a'_i:i<m)$ of $(\tup a^{m_*}_i:i<m)$ and a subsequence $(\tup b'_i:i<m)$ of $(\tup b^{m_*}_i:i<m_*)$ that are  $\Delta$-indiscernible. But then $M_{m_*}$ where we interpret the constants so as to name the two sequences $(\tup a'_i:i<m)$ of $(\tup a^{m_*}_i:i<m)$ is a model of $T_0$. Hence $T_0$ is consistent. As $T_0$ was an arbitrary finite subset of $T'$, we conclude by compactness that $T'$ is consistent.

    Let $\str M$ be a model of $T'$ and set $I =(\tup a_i:i<\omega)$ and $J=(\tup b_j:j<\omega)$ as interpreted in $\str M$. 
This yields the structure $\str M$ as described in Step 1.

%
    

\medskip\noindent
\underline{Step 2.} 
Apply Lemma~\ref{lem:limit types} to get an elementary extension $\str N$ of $\str M$, an elementary substructure $\str M'$ of $\str N$, such that 
$(\tup b_j:j<\omega)$ 
is a Morley sequence over  $\str M'\tup a_{<\omega}$ and 
$(\tup a_i:j<\omega)$
is a Morley sequence over $\str M'\tup b_{<\omega}$.
 In particular:
\begin{enumerate}
    \item $(\tup a_i:i>\omega)$ and 
     $(\tup b_j:j>\omega)$ are both indiscernible over $\str M'$,
     \item the families $\setof{\phi(x,\tup a_i)}{i<\omega}$ and $\setof{\psi(x,\tup b_j)}{j<\omega}$ are both pairwise inconsistent,
        \item for each $0\le i,j<\omega$ the set 
        $\Types[\theta](A/B)$ is infinite, where $A=\phi(\str N,\tup a_i)$ and $B=\psi(\str N,\tup b_j)$,          
    \item $\tup a_i\ind[\str M']\tup a_{<i}\tup b_{<\omega}$.
\end{enumerate}
This finishes the proof of Lemma~\ref{lem:preparation}.
\end{proof}

\section{Proof of Proposition~\ref{prop:fs dichotomy}}\label{app:fmNIP}
Before proving Proposition~\ref{prop:fs dichotomy},
we prove some lemmas.

The first lemma 
is a characterisation of theories that exclude large grids
using mutually indiscernible sequences.
\begin{lemma}\label{lem:two sequences}
    Assume that in some model $\str M$ there are two mutually indiscernible sequences $(\bar a_i:i<\omega)$ and $(\bar b_j:j<\omega)$ and a singleton $c$ and for some formulas $\phi(\bar x, \bar y;z)$ with parameters in $\str M$ such that:
    \begin{itemize}
    \item $\phi(\bar a_0,\bar b_0;c)$ holds,
    \item  $\neg  \phi(\bar a_i,\bar b_0;c)$ holds for all $i>0$,
     \item $\neg \phi(\bar a_0,\bar b_j,c)$ holds for all $j>0$.
    \end{itemize}
    Then the formula
    \[ \psi(\bar x\bar x',\bar y\bar y';z) \equiv \phi(\bar x,\bar y;z)\wedge \neg \phi(\bar x',\bar y;z)\wedge \neg \phi(\bar x,\bar y';z)\]
     defines large grids in $\str M$.
    \end{lemma}
    \begin{proof}
    By Proposition \ref{prop:ramsey 2}, there are two sequences $(\bar a'_i:i\in \mathbb Z)$ and $(\bar b'_j:j\in \mathbb Z)$ in some elementary extension of $\str M$ which are mutually indiscernible and based on $(\bar a_i:i<\omega)$ and $(\bar b_j:j<\omega)$. Since those latter sequences are already mutually indiscernible, we have $\tp((\bar a'_{i})_{i\in \mathbb N}(\bar b'_{i})_{i\in \mathbb N}) = \tp((\bar a_i)_{i<\omega} (\bar b_i)_{i<\omega})$. It follows that we can find some $\bar a_i$, $i\in \mathbb Z_-$, and $\bar b_i$, $i\in \mathbb Z_-$ so that
    \[\tp((\bar a'_{i})_{i\in \mathbb Z}(\bar b'_{i})_{i\in \mathbb Z}) = \tp((\bar a_i)_{i\in \mathbb Z} (\bar b_i)_{i\in \mathbb Z}).\]
    In particular, the two sequences $(\bar a_i:i\in \mathbb Z)$ and $(\bar b_i:i\in \mathbb Z)$ are mutually indiscernible. By removing some points with negative indices, we may assume that either for all $i<0$, $\phi(\bar a_i,\bar b_0;c)$ holds for for all $i<0$, $\neg \phi(\bar a_i,\bar b_0;c)$ holds and similarly for $\phi(\bar a_0,\bar b_j;c)$.
    
    By mutual indiscernibility of the two sequences, for any $(k,l)\in {\mathbb Z}^2$, we can find some $c_{k,l}$ in an elementary extension $\str N$ of $\str M$ such that:
    
    \begin{itemize}
    \item $\phi(\bar a_k,\bar b_l;c_{k,l})$
    \item $\neg \phi(\bar a_i,\bar b_l;c_{k,l})$ holds for $i>k$, and 
    \item $\phi(\bar a_k,\bar b_j;c_{k,l})$ holds for $j>l$.
    \end{itemize}
    
    Note that the points $c_{k,l}$ are necessarily pairwise distinct.
    Consider the formula $\psi$ as in the statement.
    Then $\psi(\bar x\bar x',\bar y\bar y';z)$ holds of a tuple $(\bar a_i\bar a_{i+1},\bar b_j\bar b_{j+1},c_{k,l})$ if and only if $(i,j)=(k,l)$. Hence, $\psi$ defines an infinite grid in $\str N$. As $\str N$ is an elementary extension of $\str M$, the formula $\psi$ also defines arbitrarily large grids in $\str M$.
    \end{proof}

    

\begin{lemma}\label{lem:extending fs'}
    Let $\str M\prec \str N$ be models. Let $p(\tup x)$ be a type with parameters in $\str N$, finitely satisfiable in $\str M$ and $\tup a\models p|\str M$. Let $\tup c\in \str N^{\tup y}$ be any tuple. Then there is $r(\tup x,\tup y)\in \St(\str N)$ a type finitely satisfiable in $\str M$ extending $p(\tup x)\cup \tp_{\tup x\tup y}(\tup a\tup c/\str M)$.
    \end{lemma}
    \begin{proof}
    By Lemma \ref{lem:extending fs}, it is enough to show that the partial type $\pi(\tup x;\tup y) = p(\tup x)\cup \tp_{\tup x\tup y}(\tup a\tup c/\str M)$ is finitely satisfiable in $\str M$. Take $\theta(\tup x)\in p(\tup x)$ and $\psi(\tup x;\tup y)\in \tp(\tup a\tup c/\str M)$ and we look for a realization of $\theta(\tup x)\wedge \psi(\tup x;\tup y)$ in $\str M$. Consider the formula $\theta'(\tup x) = \theta(\tup x)\wedge (\exists \tup y)\psi(\tup x;\tup y)$. This formula is in $p(\tup x)$, hence as $p$ is finitely satisfiable in $\str M$ it has a realization $\tup a_0$ in $\str M$. Since $\str M\models \theta'(\tup a_0)$ holds, there is $\tup c_0$ in $\str M$ such that $\str M\models \psi(\tup a_0;\tup c_0)$. Hence the pair $(\tup a_0,\tup c_0)$ is a realization of $\theta(\tup x)\wedge \psi(\tup x;\tup y)$ as required.
    \end{proof}


We are now ready to prove Proposition~\ref{prop:fs dichotomy}, which we reformulate below.
\begin{proposition*}
Let $\str M\prec \str N$ be models.
If there are tuples   $\tup a\in\str N^{\tup x}$, $\tup b\in\str N^{\tup y}$ be tuples and a single element $c\in\str N$ such that 
\[
\tup a\ind[\str M]\tup b,\quad \tup ac\nind[\str M]\tup b,\quad \tup a\nind[\str M]\tup bc,\]
then $\str M$ defines large grids.
More precisely, if $\alpha(\tup x,z;\tup  b)$ witnesses $\tup ac\nind[\str M]\tup b$ and $\beta(\tup x;\tup b,c)$ witnesses $\tup a\nind[\str M]\tup bc$ then a boolean combination $\psi$ of instances of  $\alpha$ and $\beta$
defines large grids in $\str M$.
\end{proposition*}
    \begin{proof}
    Let $p(\tup x)\in \St_{\tup x}(\str N)$ be a type finitely satisfiable in $\str M$ and extending $\tp(\tup a/\str M\tup b)$
    (cf. Lem.~\ref{lem:extending fs}).
    Let also $q(\tup y)\in\St_{\tup y}(\str N)$ be a type finitely satisfiable in $\str M$ extending $\tp(\tup b/\str M)$. 
    Now, let $p(\tup x)\otimes q(\tup y)\in \St_{\tup x\tup y}(\str N)$ be a complete type finitely satisfiable in $\str M$ extending $\tp(\tup a\tup b/\str M)$ (which is finitely satisfiable in $\str M$ by Lem.~\ref{lem:satisfaction}).
    Finally, let $r(\tup x,\tup y,z)$ be a type finitely satisfiable in $\str M$ extending $(p(\tup x)\otimes q(\tup y))\cup \tp(\tup a\tup b  c/\str M)$,
    obtained from Lemma~\ref{lem:extending fs'}.
    Reassuming, $r(\tup x,\tup y,z)$ is a type finitely satisfiable in $\str M$ extending $\tp(\tup a/\str M\tup b)\cup \tp(\tup a\tup bc/\str M)$,
    and $p(\tup x)$ and $q(\tup y)$ are its restrictions to $\tup x$ and $\tup y$, respectively.

    Let $J=(\tup b_j:j\in \mathbb Z_-)$ be a Morley sequence of $q(\tup y)$ over $\str M$, where $\mathbb Z_-$ denotes the negative integers. Let $\tup a'\tup b'c' \models r | \str MJ$. Then let $I=(\tup a_i:0<i<\omega)$ be a Morley sequence of $p$ over $J\tup a' \tup b' c'$.
    
    We then have that the sequences $J+(\tup b')$ and $(\tup a')+I$ are mutually indiscernible (the first one is a Morley sequence of $q$ over $\str M$ and the second one is a Morley sequence of $p$ over the first one). 
    
By assumption there are formulas $\alpha(\tup x,z;\tup y)$ and $\beta(\tup x;\tup y,z)$ with parameters from $\str M$ such that:
\begin{itemize}
    \item $\tup a\tup bc$ satisfies 
    $\alpha\land \beta$;
    \item $\alpha(\tup x, y;\tup b)$ is not safisfiable in $\str M$;
    \item $\beta(\tup x; \tup b,c)$ is not safisfiable in $\str M$.
\end{itemize}
The same holds true for $\tup a\tup bc$ replaced by $\tup a'\tup b'c$ as they have equal types over $\str M$.
Then:
\begin{itemize}
    \item $(\tup a',\tup b',c')$ satisfies $\alpha\land\beta$,
    \item $\neg\alpha(\tup a_i,\tup b',c)$ holds for all $0<i<\omega$,
    \item $\neg \beta(\tup a',\tup b_j,c')$ holds for all $-\omega<j<0$.
\end{itemize}
Now Lemma~\ref{lem:two sequences} applied to $\alpha\land\beta$ yields the conclusion.
    \end{proof}

\bibliographystyle{alpha}
\bibliography{ref}

\newcommand{\etalchar}[1]{$^{#1}$}
\begin{thebibliography}{BKTW20b}

\bibitem[AA14]{adler2014interpreting}
Hans Adler and Isolde Adler.
\newblock Interpreting nowhere dense graph classes as a classical notion of
  model theory.
\newblock {\em European Journal of Combinatorics}, 36:322--330, 2014.

\bibitem[ALS88]{10.1007/3-540-19488-6_105}
Stefan Arnborg, Jens Lagergren, and Detlef Seese.
\newblock Problems easy for tree-decomposable graphs extended abstract.
\newblock In Timo Lepist{\"o} and Arto Salomaa, editors, {\em Automata,
  Languages and Programming}, pages 38--51, Berlin, Heidelberg, 1988. Springer
  Berlin Heidelberg.

\bibitem[BBM06]{balogh2006hereditary}
J{\'o}zsef Balogh, B{\'e}la Bollob{\'a}s, and Robert Morris.
\newblock Hereditary properties of ordered graphs.
\newblock In {\em Topics in discrete mathematics}, pages 179--213. Springer,
  2006.

\bibitem[BGdMT21]{bonnet2021twinwidth}
{\'E}douard Bonnet, Ugo Giocanti, Patrice~Ossona de~Mendez, and St{\'e}phan
  Thomass{\'e}.
\newblock Twin-width iv: low complexity matrices, 2021.

\bibitem[BGK{\etalchar{+}}]{tww2}
{\'E}douard Bonnet, Colin Geniet, Eun~Jung Kim, St{\'e}phan Thomass{\'e}, and
  R{\'e}mi Watrigant.
\newblock {\em Twin-width II: small classes}, pages 1977--1996.

\bibitem[BGK{\etalchar{+}}20]{tww3}
{\'E}douard Bonnet, Colin Geniet, Eun~Jung Kim, St{\'e}phan Thomass{\'e}, and
  R{\'e}mi Watrigant.
\newblock Twin-width iii: Max independent set and coloring, 2020.

\bibitem[BKTW20a]{tww1}
{\'E}.~{Bonnet}, E.~J. {Kim}, S.~{Thomass{\'e}}, and R.~{Watrigant}.
\newblock Twin-width i: tractable fo model checking.
\newblock In {\em 2020 IEEE 61st Annual Symposium on Foundations of Computer
  Science (FOCS)}, pages 601--612, 2020.

\bibitem[BKTW20b]{bonnet2020twin}
{\'E}douard Bonnet, Eun~Jung Kim, St{\'e}phan Thomass{\'e}, and R{\'e}mi
  Watrigant.
\newblock Twin-width i: tractable fo model checking.
\newblock {\em arXiv preprint arXiv:2004.14789}, 2020.

\bibitem[BM15]{bova2015first}
Simone Bova and Barnaby Martin.
\newblock First-order queries on finite abelian groups.
\newblock In {\em 24th EACSL Annual Conference on Computer Science Logic (CSL
  2015)}. Schloss Dagstuhl-Leibniz-Zentrum fuer Informatik, 2015.

\bibitem[Bod15]{bodirsky_2015}
Manuel Bodirsky.
\newblock {\em Ramsey classes: examples and constructions}, pages 1--48.
\newblock London Mathematical Society Lecture Note Series. Cambridge University
  Press, 2015.

\bibitem[BS85]{DBLP:journals/ndjfl/BaldwinS85}
John~T. Baldwin and Saharon Shelah.
\newblock Second-order quantifiers and the complexity of theories.
\newblock {\em Notre Dame J. Formal Log.}, 26(3):229--303, 1985.

\bibitem[CiO07]{COURCELLE200791}
Bruno Courcelle and Sang il~Oum.
\newblock Vertex-minors, monadic second-order logic, and a conjecture by seese.
\newblock {\em Journal of Combinatorial Theory, Series B}, 97(1):91--126, 2007.

\bibitem[CK16]{DBLP:journals/corr/CibulkaK16}
Josef Cibulka and Jan Kyncl.
\newblock F{\"{u}}redi-hajnal limits are typically subexponential.
\newblock {\em CoRR}, abs/1607.07491, 2016.

\bibitem[Cou94]{COURCELLE199453}
Bruno Courcelle.
\newblock Monadic second-order definable graph transductions: a survey.
\newblock {\em Theoretical Computer Science}, 126(1):53--75, 1994.

\bibitem[EK17]{10.1007/978-3-662-55751-8_17}
Kord Eickmeyer and Ken-ichi Kawarabayashi.
\newblock Fo model checking on map graphs.
\newblock In Ralf Klasing and Marc Zeitoun, editors, {\em Fundamentals of
  Computation Theory}, pages 204--216, Berlin, Heidelberg, 2017. Springer
  Berlin Heidelberg.

\bibitem[FG06]{10.5555/1121738}
J.~Flum and M.~Grohe.
\newblock {\em Parameterized Complexity Theory (Texts in Theoretical Computer
  Science. An EATCS Series)}.
\newblock Springer-Verlag, Berlin, Heidelberg, 2006.

\bibitem[GHO{\etalchar{+}}20]{10.1145/3383206}
Jakub Gajarsk\'{y}, Petr Hlin\v{e}n\'{y}, Jan Obdr\v{z}\'{a}lek, Daniel
  Lokshtanov, and M.~S. Ramanujan.
\newblock A new perspective on fo model checking of dense graph classes.
\newblock {\em ACM Trans. Comput. Logic}, 21(4), July 2020.

\bibitem[GKS14]{grohe2014deciding}
Martin Grohe, Stephan Kreutzer, and Sebastian Siebertz.
\newblock Deciding first-order properties of nowhere dense graphs.
\newblock In {\em STOC 2014}, pages 89--98. ACM, 2014.

\bibitem[MT04]{MARCUS2004153}
Adam Marcus and G{\'a}bor Tardos.
\newblock Excluded permutation matrices and the stanley--wilf conjecture.
\newblock {\em Journal of Combinatorial Theory, Series A}, 107(1):153--160,
  2004.

\bibitem[NOdM11]{nevsetvril2011nowhere}
Jaroslav Ne{\v{s}}et{\v{r}}il and Patrice Ossona~de Mendez.
\newblock On nowhere dense graphs.
\newblock {\em European Journal of Combinatorics}, 32(4):600--617, 2011.

\bibitem[RS86]{ROBERTSON198692}
Neil Robertson and P.D Seymour.
\newblock Graph minors. v. excluding a planar graph.
\newblock {\em Journal of Combinatorial Theory, Series B}, 41(1):92--114, 1986.

\bibitem[She86]{shelah:hanfnumbers}
Saharon Shelah.
\newblock {\em Monadic logic: Hanf Numbers. In: Around Classification Theory of
  Models}, volume vol 1182 of {\em Lecture Notes in Mathematics}.
\newblock Springer, Berlin, Heidelberg, 1986.

\bibitem[Sim21]{nip-dim-1}
Pierre Simon.
\newblock A note on nip and stability in dimension one, January 2021.

\end{thebibliography}

\end{document}